\documentclass[letter,11pt]{article}

\usepackage{setspace}
\usepackage{subcaption}
\usepackage{enumerate}
\usepackage{amsmath}
\usepackage{amsfonts}
\usepackage{graphicx}
\usepackage{tcolorbox}
\usepackage{amssymb}
\usepackage{multirow}
\usepackage{geometry}
\usepackage{soul}
\usepackage{colortbl}
\usepackage{wrapfig}
\usepackage{algorithm}
\usepackage{algorithmicx}
\usepackage{algpseudocode}
\usepackage{amsthm}
\usepackage{todonotes}
\usepackage{mathtools}
\usepackage{mathrsfs}

\usepackage[pdftex, plainpages = false, pdfpagelabels, 
                 bookmarks=false,
                 bookmarksopen = true,
                 bookmarksnumbered = true,
                 breaklinks = true,
                 linktocpage,
                 pagebackref,
                 colorlinks = true,  
                 linkcolor = blue,
                 urlcolor  = blue,
                 citecolor = red,
                 anchorcolor = green,
                 hyperindex = true,
                 hyperfigures
                 ]{hyperref}

\usepackage{thmtools} 
\usepackage{thm-restate}

\newcommand{\dist}{{\sf dist}}

\algtext*{EndWhile}
\algtext*{EndIf}
\algtext*{EndFor}

\newtheorem{lemma}{Lemma}

\numberwithin{lemma}{section}
\newtheorem{theorem}[lemma]{Theorem}
\newtheorem{theorem*}[lemma]{Theorem*}
\newtheorem{corollary}[lemma]{Corollary}
\newtheorem{definition}[lemma]{Definition}
\newtheorem{observation}[lemma]{Observation}

\usepackage{wrapfig}
\usepackage{bm}

\title{\LARGE Efficient Approximation for Subgraph-Hitting Problems in Sparse Graphs and Geometric Intersection Graphs}

\author{Zden{\v{e}}k Dvo{\v{r}}{\'a}k\footnote{Charles University, Prague, Czech Republic. Email: \texttt{rakdver@iuuk.mff.cuni.cz}} \and Daniel Lokshtanov\footnote{University of California, Santa Barbara, USA. Email: \texttt{daniello@ucsb.edu}} \and Fahad Panolan\footnote{IIT Hyderabad, India. Email: \texttt{fahad@cse.iith.ac.in}} \and Saket Saurabh\footnote{Institute of Mathematical Sciences, Chennai, India. Email: \texttt{saket@imsc.res.in}} \and Jie Xue\footnote{NYU Shanghai, China. Email: \texttt{jiexue@nyu.edu}} \and Meirav Zehavi\footnote{Ben-Gurion University, Israel. Email: \texttt{meiravze@bgu.ac.il}}}

\date{}

\begin{document}


\maketitle

\begin{abstract}
We investigate a fundamental vertex-deletion problem called \textsc{(Induced) Subgraph Hitting}: given a graph $G$ and a set $\mathcal{F}$ of forbidden graphs, the goal is to compute a minimum-sized set $S$ of vertices of $G$ such that $G-S$ does not contain any graph in $\mathcal{F}$ as an (induced) subgraph. This is a generic problem that encompasses many well-known problems that were extensively studied on their own, particularly (but not only) from the perspectives of both approximation and parameterization.

In this paper, we study the approximability of the problem on a large variety of graph classes.
Our first result is a \textit{linear-time} $(1+\varepsilon)$-approximation reduction from \textsc{(Induced) Subgraph Hitting} on any graph class $\mathcal{G}$ of bounded expansion to the same problem on \textit{bounded degree} graphs within $\mathcal{G}$.
This directly yields linear-size $(1+\varepsilon)$-approximation lossy kernels for the problems on any bounded-expansion graph classes.
Our second result is a \textit{linear-time} approximation scheme for \textsc{(Induced) Subgraph Hitting} on any graph class $\mathcal{G}$ of polynomial expansion, based on the local-search framework of Har-Peled and Quanrud [SICOMP 2017].
This approximation scheme can be applied to a more general family of problems that aim to hit all subgraphs satisfying a certain property $\pi$ that is efficiently testable and has bounded diameter.
Both of our results have applications to \textsc{Subgraph Hitting} (not induced) on wide classes of geometric intersection graphs, resulting in linear-size lossy kernels and (near-)linear time approximation schemes for the problem.
\end{abstract}




\section{Introduction}

The generic definition of any vertex-deletion problem  is: given an (undirected) $n$-vertex graph $G$ (taken from some graph class), find a minimum-sized vertex subset $S \subseteq V(G)$ of $G$ such that $G-S$ satisfies some property, where $G-S$ is the graph obtained from $G$ by deleting all vertices in $S$. 
Often, the desired property is closed under taking (possibly induced) subgraphs. Whenever this is the case, the vertex-deletion problem can be formulated as hitting certain forbidden (induced) subgraphs of $G$ using fewest vertices.
This motivates the following generic vertex-deletion problem,  called \textsc{(Induced) Subgraph Hitting}.

\medskip
\noindent
\fbox{
    \parbox{0.97\textwidth}{
    \textsc{(Induced) Subgraph Hitting} \\
    \textbf{Input:} A graph $G$ and a set $\mathcal{F}$ of graphs called \textit{forbidden patterns}.\\
    \textbf{Goal:} Compute a subset $S \subseteq V(G)$ of minimum size such that $G-S$ does not contain any forbidden pattern $F \in \mathcal{F}$ as a (induced) subgraph.}
}
\medskip

In particular, by setting the forbidden set $\mathcal{F}$ to be various fixed (finite) sets, it is seen that \textsc{(Induced) Subgraph Hitting} generalizes many fundamental vertex-deletion problems, which were previously extensively studied on their own in the literature, particularly (but not only) from the perspectives of approximation and parameterization.
To name a few, these include \textsc{Vertex Cover}, \textsc{(Induced) $P_k$-Hitting} where $P_k$ is the path of length $k$~\cite{brevsar2013vertex,gupta2019losing,lee2017partitioning} (which subsumes {\sc Cluster Vertex Deletion}~\cite{aprile2021tight,fiorini2020improved,you2017approximate}), \textsc{Triangle Hitting} \cite{lokshtanov2022subexponential,lokshtanov2023framework} or more generally \textsc{(Induced) $C_k$-Hitting} where $C_k$ is the cycle of length $k$~\cite{groshaus2009cycle,pilipczuk2011problems}, \textsc{$K_k$-Hitting} where $K_k$ is the clique of size $k$~\cite{fiorini2018approximability}, \textsc{(Induced) Biclique Hitting}~\cite{goldmann2021parameterized}, \textsc{Component Order Connectivity}~\cite{drange2016computational,gross2013survey,kumar20172lk}, \textsc{Degree Modulator}~\cite{betzler2012bounded,ganian2021structural,gupta2019losing}, and \textsc{Treedepth Modulator}~\cite{bougeret2019much,eiben2022lossy,gajarsky2017kernelization}. We elaborate on related works, particularly on these specific problems,  in Appendix~\ref{sec-related}. 

Our main contribution is a \textit{linear-time} approximation-preserving reduction from \textsc{(Induced) Subgraph Hitting} on any graph class $\mathcal{G}$ of bounded expansion to the same problem on bounded degree graphs within $\mathcal{G}$. This yields a novel algorithmic technique to design (efficient) approximation schemes on very broad graph classes, well beyond the state-of-the-art, on which we elaborate below. First, we start with some relevant background for context and motivation.



\medskip
\noindent{\bf The (arguably grim) starting point of our research.} The seminal result of Lund and Yannakakis \cite{lund1993approximation} shows that \textsc{(Induced) Subgraph Hitting}, for any fixed nontrivial\footnote{Here ``nontrivial'' means that $\mathcal{F}$ does not contain the single-vertex graph.} $\mathcal{F}$, is APX-hard, while the problem admits a polynomial-time constant-approximation algorithm for any fixed $\mathcal{F}$.
As such, when the input $G$ is a general graph, one cannot expect the existence of polynomial-time approximation schemes (PTASes)---that is, polynomial-time $(1+\varepsilon)$-approximation algorithms---even with a fixed $\mathcal{F}$.
Now, there is a natural question to be asked: on which graph classes \textsc{(Induced) Subgraph Hitting} admits efficient approximation schemes?
The classical work of Baker \cite{Baker94} gives some answers to this question.
Specifically, Baker's approach yields an approximation scheme for \textsc{(Induced) Subgraph Hitting} on planar graphs with running time $f(\varepsilon,\mathcal{F}) \cdot n^{O(1)}$ for some function $f$, which can be extended to minor-free graphs \cite{dawar2006approximation,dvovrak2020baker}.
In addition, Baker's approach (together with some preprocessing) also results in an approximation scheme with the same running time for \textsc{Subgraph Hitting} (not induced) on unit-disk graphs \cite{FominLS18}.
However, beyond these results, little was known about the approximability of the problem (even for very special cases of it). 

Our purpose is to study \textsc{(Induced) Subgraph Hitting} (and obtain positive results) on substantially broader graph classes.
Specifically, we consider graph classes of \textit{bounded expansion} and other important graph classes 
related to them.
Graph classes of bounded expansion were introduced by Ne{\v{s}}et{\v{r}}il and De Mendez \cite{nevsetvril2008grad,nevsetvril2008grad2} as a general model of structurally sparse graph classes, generalizing many well-studied graph classes, such as graphs excluded a (topological) minor, graphs of bounded degree, graphs with bounded stack or queue number, many graph classes defined geometrically, etc.
There has been an extensive study on graphs of bounded expansion from both combinatorial perspective \cite{DorakN2016,nevsetvril2008grad,nevsetvril2011nowhere,nevsetvril2012characterisations,reidl2019characterising,zhu2009colouring} and algorithmic perspective \cite{dvovrak2013constant,dvovrak2021approximation,dvovrak2013testing,dvovrak2021approximation2,grohe2017deciding,har2017approximation,nevsetvril2008grad2}.
We refer the interested reader to the book of Ne{\v{s}}et{\v{r}}il and De Mendez \cite{nevsetvril2012sparsity} for a deeper introduction. Unfortunately, we cannot expect \textsc{(Induced) Subgraph Hitting} to admit efficient approximation schemes on a (general) graph class of bounded expansion, because the simplest special case, \textsc{Vertex Cover}, is already APX-hard even on bounded-degree graphs~\cite{dinur2005hardness}.

Particularly interesting subclass of bounded-expansion graphs are graph classes of \textit{polynomial expansion}.
Some well-studied examples include minor-free graphs, graphs drawn in the plane (or on a fixed surface) with a bounded number of crossings on each edge~\cite{nevsetvril2012characterisations}, $k$-nearest neighbor graphs of a point set in $\mathbb{R}^d$ for fixed $k$ and $d$~\cite{miller1997separators}, greedy Euclidean spanners~\cite{le2022greedy}, and many geometric intersection graphs without dense structures (discussed below), etc.
One feature that makes these graph classes especially important is their strong connection to graph separators.
It was shown~\cite{DorakN2016} that graph classes of polynomial expansion are exactly those having \textit{strongly sublinear separators}, which serve as a key ingredient needed in many efficient graph algorithms.

An important open question in the topic of graph sparsity is whether every optimization problem expressible in the first-order logic admits a PTAS on graph classes of polynomial expansion \cite{dvovrak2020baker,dvovrak2021approximation,dvovrak2021approximation2}.
For maximization problems, considerable progress has been made recently \cite{dvovrak2021approximation,dvovrak2021approximation2,galby2023polynomial,mezei2023ptas,romero2021treewidth}.
Notably, it was shown by Dvo{\v{r}}{\'a}k~\cite{dvovrak2021approximation} that every monotone maximization problem expressible in the first-order logic admits a QPTAS on polynomial-expansion graph classes (and admits a PTAS in many special cases).
However, the Baker-like approaches used in the literature \cite{dvovrak2021approximation,dvovrak2021approximation2,galby2023polynomial} do not work for minimization problems.
Therefore, the question for minimization problems in this setting is more challenging and their approximability is less understood.

\subsection{Our results}

Surprisingly, we show that the inapproximability of the problem on any graph class of bounded expansion comes \textit{exactly} from the bounded-degree graphs in that class! In other words, if the bounded-degree instances in the class can be approximated efficiently, so do all instances in the class.
Formally, we prove the following theorem, which is the first main contribution of this paper (here ``hereditary'' means closed under taking induced subgraphs).

\begin{restatable}{theorem}{mainthm} \label{thm-main}
Let $\mathcal{G}$ be any hereditary graph class of bounded expansion.
If \textsc{(Induced) Subgraph Hitting} on $\mathcal{G}$ admits an approximation scheme with running time $f_0(\varepsilon,\mathcal{F},\Delta) \cdot n^c$ for some function $f_0$ and some constant $c \geq 1$, then the same problem also admits an approximation scheme with running time $f(\varepsilon,\mathcal{F}) \cdot n^c$ for some function $f$, where $n$ is the number of vertices in the input graph $G \in \mathcal{G}$ and $\Delta$ is the maximum degree of $G$.
\end{restatable}

The reduction in Theorem~\ref{thm-main} consists of two steps.
Both steps are $(1+\varepsilon)$-approximation reductions and can be done in linear time.
The first step reduces an instance $(G,\mathcal{F})$ of \textsc{(Induced) Subgraph Hitting} on $\mathcal{G}$ to $r = O_\mathcal{F}(1)$ instances $(G,\mathcal{F}_1),\dots,(G,\mathcal{F}_r)$, where $\mathcal{F}_1,\dots,\mathcal{F}_r$ only depend on $\mathcal{F}$ and only consist of connected graphs.
Then the second step reduces each instance $(G,\mathcal{F}_i)$ to an instance $(G_i,\mathcal{F}_i)$ where $G_i$ is an induced subgraph of $G$ of degree $O_{\varepsilon,\mathcal{F}}(1)$.
The same reduction can also result in some variants of Theorem~\ref{thm-main}.
For example, it can be used to prove that an approximation scheme with running time $n^{f_0(\varepsilon,\mathcal{F},\Delta)}$ implies an approximation scheme with running time $n^{f(\varepsilon,\mathcal{F})}$, an approximation scheme with running time $f_0(\varepsilon,\mathcal{F},\Delta) \cdot n^{g(\varepsilon)}$ implies an approximation scheme with running time $f(\varepsilon,\mathcal{F}) \cdot n^{g(\varepsilon)}$, an approximation scheme with running time $f_0(\varepsilon,\mathcal{F},\Delta) \cdot n^{g(\mathcal{F})}$ implies an approximation scheme with running time $f(\varepsilon,\mathcal{F}) \cdot n^{g(\mathcal{F})}$, etc.
Although we shall only use Theorem~\ref{thm-main} in this paper (as it results in the best bounds), we believe that the other variants can also find their applications in the future.

A direct application of (the second step of) our reduction in Theorem~\ref{thm-main} is a $(1+\varepsilon)$-approximate lossy kernel of \textit{linear} size for the problem with a fixed $\mathcal{F}$ that only contain connected graphs. We now take a short detour and define the notion of lossy kernels.  Kernelization is a subfield of Parameterized Complexity, that provides a mathematical framework to analyze polynomial time preprocessing~\cite{CyganFKLMPPS15,fomin2019kernelization}.   Let $g : \mathbb{N} \rightarrow \mathbb{N}$ be a function. A {\em kernel of size $g(k)$} for a parameterized problem $\Pi$ is a polynomial time algorithm that takes as input an instance $(I,k)$ and outputs another instance $(I',k')$ such that $(I,k) \in \Pi$ if and only if $(I',k') \in \Pi$ and $|I'|+k' \leq g(k)$. If $g(k)$ is a linear, quadratic or polynomial function of $k$, we say that this is a linear, quadratic or polynomial kernel, respectively. Lokshtanov et al.~\cite{DBLP:conf/stoc/LokshtanovPRS17} defined the notion of lossy kernels that combines well with approximation algorithms. Informally speaking, an $(\alpha)$-lossy kernel of size $g(k)$ is a polynomial time algorithm that, given an instance $(I,k)$, outputs an instance $(I',k')$ such that $|I'|+ k' \leq g(k)$ and any $c$-approximate solution $s'$ to the instance $(I',k')$ can be turned in polynomial time into a $(c \cdot \alpha)$-approximate solution $s$ to the original instance $(I,k)$. Our result regarding lossy kernels is as follows.


\begin{restatable}{theorem}{kernel} \label{thm-kernel}
Let $\mathcal{G}$ be any graph class of bounded expansion.
For any fixed (finite) $\mathcal{F}$ consisting of connected graphs, \textsc{(Induced) Subgraph Hitting} on $\mathcal{G}$ with forbidden set $\mathcal{F}$ admits a $(1+\varepsilon)$-approximation lossy kernel of size $f(\varepsilon) \cdot k$ for some function $f$.
The kernelization algorithm runs in $g(\varepsilon) \cdot n$ time for some function $g$.
\end{restatable}

The linear-size lossy kernels in Theorem~\ref{thm-kernel} (together with the first step of our reduction in Theorem~\ref{thm-main}) already give approximation schemes with running time $f(\varepsilon,\mathcal{F}) \cdot n^{O(1)}$ for \textsc{(Induced) Subgraph Hitting} on any graph class of polynomial expansion.
Indeed, as long as a graph admits (strongly) sublinear separators, we can apply the standard Lipton-Tarjan approach to compute a subset $S \subseteq V(G)$ of size $\varepsilon n$ such that each connected component of $G - S$ has size $f(\varepsilon)$.
We then include $S$ in our solution together with an optimal solution for every component of $G - S$, which can be computed independently.
As the lossy kernels in Theorem~\ref{thm-kernel} are of linear size, we have $\mathsf{opt} = \Omega(n)$ for the kernel instance, and therefore the solution we compute is a $(1+O(\varepsilon))$-approximation of the optimal one.

A downside of the approximation schemes obtained from Theorem~\ref{thm-kernel} is that their running time is super-linear in $n$, since computing separators on a general graph class of polynomial expansion is time-consuming.
Furthermore, it seems difficult to generalize the results of Theorems~\ref{thm-main} and~\ref{thm-kernel} to other related first-order optimization problems on polynomial-expansion graph classes.
Therefore, in our second main result, we directly give general approximation schemes, based on the local-search framework of Har-Peled and Quanrud \cite{har2017approximation}, that can be implemented linear time and apply to a large family of first-order optimization problems on graph class of polynomial expansion, including \textsc{(Induced) Subgraph Hitting}. Note that the Theorem~\ref{thm-kernel} applies on all graphs classes of bounded expansion, however, the approximation schemes that we propose only apply to classes of polynomial expansion. As observed earlier, this restriction on graph classes is unavoidable  because the simplest special case of our problem, \textsc{Vertex Cover}, is already APX-hard even on bounded-degree graphs~\cite{dinur2005hardness}.



In order to present the most general version of our results, we first define a general family of hitting set problems on graphs that subsumes  \textsc{(Induced) Subgraph Hitting}. 
Let $\pi$ be a property of pairs $(G,Z)$, where $G$ is a graph and $Z\subseteq V(G)$ is non-empty, i.e., $\pi$ is a class of such pairs closed under isomorphisms. More precisely, given isomorphic pairs $(G_1,Z_1)$ and  $(G_2,Z_2)$ (that is, $G_1$ is isomorphic to $G_2$ by a function $f$ such that $f(Z_1)=Z_2$), either  both pairs belong to $\pi$ or both don’t. 
The \emph{diameter} of the property $\pi$ is the maximum diameter of $G[Z]$ over all pairs $(G,Z)\in \pi$.
A set $S\subseteq V(G)$ is a \emph{$\pi$-hitting set} if $S\cap Z\neq\emptyset$ for all sets $Z\subseteq V(G)$ such that $(G,Z)\in \pi$.
The $\pi$-\textsc{Hitting} problem asks us to find a minimum-size $\pi$-hitting set for a given graph $G$.
For example:
\begin{itemize}
\item For a finite set $\mathcal{F}$ of connected graphs, if $\pi$ is the property of pairs $(G,Z)$ such that $G[Z]$ contains a graph in $\mathcal{F}$ as a subgraph, and $Z$ is minimal (that is, for any $z\in Z$, $G[Z\setminus \{z\}]$ does not contain any graph in $\mathcal{F}$ as a subgraph), then  
we obtain the \textsc{Subgraph Hitting} problem with forbidden set $\mathcal{F}$, and the diameter of $\pi$ is
the maximum diameter of a graph in $\mathcal{F}$.
Similarly, if $\pi$ is the property that $G[Z]$ is isomorphic to a graph in $\mathcal{F}$,
we obtain \textsc{Induced Subgraph Hitting}. 
\item Let $r$ be a positive integer, and let $\pi$ be satisfied by $(G,Z)$ if and only if  there exists a vertex $z\in V(G)$ such that $Z$ consists
exactly of vertices at distance at most $r$ from $z$ in $G$.  Then $S$ is a $\pi$-hitting set if and only if it is an \emph{$r$-dominating set},
i.e., every vertex of $G$ is at distance at most $r$ from $S$.  The diameter of this property $\pi$ is $2r$.
\end{itemize}
Thus,  $\pi$-\textsc{Hitting} problem is a common generalization of both covering and packing problems. 
Although, we were not able to extend Theorem~\ref{thm-main} to the $\pi$-\textsc{Hitting} problem for properties $\pi$ that are of bounded diameter (which generalizes the problems  captured by Theorem~\ref{thm-main}), some of the ideas we developed in the proof apply in this setting also. These ideas allow us to show that the local-search approach of Har-Peled and Quanrud~\cite{har2017approximation} yields approximation schemes for the problem.
For a positive integer $c$, the \emph{$c$-local search} heuristic starts with $S=V(G)$ and repeats the following local improvement as long as possible:
If there exists a set $Y\subseteq V(G)$ of size at most $c$ such that $|S\triangle Y|<|S|$ and $S \triangle Y$ is a $\pi$-hitting set,
then replace $S$ by $S\triangle Y$.
While this na{\"i}ve algorithm of course fails to get a reasonable approximation in general,
as our second main result, we show that it succeeds for bounded diameter properties on any graph class of polynomial expansion.
In addition, if the property $\pi$ can be defined in first-order logic, we can implement the local-search procedure in linear time.
\begin{restatable}{theorem}{lsearchhitting}\label{thm-lsearch-hitting}
Let $\mathcal{G}$ be a graph class of polynomial expansion and $\pi$ be a property of finite diameter.
Then for every $\varepsilon>0$, there exists a constant $c > 0$ such that $c$-local search gives a $(1+\varepsilon)$-approximation solution for \textsc{$\pi$-Hitting} on $\mathcal{G}$.
Furthermore, if $\pi$ is first-order definable, the $c$-local-search algorithm can be implemented in $f(c) \cdot n$ time for some function $f$.
\end{restatable}

\noindent
When $\pi$ is not first-order definable but is efficiently testable in $\mathcal{G}$ (i.e., there exists a polynomial-time algorithm that decides, for a graph $G\in\mathcal{G}$ and a set $S\subseteq V(G)$, whether $S$ is a $\pi$-hitting set), we can implement $c$-local search in $n^{O(c)}$ time, which gives us a PTAS for \textsc{$\pi$-Hitting} on $\mathcal{G}$.

The property $\pi$ for \textsc{(Induced) Subgraph Hitting} (with a fixed forbidden set $\mathcal{F}$) is first-order definable.
Therefore, Theorem~\ref{thm-lsearch-hitting} directly gives us a \textit{linear-time} approximation scheme for the problem when $\mathcal{F}$ consists of connected graphs.
\begin{restatable}{corollary}{lsearchsub}\label{cor-lsearchsub}
\textsc{(Induced) Subgraph Hitting} on any graph class $\mathcal{G}$ of polynomial expansion admits an approximation scheme with running time $f(\varepsilon,\mathcal{F}) \cdot n$, where $n$ is the number of vertices in the input graph $G \in \mathcal{G}$.
\end{restatable}

Interestingly, the local-search approach also applies to the dual notion of packing. More precisely, 
a set $\mathcal{P}$ of disjoint subsets of vertices of $G$ is a \emph{$\pi$-packing}
if $(G,P)$ has the property $\pi$ for every $P\in\mathcal{P}$; it is an \emph{induced $\pi$-packing} if additionally no edge of $G$
has ends in distinct elements of $\mathcal{P}$.
The \textsc{$\pi$-Packing} problem asks us to find a maximum-size $\pi$-packing in a given graph $G$.
For example:
\begin{itemize}
\item For a graph $F$, if $\pi$ is the property that $G[P]$ is isomorphic to $F$,
we obtain the $F$-\textsc{matching} (and \textsc{induced} $F$-\textsc{matching}) problem
of finding the maximum number of disjoint (and non-adjacent) copies of $F$ in the input graph. When $F$ is a single vertex then the \textsc{induced} $F$-\textsc{matching} corresponds to the classical {\sc Independent Set} problem. 
\item Let $r$ be a positive integer, and let $\pi$ be satisfied by $(G,Z)$ if and only if there exists $z\in V(G)$ such that $Z$ consists
exactly of vertices at distance at most $r$ from $z$ in $G$.  Then a $\pi$-packing corresponds to a $2r$-independent set (a set
of vertices is \emph{$t$-independent} if the distance between any two of them is greater than $t$), and an induced $\pi$-packing
corresponds to a $(2r+1)$-independent set.
\end{itemize}
For the \textsc{(Induced) $\pi$-Packing} problem, we can run the following maximization variant of the $c$-local search algorithm:
Start with $\mathcal{P}=\emptyset$, and repeat the following local improvement step as long as possible:
If there exists $Y\subseteq \mathcal{P}$ of size less than $c$ such that
$\mathcal{P}\setminus Y$ can be extended to an (induced) $\pi$-packing $\mathcal{P}'$ of size greater than $|\mathcal{P}|$,
then replace $\mathcal{P}$ by $\mathcal{P}'$.
We can show that the local-search approach also results in approximation schemes for \textsc{(Induced) $\pi$-Packing} on graph classes of polynomial expansion, as long as $\pi$ is of finite diameter.
\begin{restatable}{theorem}{lsearchpacking}\label{thm-lsearch-packing}
Let $\mathcal{G}$ be a graph class of polynomial expansion and $\pi$ be a property of finite diameter.
Then for every $\varepsilon>0$, there exists a constant $c > 0$ such that $c$-local search gives $(1-\varepsilon)$-approximation solution for \textsc{(Induced) $\pi$-Packing} on $\mathcal{G}$,
Furthermore, if $\pi$ is first-order definable, the $c$-local-search algorithm can be implemented in $f(c) \cdot n$ time for some function $f$.
\end{restatable}

\noindent
Let us remark that the proof of Theorem~\ref{thm-lsearch-packing} is less involved than that of Theorem~\ref{thm-lsearch-hitting}, and is essentially implicit in~\cite{har2017approximation}.
So we view Theorem~\ref{thm-lsearch-packing} as a byproduct of the paper.

\subsection{Applications to geometric intersection graphs} \label{sec-app}
Geometric intersection graphs are graphs that depict the intersection patterns of geometric objects in a Euclidean space $\mathbb{R}^d$.
In a geometric intersection graph, each vertex corresponds to a geometric object in $\mathbb{R}^d$ and two vertices are connected by an edge if the two corresponding objects intersect.
In general, geometric intersection graphs are not sparse (as the geometric objects can stack in one place to form a large clique) and thus not of polynomial expansion.
However, most interesting classes of geometric intersection graphs are closely related to polynomial-expansion graphs in the sense that either the graph contains a large clique or it has strongly sublinear separators (i.e., belongs to a subclass of polynomial expansion).
Formally, we say a graph class $\mathcal{G}$ is of \textit{clique-dependent polynomial expansion} if any subclass of $\mathcal{G}$ with bounded clique size is of polynomial expansion.
Note that for \textsc{Subgraph Hitting} (not induced), large cliques are great for approximation, since we need to delete almost all vertices in the clique (so we might as well delete the entire clique).
We design a divide-and-conquer algorithm that can partition a graph into several $k$-cliques and a part without $k$-cliques (called a \textit{$k$-clique decomposition}) for any given integer $k \geq 1$.
Our algorithm runs in $f(k) \cdot n \log n + O(m)$ time when applying on a graph class of clique-dependent bounded expansion.
Using this algorithm, we can extend Theorem~\ref{thm-kernel} and Corollary~\ref{cor-lsearchsub} to any class of geometric intersection graphs that is clique-dependent polynomial expansion. Examples of such geometric intersection graphs include the following. 

\begin{itemize}
\item All classes of intersection graphs of any fat objects (i.e., convex objects whose diameter-width ratio is bounded). These classes are of clique-dependent polynomial expansion~\cite{de2020framework,dvovrak2021approximation,ErlebachJS05}, and include, for example, intersection graphs of balls, hypercubes, and similar convex objects.

\vspace{-0.5em}
\item The class of intersection graphs of pseudo-disks (i.e., topological disks in the plane satisfying that the boundaries of every pair of them are either disjoint or intersect twice).
We show that this class is also of clique-dependent polynomial expansion.
\end{itemize}
Both fat-object graphs and pseudo-disk graphs generalize the well-studied classes of disk graphs and unit-disk graphs.
Because these classes are of clique-dependent polynomial expansion, the aforementioned approximation scheme and lossy kernels on polynomial-expansion graph classes
 can also be applied to them.
Thus, applying Theorem~\ref{thm-kernel} and Corollary~\ref{cor-lsearchsub} yields the following. 

\begin{theorem}
\label{thm-geo}
    Let $\mathcal{G}$ be a class of fat-object graphs or the class of pseudo-disk graphs.
   \begin{enumerate}
   \setlength{\itemsep}{-2pt}
   \item  For any fixed (finite) $\mathcal{F}$ of connected graphs, \textsc{Subgraph Hitting} on $\mathcal{G}$ with forbidden list $\mathcal{F}$ admits a $(1+\varepsilon)$-approximation lossy kernel of size $f(\varepsilon) \cdot k$ for some function $f$.
    The kernelization algorithm runs in $g(\varepsilon) \cdot n \log n + O(m)$ time for some function $g$.
    \item \textsc{Subgraph Hitting} on $\mathcal{G}$ admits an approximation scheme with running time $f(\varepsilon,\mathcal{F}) \cdot n \log n + O(m)$ for some function $f$.
\end{enumerate}
\end{theorem}

Prior to this work, only special cases of the \textsc{Subgraph Hitting} problem have been considered on geometric intersection graphs beyond unit-disk graphs.
It was known that \textsc{Vertex Cover} admits PTASes on fat-object graphs \cite{ErlebachJS05,Leeuwen06,zhang2014minimum}.
Zhang et al. \cite{zhang2017ptas} showed that the \textsc{$P_k$-hitting} problem (also known as \textsc{Path Transversal}) admits a PTAS on unit-ball graphs.
Very recently, Lokshtanov et al.~\cite{lokshtanov2023framework} gave a general approach that leads to EPTASes for various cases of \textsc{Subgraph Hitting} on disk graphs when the forbidden list $\mathcal{F}$ contains a ``triangle-bundle'' and satisfies the so-called ``clone-closed'' property.
The approach of Lokshtanov et al.~relies heavily on these properties of the forbidden patterns as well as the plane geometry of disk graphs, and is thus unlikely to work for the general \textsc{Subgraph Hitting} problem or be generalized to geometric intersection graphs in higher dimensions (such as ball graphs).
In fact, Lokshtanov et al.~\cite{lokshtanov2023framework} asked explicitly as an open question whether one can obtain EPTASes for hitting any fixed forbidden pattern in disk graphs or more general classes of geometric intersection graphs.
Our result affirmatively answers this question, and improves the running time of their EPTASes (for \textsc{Subgraph Hitting} related problems) to near-linear.

\medskip
\noindent{\bf\em Application to string graphs.}
String graphs form the widest class of geometric intersection graphs in $\mathbb{R}^2$, which are defined by arbitrary connected geometric objects in $\mathbb{R}^2$.
This includes in particular intersection graphs of rectangles and segments (which are not covered by the aforementioned classes of geometric intersection graphs).
Unfortunately, string graphs are not of clique-dependent polynomial expansion.
But they are of \textit{biclique-dependent} polynomial expansion: any subclass with bounded biclique size is of polynomial expansion~\cite{fox2010separator}.
The approximation scheme on polynomial-expansion graph classes with running time $f(\varepsilon,\mathcal{F}) \cdot n$ also works on string graphs.
Exploiting this property, we can use Theorem~\ref{thm-kernel} and Corollary~\ref{cor-lsearchsub} to obtain $(2+\varepsilon)$-approximation lossy kernels and algorithms for \textsc{Bipartite Subgraph Hitting}, which is the special case of \textsc{Subgraph Hitting} where $\mathcal{F}$ contains at least one bipartite graph.

\begin{restatable}{theorem}{corstring} \label{thm-string}
    Let $\mathcal{G}$ be the class of string graphs.
       \begin{enumerate}
   \setlength{\itemsep}{-2pt}
    
   \item For any fixed (finite) set $\mathcal{F}$ of connected graphs among which at least one is bipartite, \textsc{Subgraph Hitting} on $\mathcal{G}$ with forbidden list $\mathcal{F}$ admits a $(2+\varepsilon)$-approximation lossy kernel of size $f(\varepsilon) \cdot k$ for some function $f$.
    The kernelization algorithm runs in $g(\varepsilon) \cdot n^2$ time for some function $g$.
    
   \item \textsc{Bipartite Subgraph Hitting} on $\mathcal{G}$ admits a $(2+\varepsilon)$-approximation algorithm with running time $f(\varepsilon,\mathcal{F}) \cdot n^2$ for some function $f$.

\end{enumerate}
\end{restatable}

\textsc{Bipartite Subgraph Hitting}, though more restrictive compared to the general \textsc{Subgraph Hitting} problem, already covers many interesting vertex-deletion problems, such as \textsc{Vertex Cover}, \textsc{$P_k$-Hitting}, \textsc{Component Order Connectivity}, \textsc{Degree Modulator}, and more.
While $(2+\varepsilon)$-approximation algorithms are less satisfying than approximation schemes, obtaining approximation schemes on string graphs is very difficult.
Indeed, we are not aware of any known approximation schemes for NP-hard problems on string graphs (in fact, even on the subclass of rectangle graphs or segment graphs).
Even worse, prior to this work, the only problems whose approximation ratio on string graphs is better than that on general graphs are \textsc{Independent Set} and \textsc{Chromatic Number} \cite{fox2011computing}, to the best of our knowledge.
Our result in Theorem~\ref{thm-string} adds many new problems to this small club.

\subsection{Subgraph isomorphism} \label{sec-addition}
\textsc{Subgraph Isomorphism} is a special case of \textsc{Subgraph Hitting}, which aims to find a subgraph of the input graph $G$ that is isomorphic to a given graph $H$.
When the number of vertices in $H$ is bounded by $k$, the problem is also called \textsc{$k$-Subgraph Isomorphism}.
Our aforementioned efficient $k$-clique decomposition algorithm straightforwardly gives us a byproduct, which is a \textsc{$k$-Subgraph Isomorphism} algorithm on any graph classes of clique-dependent polynomial expansion (and thus on the classes of fat-object graphs and pseudo-disk graphs) with running time $f(k) \cdot n \log n + O(m)$ (see Theorem~\ref{thm-kiso}). 
The \textsc{$k$-Subgraph Isomorphism} problem on fat-object graphs was considered very recently by Chan \cite{chan2023finding}, who gave an algorithm with running time $f(k) \cdot n \log n$ time, which requires the geometric realization of the graph (i.e., the geometric objects) to be given.
Contrary to this, our algorithm is \textit{robust} in the sense that it does not rely on the realization.
Note that if we are only given the graph, the additional $O(m)$ term in our running time is necessary.

\begin{theorem}
Let $\mathcal{G}$ be a class of fat-object graphs or the class of pseudo-disk graphs.
\textsc{$k$-Subgraph Isomorphism} on $\mathcal{G}$ can be solved in $f(k) \cdot n \log n + O(m)$ time for some function $f$, where $n$ and $m$ are the numbers of vertices and edges in the input graph $G \in \mathcal{G}$, respectively.
\end{theorem}

\section{Our techniques} \label{sec-overview}

\subsection{Overview of the proof of Theorem~\ref{thm-main}} \label{sec-overview1}

In this section, we give an informal overview of the proof of our main theorem.
Let $\mathcal{G}$ be a hereditary graph class of bounded expansion.
For convenience, we say a graph $G$ is (\textit{induced}) \textit{$F$-free} (for some graph $F$) if it does not contain $F$ as a (induced) subgraph.
Similarly, we say $G$ is (\textit{induced}) \textit{$\mathcal{F}$-free} (for some set $\mathcal{F}$ of graphs) if (\textit{induced}) \textit{$F$-free} for all $F \in \mathcal{F}$.

Our first step of proving Theorem~\ref{thm-main} is to reduce to the case in which all forbidden patterns are connected.
Specifically, we show that if we have an approximation scheme $\mathbf{A}$ for \textsc{(Induced) Subgraph Hitting} on $\mathcal{G}$ which works only when $\mathcal{F}$ consists of connected graphs, then we have an approximation scheme $\mathbf{A}'$ for the same problem which works for any $\mathcal{F}$ (with almost the same running time).
We now sketch a proof for this step.
To design the approximation scheme $\mathbf{A}'$ (provided the approximation scheme $\mathbf{A}$), consider a \textsc{(Induced) Subgraph Hitting} instance $(G,\mathcal{F})$.
For simplicity, let us assume $\mathcal{F} = \{F\}$ where $F$ consists of only two connected components $C$ and $D$.
Our ideas for handling this special case can directly be generalized to any $\mathcal{F}$.
Set $\gamma = |V(F)|$.
Let $S \subseteq V(G)$, $S_C \subseteq V(G)$, $S_D \subseteq V(G)$ be optimal solutions of the \textsc{(Induced) Subgraph Hitting} instances $(G,\mathcal{F})$, $(G,\{C\})$, $(G,\{D\})$, respectively.
The key observation here is that $\min\{|S_C|,|S_D|\} - |S| \leq \alpha(F)$ for some (computable) function $\alpha$.
Before showing this observation, we first explain why it is useful.
Observe that a feasible solution of $(G,\{C\})$ or $(G,\{D\})$ is also a feasible solution of $(G,\mathcal{F})$, since if a graph is (induced) $C$-free or (induced) $D$-free, then it must be (induced) $F$-free.
We can apply $\textbf{A}$ to compute $(1+\varepsilon)$-approximation solutions for $(G,\{C\})$ and $(G,\{D\})$, as $C$ and $D$ are connected.
According to the inequality $\min\{|S_C|,|S_D|\} - |S| \leq \alpha(F)$, we know that one of these two solutions is of size at most $(1+\varepsilon) \cdot (|S| + \alpha(F))$.
If $|S| = \Omega(\alpha(F)/\varepsilon)$, this gives us a $(1+O(\varepsilon))$-approximation solution of $(G,\mathcal{F})$.
Otherwise, $S$ is small (i.e. $|S|$ is bounded by a function of $\varepsilon$ and $\mathcal{F}$) and we can apply the known FPT algorithm for \textsc{(Induced) Subgraph Hitting} on bounded-expansion graphs to compute an optimal solution in $f(\mathcal{F},|S|) \cdot n$ time.

To show $\min\{|S_C|,|S_D|\} - |S| \leq \alpha(F)$, we first consider \textsc{Subgraph Hitting} (which is easier).
Recall $\gamma = |F|$.
The key observation is that $G-S$ cannot simultaneously contain many disjoint subgraphs isomorphic to $C$ and many disjoint subgraphs isomorphic to $D$ (here ``disjoint'' means \textit{vertex-disjoint}).
To see this, assume $G-S$ contains $\gamma$ disjoint subgraphs $C_1,\dots,C_\gamma$ isomorphic to $C$ and $\gamma$ disjoint subgraphs $D_1,\dots,D_\gamma$ isomorphic to $D$.
Since $C_1$ has at most $|V(C)| \leq \gamma - 1$ vertices, there exists $i \in [\gamma]$ such that $C_1$ and $D_i$ are disjoint.
Now $C_1$ and $D_i$ together form a subgraph of $G-S$ isomorphic to $F$.
But $G-S$ is $F$-free, as $S$ is a solution of $(G,\mathcal{F})$.
Thus, the above observation holds.
Now we can assume without loss of generality that $G-S$ does not contain $\gamma$ disjoint subgraphs isomorphic to $C$.
That means we can hit all subgraphs of $G-S$ isomorphic to $C$ using $\gamma^2$ vertices of $G-S$: simply take a maximal set $\mathcal{C}$ of disjoint subgraphs of $G-S$ isomorphic to $C$ (whose size is smaller than $\gamma$), then $\bigcup_{C \in \mathcal{C}} V(C)$ hits all subgraphs of $G-S$ isomorphic to $C$ (by the maximality of $\mathcal{C}$) and $|\bigcup_{C \in \mathcal{C}} V(C)| \leq \gamma^2$.
It follows that $|S_C| \leq |S| + \gamma^2$.
Setting $\alpha(F) = \gamma^2$, we have $\min\{|S_C|,|S_D|\} - |S| \leq \alpha(F)$.
Note that in this argument, we did not use the bounded expansion of $\mathcal{G}$ and the function $\alpha$ is also independent of $\mathcal{G}$.

For \textsc{Induced Subgraph Hitting}, the key observation is the same: $G-S$ cannot simultaneously contain many disjoint induced subgraphs isomorphic to $C$ and many disjoint induced subgraphs isomorphic to $D$.
But the proof is slightly more involved, and relies on the fact that $\mathcal{G}$ is of bounded expansion.
The bounded expansion of $\mathcal{G}$ actually implies that the average degree of every subgraph of $G$ is bounded by some constant $d$.
Assume $G-S$ contains $\gamma' = (2d+1) \cdot \gamma$ disjoint induced subgraphs $C_1,\dots,C_{\gamma'}$ isomorphic to $C$ and $\gamma'$ disjoint induced subgraphs $D_1,\dots,D_{\gamma'}$ isomorphic to $D$.
Let $V = \bigcup_{i=1}^{\gamma'} (V(C_i) \cup V(D_i))$.
We have $|V| \leq 2 \gamma \gamma'$.
The average degree of $G[V]$ is at most $d$.
So a simple averaging argument shows that there exists $i \in [\gamma']$ such that the $\sum_{v \in V(C_i)} \mathsf{deg}_{G[V]}(v) \leq d|V|/\gamma' \leq 2d \gamma$.
As $|V(C_i)| = |V(C)| \leq \gamma - 1$ and $\sum_{v \in V(C_i)} \mathsf{deg}_{G[V]}(v) \leq 2d \gamma$, there are at most $\gamma - 1$ indices $j \in [\gamma']$ such that $C_i$ intersects $D_j$ and at most $2d \gamma$ indices $j \in [\gamma']$ such that $C_i$ is neighboring to $D_j$.
So there exists $j \in [\gamma']$ such that $C_i$ and $D_j$ are disjoint and non-adjacent, which implies that $C_i$ and $D_j$ together form an induced subgraph of $G-S$ isomorphic to $F$.
Thus, the key observation holds and the same argument as before shows the inequality $\min\{|S_C|,|S_D|\} - |S| \leq \alpha(F)$.
Here the function $\alpha$ may depend on the graph class $\mathcal{G}$.

With the above argument, we can now restrict ourselves to connected forbidden patterns.
Next, we discuss the main part of our proof, how to solve the problem with connected forbidden patterns.
Before we start, we need to briefly discuss an important invariant of graphs, called \textit{weak coloring number}.
Consider a graph $G$ and an ordering $\sigma$ of the vertices of $G$.
For $u,v \in V(G)$, we write $u <_\sigma v$ if $u$ is before $v$ under the ordering $\sigma$.
Define $>_\sigma$, $\leq_\sigma$, $\geq_\sigma$ similarly.
For an integer $r \geq 0$, we say $u$ is \textit{$r$-weakly reachable} from $v$ \textit{under $\sigma$} if there is a path $\pi$ between $v$ and $u$ of length at most $r$ such that $u \geq_\sigma w$ for any vertex $w$ on $\pi$ (in particular, $u \geq_\sigma v$).
Let $\text{WR}_r(G,\sigma,v)$ denote the set of vertices in $G$ that are $r$-weakly reachable from $v$ under $\sigma$.
The \textit{weak $r$-coloring number} of $G$ \textit{under $\sigma$} is $\mathsf{wcol}_r(G,\sigma) = \max_{v \in V(G)} |\text{WR}_r(G,\sigma,v)|$.
Then the \textit{weak $r$-coloring number} of $G$ is defined as $\mathsf{wcol}_r(G) = \min_{\sigma \in \Pi(G)} \mathsf{wcol}_r(G,\sigma)$ where $\Pi(G)$ is the set of all orderings of the vertices of $G$.
An important characterization for bounded-expansion graphs \cite{nevsetvril2008grad,zhu2009colouring} is that a graph class $\mathcal{G}$ is of bounded expansion iff there is a function $\chi$ such that for every $r \geq 0$, the weak $r$-coloring numbers of the graphs in $\mathcal{G}$ has an upper bound $\chi(r)$, i.e., $\mathsf{wcol}_r(G) \leq \chi(r)$ for all $G \in \mathcal{G}$.
In our proof, we shall repeatedly exploit the bounded weak coloring number of our graph class $\mathcal{G}$.

Consider a \textsc{(Induced) Subgraph Hitting} instance $(G,\mathcal{F})$ where $G \in \mathcal{G}$ and $\mathcal{F}$ consists of connected forbidden patterns.
We shall show how to reduce $(G,\mathcal{F})$ to another instance $(G',\mathcal{F})$ in $f(\varepsilon,\mathcal{F}) \cdot n$ time, where $G'$ is an induced subgraph of $G$ of degree $O_{\varepsilon,\mathcal{F}}(1)$.
Here the reduction is in the sense that one can obtain a $(1+O(\varepsilon))$-approximation solution $S$ of $(G,\mathcal{F})$ efficiently from a $(1+\varepsilon)$-approximation solution $S'$ of $(G',\mathcal{F})$.
This solves the problem, because we can compute $S'$ by applying the given approximation scheme with running time $f_0(\varepsilon,\mathcal{F},\Delta) \cdot n^c$ on $G'$ (note that $G' \in \mathcal{G}$ as $\mathcal{G}$ is hereditary), which takes $f(\varepsilon,\mathcal{F}) \cdot n^c$ time due to the bounded degree of $G'$.
For convenience, we define some notations.
Let $\gamma = \max_{F \in \mathcal{F}} |V(F)|$ and $\mathcal{V} = \{V(H): H \text{ is a forbidden (induced) subgraph of } G\}$ where a forbidden (induced) subgraph refers to a (induced) subgraph isomorphic to some $F \in \mathcal{F}$.
A solution of $(G,\mathcal{F})$ is just a hitting set of $\mathcal{V}$.
In our reduction, we shall only use two properties of $\mathcal{V}$: 
\begin{itemize}
    \item[\textbf{(P1)}] $|V| \leq \gamma$ for all $V \in \mathcal{V}$, 
    \item[\textbf{(P2)}] for all $V \in \mathcal{V}$, the subgraph of $G$ induced by $V$ (denoted by $G[V]$) is connected (this follows from the assumption that $\mathcal{F}$ consists of connected graphs), and 
\end{itemize}
These properties hold for both \textsc{Subgraph Hitting} and \textsc{Induced Subgraph Hitting}.
Thus, in what follows, we shall not distinguish the two problems.
Fix an ordering $\sigma$ on $V(G)$ that satisfies $\mathsf{wcol}_\gamma(G,\sigma) = \mathsf{wcol}_\gamma(G)$.
Throughout this section, when we say a vertex is ``large'' or ``small'', it is always in terms of the ordering $\sigma$.
As $\mathcal{G}$ is of bounded expansion, $\mathsf{wcol}_\gamma(G) \leq \chi(\gamma)$ for some function $\chi$ only depending on $\mathcal{G}$.
Thus, $|\text{WR}_\gamma(G,\sigma,v)| \leq \mathsf{wcol}_\gamma(G,\sigma) \leq \chi(\gamma)$ for all $v \in V(G)$.
In other words, the size of $\text{WR}_\gamma(G,\sigma,v)$ is bounded by a function of $\mathcal{F}$ for $v \in V(G)$.
Our reduction consists of two steps.

\paragraph{Step 1.}
The first step reduces $(G,\mathcal{F})$ to an instance $(G_1,\mathcal{F})$ where $G_1$ is an induced subgraph of $G$ satisfying certain properties.
In this step, we deal with the \textit{sunflowers} in $\mathcal{V}$.
Recall that an \textit{$r$-sunflower} is a collection $\{V_1,\dots,V_r\}$ of $r$ sets which have a common intersection $X$ such that $V_1 \backslash X, \dots, V_r \backslash X$ are disjoint.
If $X$ is nonempty, it is called the \textit{core} of the sunflower.
Suppose $\mathcal{V}$ contains a sunflower $\{V_1,\dots,V_r\}$ with core $X$, for a large $r$.
In order to optimally hit all sets in $\mathcal{V}$, intuitively, we should pick a vertex in $X$, for otherwise we need to waste at least $r$ vertices outside $X$ to hit all of $V_1,\dots,V_r$.
But this intuition is incorrect in general: if we choose to hit $V_1,\dots,V_r$ using $r$ vertices outside $X$, those vertices can also be used to hit other sets in $\mathcal{V}$ (and thus are not necessarily wasted).
Interestingly, as we will see, it turns out to be useful in our setting.

For a vertex a vertex $v \in V(G)$ and a set $S \subseteq V(G)$, we denote by $\rho(v,S)$ the number of vertices $u \in S$ satisfying $v \in \text{WR}_\gamma(G,\sigma,u)$.
For an integer $t \geq 0$, we say a set $S \subseteq V(G)$ is \textit{$t$-closed} if $\rho(v,S) < t$ for all $v \in V(G) \backslash S$.
Note that if $S_1$ and $S_2$ are both $t$-closed, then $S_1 \cap S_2$ is also $t$-closed, because simply because $\rho(v,S_1 \cap S_2) \leq \min\{\rho(v,S_1),\rho(v,S_2)\}$ for all $v \in V(G)$.
Therefore, for any $S \subseteq V(G)$, there exists a unique minimal superset $S' \supseteq S$ that is $t$-closed, which we call the \textit{$t$-closure} of $S$.
The $t$-closure $S'$ of $S$ can be easily constructed as follows.
We begin with $S' = S$ and repeatedly add vertices $v \in V(G) \backslash S'$ satisfying $\rho(v,S') \geq t$ to $S'$.
One can easily verify that the resulting $S'$ is the $t$-closure of $S$.
\begin{lemma} \label{Olem-closure1}
    If $S'$ is the $t$-closure of $S$, then $|S'| \leq |S|/(1-\chi(\gamma)/t)$.
\end{lemma}
\begin{proof}
From the construction of $t$-closure, we see that every $v \in S' \backslash S$ satisfies $\rho(v,S') \geq t$, which implies $\sum_{v \in S' \backslash S} \rho(v,S') \geq t \cdot |S' \backslash S|$.
On the other hand, notice the equation
\begin{equation*}
    \sum_{v \in S' \backslash S} \rho(v,S') = \sum_{u \in S'} |\text{WR}_\gamma(G,\sigma,u) \cap (S' \backslash S)| \leq \sum_{u \in S'} |\text{WR}_\gamma(G,\sigma,u)| \leq \chi(\gamma) \cdot |S'|.
\end{equation*}
It follows that $t \cdot |S' \backslash S| \leq \chi(\gamma) \cdot |S'|$, which implies $|S'| \leq |S|/(1-\chi(\gamma)/t)$.
\end{proof}

\begin{lemma} \label{Olem-closure2}
    Let $S \subseteq V(G)$, $S'$ be the $(r-\chi(G))$-closure of $S$, and $\{V_1,\dots,V_r\}$ be a sunflower in $\mathcal{V}$ with core $X$.
    If $S$ is a hitting set of $\{V_1,\dots,V_r\}$, then $S' \cap X \neq \emptyset$.
\end{lemma}
\begin{proof}
If $S \cap X = \emptyset$, we are done.
Otherwise, for each $i \in [r]$, $S$ contains a vertex $u_i \in V_i \backslash X$.
Note that $u_1,\dots,u_r$ are all distinct, because $V_1 \backslash X,\dots,V_r \backslash X$ are disjoint.
Let $v^* \in X$ be the largest vertex in $X$, and $v_i \in V_i$ be the largest vertex in $V_i$ for $i \in [r]$.
We show $v^* \in S'$ and thus $S' \cap X \neq \emptyset$.
Define $I = \{i \in [r]: v_i = v^*\}$.
Observe the following facts.
\begin{itemize}
    \item $r - |I| \leq \chi(\gamma)$.
    To see this, first notice that $v_i \in \text{WR}_\gamma(G,\sigma,v^*)$ for all $i \in [r]$.
    Indeed, we have $v_i \geq v^*$ by construction.
    By properties \textbf{(P1)} and \textbf{(P2)}, $G[V_i]$ is connected and $|V_i| \leq \gamma$.
    So $v_i$ is weakly $\gamma$-reachable from any vertex in $V_i$ and in particular from $v^*$.
    Furthermore, $v_i \in V_i \backslash X$ for all $i \in [r] \backslash I$ (for otherwise $v_i$ is the largest vertex in $X$ and thus $v_i = v^*$).
    Thus, the vertices $v_i$'s for $i \in [r] \backslash I$ are distinct, as $V_1 \backslash X,\dots,V_r \backslash X$ are distinct.
    As all these vertices are in $\text{WR}_\gamma(G,\sigma,v^*)$, we have $r - |I| \leq |\text{WR}_\gamma(G,\sigma,v^*)| \leq \chi(\gamma)$.
    \item $v^* \in \text{WR}_\gamma(G,\sigma,u_i)$ for all $i \in I$.
    Note that $v^*$ is the largest vertex in $V_i$ for $i \in I$.
    As $G[V_i]$ is connected and $|V_i| \leq \gamma$, $v^*$ is weakly $\gamma$-reachable from any vertex in $V_i$ and thus from $u_i$.
\end{itemize}
The first fact implies $|I| \geq r - \chi(\gamma)$.
The second fact implies $\rho(v^*,S) \geq |I|$, since $u_1,\dots,u_r \in S$ and these vertices are distinct.
Therefore, $\rho(v^*,S) \geq r - \chi(\gamma)$ and thus $v^* \in S'$.
\end{proof}

Now we use the above lemmas to do our reduction.
Set $t = (1+\varepsilon) \cdot (\chi(\gamma)/\varepsilon)$ and $r = t + \chi(\gamma)$.
Then $1/(1-\chi(\gamma)/t) = 1+\varepsilon$.
We begin with $G_1 = G$ and keep removing vertices from $G_1$ while guaranteeing the following invariant.
\begin{itemize}
    \item For any solution $S \subseteq V(G_1)$ of the instance $(G_1,\mathcal{F})$, the $t$-closure of $S$ is a solution of $(G,\mathcal{F})$.
\end{itemize}
For a set $A \subseteq V(G)$, we denote by $\mathcal{V} \Cap A$ the sub-collection of $\mathcal{V}$ consisting of all sets that are contained in $V$.
We say a vertex $v \in V(G_1)$ is \textit{redundant} in $G_1$ if every $V \in \mathcal{V} \Cap V(G_1)$ containing $v$ also contains the core of an $r$-sunflower in $V \in \mathcal{V} \Cap (V(G_1) \backslash \{v\})$.
Then what we do is very simple: as long as the current $G_1$ has a redundant vertex, remove it from $G_1$.
The procedure terminates when $G_1$ does not have any redundant vertices.
We show that the desired invariant always holds during the procedure.
Assume it holds for $G_1$, and we want it to also hold for $G_1 - \{v\}$ for a redundant vertex $v \in V(G_1)$.
Let $S \subseteq V(G_1) \backslash \{v\}$ be a solution of $(G_1 - \{v\},\mathcal{F})$, and $S'$ be its $t$-closure.
We claim that $S'$ is a solution of $(G_1,\mathcal{F})$.
By Lemma~\ref{Olem-closure2}, $S'$ hits the cores of all $r$-sunflowers in $\mathcal{V} \Cap (V(G_1) \backslash \{v\})$.
Consider a set $V \in \mathcal{V} \Cap V(G_1)$.
If $v \notin V$, then $V$ is hit by $S$ and hence by $S'$.
Otherwise, since $v$ is redundant in $G_1$, $V$ should contain the core of an $r$-sunflower in $\mathcal{V} \Cap (V(G_1) \backslash \{v\})$.
But $S'$ hits that core, and thus hits $V$.
Thus, $S'$ hits all sets in $\mathcal{V} \Cap V(G_1)$ and is a solution of $(G_1,\mathcal{F})$.
As the invariant holds for $G_1$, the $t$-closure of $S'$, which is $S'$ itself, is a solution of $(G,\mathcal{F})$.

The invariant guarantees that we can reduce the original instance $(G,\mathcal{F})$ to the new instance $(G_1,\mathcal{F})$.
Indeed, given a $(1+\varepsilon)$-approximation solution $S \subseteq V(G_1)$ of $(G_1,\mathcal{F})$, we can compute the $t$-closure $S'$ of $S$, which is a solution of $(G,\mathcal{F})$ by the invariant and satisfies that $|S'| \leq |S|/(1-\chi(\gamma)/t)$ by Lemma~\ref{Olem-closure1}.
Since $1/(1-\chi(\gamma)/t) = 1+\varepsilon$, $S'$ is a $(1+O(\varepsilon))$-approximation solution of $(G,\mathcal{F})$.
This completes the first step of our reduction.
We omit the implementation details for constructing $G_1$ and for computing the $t$-closure of a set $S$ --- by properly using the existing algorithms for first-order model checking on bounded-expansion graphs and other results, both can be implemented in linear time.

\paragraph{Step 2.}
In this step, we further reduce $(G_1,\mathcal{F})$ to another instance $(G_2,\mathcal{F})$ where the maximum degree of $G_2$ is bounded (by a function of $\varepsilon$ and $\gamma$) and thus complete our proof.
What we do is very simple: pick a large threshold $d = O_{\varepsilon,\mathcal{F}}(1)$ and then let $G_2 = G_1 - V^*$, where $V^*$ consists of all vertices in $G_1$ whose degree is larger than $d$, i.e., $V^* = \{v \in V(G_1): \mathsf{deg}_{G_1}(v) > d\}$.
Given a solution $S \subseteq V(G_2)$ of $(G_2,\mathcal{F})$, we can simply construct a solution $S' \subseteq V(G_1)$ of $(G_1,\mathcal{F})$ by setting $S' = S \cup V^*$.
The only thing non-obvious is that if $S$ is a $(1+\varepsilon)$-approximation solution of $(G_2,\mathcal{F})$, then $S'$ is a $(1+O(\varepsilon))$-approximation solution of $(G_1,\mathcal{F})$.
We show this is true if the threshold $d$ is chosen properly.

It suffices to have $|V^*| = O(\varepsilon |S^*|)$ for an optimal solution $S^*$ of $(G_1,\mathcal{F})$.
We rely on a nice property of the graph $G_1$ guaranteed in the first step: it does not contain any redundant vertices.
In other words, for every $v \in V(G_1)$, there exists $V \in \mathcal{V} \Cap V(G_1)$ containing $v$ such that $V$ does not contain the core of any $t$-sunflower in $\mathcal{V} \Cap (V(G_1) \cap \{v\})$.
Let $S^*$ be an optimal solution of $(G_1,\mathcal{F})$ and $t = (1+\varepsilon) \cdot (\chi(\gamma)/\varepsilon)$ (same as what we used in Step~1).

\begin{lemma}
    Let $r > \max\{\gamma,t+\chi(\gamma)\}$ be an integer.
    If $d \geq r^\gamma \gamma!$, then $\rho(v,S^*) \geq t$ for all $v \in V^*$.
\end{lemma}
\begin{proof}
Since $v \in V^*$, the degree of $v$ in $G_1$ is larger than $d$.
Let $v_1,\dots,v_d \in V(G_1)$ be $d$ neighbors of $v$.
As aforementioned, for each $v_i$, there exists $V_i \in \mathcal{V} \Cap V(G_1)$ containing $v_i$ such that $V_i$ does not contain the core of any $t$-sunflower in $\mathcal{V} \Cap (V(G_1) \cap \{v_i\})$.
By the sunflower lemma, there exists an $r$-sunflower in $\{V_1,\dots,V_d\}$, and assume it is $V_1,\dots,V_r$ without loss of generality.
We claim that $V_1,\dots,V_r$ are disjoint.
Assume this is not the case.
Let $X$ be the core of $V_1,\dots,V_r$.
As $|X| \leq \gamma$ and $r > \gamma$, there exists $i \in [r]$ such that $v_i \notin X$.
Thus, $X$ is the core of an $r$-sunflower (and hence a $t$-sunflower as $r > t$) in $\mathcal{V} \Cap (V(G_1) \cap \{v_i\})$.
We have $X \subseteq V_i$, which contradicts the fact that $V_i$ does not contain the core of any $t$-sunflower in $\mathcal{V} \Cap (V(G_1) \cap \{v_i\})$.
So $V_1,\dots,V_r$ must be disjoint.

Nex we apply an argument similar to the proof of Lemma~\ref{Olem-closure2}.
We have $S^* \cap V_i \neq \empty$ for all $i \in [r]$, as $S^*$ is a solution of $(G_1,\mathcal{F})$.
Let $u_i \in S^* \cap V_i$ for $i \in [r]$.
For each $u_i$, there exists a path $\pi_i$ between $v$ and $u_i$ in $G[V_i \cup \{v\}]$ of length at most $\gamma$, since $G[V_i]$ is connected and $v$ is either in $V_i$ or neighboring to $V_i$.
Let $p_i$ be the largest vertex on $\pi_i$, and define $I = \{i \in [r]: p_i = v\}$.
Note that $p_i \in V_i$ for all $i \in [r] \backslash I$.
As $V_1,\dots,V_r$ are disjoint, the vertices $p_i$'s for $i \in [r] \backslash I$ are distinct.
We have $p_i \in \text{WR}_\gamma(G,\sigma,v)$, because of the sub-path of $\pi_i$ between $v$ and $p_i$.
It follows that $r - |I| \leq |\text{WR}_\gamma(G,\sigma,v)| \leq \chi(\gamma)$ and hence $|I| \geq r - \chi(\gamma) > t$.
On the other hand, we have $v \in \text{WR}_\gamma(G,\sigma,u_i)$ for all $i \in I$, because of the path $\pi_i$.
As $u_1,\dots,u_r \in S^*$ are distinct, $\rho(v,S^*) \geq |I| > t$.
\end{proof}

Set $d \geq (t+\chi(\gamma)+\gamma)^\gamma \gamma!$.
The above lemma implies $\rho(v,S^*) \geq t$ for all $v \in V^*$.
Finally, we use the same argument as in the proof of Lemma~\ref{Olem-closure1} to show that $|V^*| \leq \varepsilon |S^*|$:
\begin{equation*}
    t|V^*| \leq \sum_{v \in v^*} \rho(v,S^*) \leq \sum_{u \in S^*} |\text{WR}_\gamma(G,\sigma,u)| \leq \chi(\gamma) |S^*|.
\end{equation*}
This implies $|V^*| \leq \chi(\gamma) |S^*|/t = O(\varepsilon|S^*|)$, and we are done.

\subsection{Overview of the proofs of Theorems~\ref{thm-lsearch-hitting} and~\ref{thm-lsearch-packing}}

Let $\mathcal{C}$ be a collection of subsets of vertices of a graph $G$.
The \emph{packing graph} $G[\mathcal{C}]$ is defined as the graph with vertex set $\mathcal{C}$ where $C_1$ and $C_2$ are adjacent if there exist $v_1\in C_1$ and $v_2\in C_2$
such that $v_1=v_2$ or $v_1v_2\in E(G)$.  The collection is \emph{shallow} if for each $C\in \mathcal{C}$, the graph $G[C]$
has radius bounded by a constant, and \emph{thin} if every vertex of $G$ belongs to only constantly many sets of $\mathcal{C}$.
For $K\subseteq\mathcal{C}$, let $\partial K$ be the set of vertices of $K$ with a neighbor outside of $K$.
Consider an arbitrary {\sc Hitting Set} problem.  The framework of Har-Peled and Quanrud~\cite{har2017approximation} shows that
the local search gives a PTAS for this problem in graph classes with polynomial expansion if the following technical condition
holds:  For any solutions $A$ and $O$ to the problem in a graph $G$, we can find a function $C$ assigning subsets of $V(G)$ to vertices
of $A\cup O$ such that
\begin{itemize}
\item[(i)] the system $\mathcal{C}=\{C(v):v\in A\cup O\}$ of subsets of $G$ is shallow and thin, and
\item[(ii)] for any $K\subseteq V(G[\mathcal{C}])$, the set
$$A_K=(A\setminus C^{-1}(K\setminus\partial K))\cup (O\cap C^{-1}(K))$$
is also a valid solution.
\end{itemize}
The condition (i) can be used to show that if $G$ is from a class with polynomial expansion, then the packing graph $G[\mathcal{C}]$
also has expansion bounded by a polynomial.  Hence, $G[\mathcal{C}]$ has strongly sublinear separators, which by a standard result
dating back to Lipton and Tarjan~\cite{LiptonTargenSep} implies that for every $\varepsilon>0$, $G[\mathcal{C}]$ can be broken
into parts of size $O(\text{poly}(1/\varepsilon))$ with overlaps of total size at most $\varepsilon|V(G[\mathcal{C}])|\le \varepsilon(|A|+|O|)$.
We apply this construction with $O$ being an optimal solution and $A$ being the local search solution.  Then for each part $K$,
the condition (ii) and the local optimality of $A$ implies that $|A_K|\ge |A|$, and thus $O$ is not much better on $K$ than $A$.
From this, an easy computation concludes that $A$ is a good approximation of the optimal solution $O$.

For example, in the case of {\sc Vertex Cover} or  {\sc $K_2$-Hitting  Set},
we can let $C(v)=\{v\}$ for each $v\in O\cup A$.
Consider any edge $e=uv$ of $G$.  If $\{u,v\}\cap A\setminus C^{-1}(K\setminus\partial K)\neq\emptyset$, then $A_K$ hits $e$.
Otherwise, since $A$ is a vertex cover, we can assume $u\in A\cap C^{-1}(K\setminus\partial K)$.
If $u\in O$, then $u\in O\cap C^{-1}(K)$ and again $A_K$ hits $e$.  Hence, suppose that $u\not\in O$, and since $O$ is
a vertex cover, we have $v\in O$.  Then $C(u)$ and $C(v)$ are adjacent in $G[\mathcal{C}]$, and since $C(u)\in K\setminus \partial K$,
we have $C(v)\in K$.  But then $v\in O\cap C^{-1}$, and $A_K$ hits $e$.  Therefore, $A_K$ is also a vertex cover, and the
condition (ii) holds.

To illustrate the difficulty with generalizing this idea to the proof of Theorem~\ref{thm-lsearch-hitting},
consider the {\sc $P_3$-Hitting Set}  problem.  If $P=uxv$ is a $3$-vertex path in $G$, $A$ intersects $P$ in $u$ and $O$ in $v$,
then $u$ and $v$ would not necessarily be adjacent in the graph $G[\mathcal{C}])$ defined in the same way, and thus
the condition (ii) could fail to hold.  An obvious solution is to let $C(v)$ for $v\in O$ be not just a single vertex, but
a subgraph of $G$ of bounded radius.  For example, letting $C(v)$ be the closed neighborhood $N[v]$ of $v$ in $G$ would
ensure that $C(u)C(v)\in E(G[\mathcal{C}])$, which in turn would imply (ii).  An issue is that then (i) might be false,
since if a vertex $z$ has many neighbors in $O$, it would belong to many sets of the system $\mathcal{C}$.

The observation that we use to solve this issue is that because $G$ has bounded expansion, the number of vertices with many neighbors
in $O$ is small, say less than $\varepsilon|O|$, and instead of comparing $A$ with the optimal solution $O$,
we can compare it with a near-optimal solution $O'$ consisting of $O$ and the set $O_1$ of vertices of $G$ with many neighbors in $O$.
We can then define $C(v)=N[v]\setminus O_1$ for $v\in O$ and $C(v)=\{v\}$ for $v\in (A\cup O')\setminus O$.
This system is thin, since a vertex $z$ can belong to $C(v)\setminus \{v\}$ only for its neighbors $v\in O$ and only if $z\not\in O_1$, and
the number of such neighbors is small by the definition of $O_1$.  Moreover, this is sufficient to ensure that if $u\in A$ and $v\in O$ both intersect a $3$-vertex
path in $G$, then $C(u)$ has a neighbor $C(v')$ in $G[\mathcal{C}]$ with $v'\in O'$, where either $v'=v$, or $v'$ is a neighbor of $v$ in $O_1$.
This implies that both (i) and (ii) hold, and the framework of Har-Peled and Quanrud~\cite{har2017approximation} applies.

For general bounded-diameter properties, we need to be somewhat more careful in selecting the elements of the shallow cover
(selecting full neighborhoods up to the diameter of the property would work for (ii), but it is completely hopeless for (i)).
However, with a little extra effort, this can be achieved by basing the elements of the cover on weak reachability
in an ordering with bounded weak coloring number, similarly to the arguments from Section~\ref{sec-overview1}.

The proof of Theorem~\ref{thm-lsearch-packing} is much less involved, and all its ideas come from~\cite{har2017approximation};
we include it for completeness, as it is interesting and not explicitly formulated there.

\subsection{Overview of the proofs of Theorems~\ref{thm-geo}}
As mentioned in Section~\ref{sec-app}, the main ingredient of Theorem~\ref{thm-geo} is an efficient algorithm for computing $k$-clique decompositions on graph classes of clique-dependent bounded expansion.
Let $\mathcal{G}$ be a hereditary graph class of clique-dependent bounded expansion, and $G \in \mathcal{G}$ be a graph.
Our algorithm is based on a simple divide-and-conquer approach.
We evenly divide $V(G)$ into two subsets $V'$ and $V''$.
Then we recursively compute $k$-clique decompositions on $G[V']$ and $G[V'']$.
Once $G[V']$ and $G[V'']$ have been decomposed, we only need to further compute a $k$-clique decomposition on $G[V_0' \cup V_0'']$, where $V_0'$ (resp., $V_0''$) is the part of $G[V']$ (resp., $G[V'']$) without $k$-cliques.
As $G[V_0']$ and $G[V_0'']$ do not contain $k$-cliques, $G[V_0' \cup V_0'']$ does not contain $2k$-cliques.
Therefore, $G[V_0' \cup V_0''] \in \mathcal{G}_{2k}$ where $\mathcal{G}_{2k} \subseteq \mathcal{G}$ consists of all graphs with maximum clique size at most $2k$.
Since $\mathcal{G}$ is of clique-dependent bounded expansion, $\mathcal{G}_{2k}$ is of bounded expansion.
Finding small subgraphs in bounded-expansion graph classes takes linear time~\cite{dvovrak2013testing}.
It follows that computing a $k$-clique decomposition on $G[V_0' \cup V_0'']$ can be done in $f(k) \cdot n$ time.
A straightforward implementation of this approach takes $f(k) \cdot n \log n + O(m \log n)$ time.
A more careful implementation using the algorithm of Harel and Tarjan~\cite{harel1984fast} for finding lowest common ancestors can improve the running time to $f(k) \cdot n \log n + O(m)$ time.

\section{Preliminaries}

\paragraph{Basic notations.}
We use $\mathbb{N}$ to denote the set $\{1,2,3,\ldots\}$. For an integer $n\in \mathbb{N}$, $[n]=\{1,2,\ldots,n\}$.
Let $G$ be a graph. We use $V(G)$ and $E(G)$ to denote the set of vertices and the set of edges in $G$, respectively.
For a vertex subset $S\subseteq V(G)$, $N_G(S)=\{u\in V(G)\setminus S~\colon~ (u,x)\in E(G) \mbox{ for some }x\in S\}$ and 
$N_G[S]=N_G(S)\cup S$. 
A \textit{subgraph} of $G$ is a graph $G' = (V',E')$, denoted by $G'\subseteq G$, where $V' \subseteq V(G)$ and $E' \subseteq E(G)$. The graph $G'$ is an {\em induced subgraph} of $G$, denoted by $G'\subseteq_{\mathsf{in}} G$, if $E'=\{(u,v)\in E(G)~\colon~u\in V' \mbox{ and } v\in V'\}$. 
For a set $V \subseteq V(G)$, the notation $G[V]$ denotes the subgraph of $G$ induced by $V$.
The notation $\omega(G)$ denotes the size (i.e., number of vertices) of a maximum clique in $G$.
For a graph $G$ and a collection $\mathcal{F}$ of graphs, we define
\begin{itemize}
    \item $\mathcal{V}_\mathcal{F}(G) = \{V(G'): G' \text{ is a subgraph of } G \text{ isomorphic to some } F \in \mathcal{F}\}$,
    \item $\mathcal{V}_\mathcal{F}^\mathsf{in}(G) = \{V(G'): G' \text{ is an induced subgraph of } G \text{ isomorphic to some } F \in \mathcal{F}\}$.
\end{itemize}


For a graph $G$, a vertex ordering $\sigma$ is a bijection from $V(G)$ to $\{1,2,\ldots |V(G)|\}$. 
For $u,v \in V(G)$, we write $u <_\sigma v$ if $\sigma(u)<\sigma(v)$. 
The notations $>_\sigma$, $\leq_\sigma$, $\geq_\sigma$ are defined analogously. For a vertex subset $Y\subseteq V(G)$, we use $\max(Y,\sigma)$ to denote the vertex $x\in Y$ such that $\sigma(x)=\max_{y\in S}\{\sigma(y)\}$. Here, we say that $x$ is the largest vertex in $Y$ under the ordering $\sigma$.

\paragraph{Set system and sunflowers.}
A \textit{set system} refers to a collection $\mathcal{S}$ of sets whose elements belong to the same ground set $U$.
Let $\mathcal{S}$ be a set system with ground set $U$.
For a subset $U' \subseteq U$, we denote by $\mathcal{S} \Cap U'$ the sub-collection of $\mathcal{S}$ consisting of all sets that are contained in $U'$, i.e., $\mathcal{S} \Cap U' = \{S \in \mathcal{S}: S \subseteq U'\}$.

An \textit{$r$-sunflower} (for $r \geq 2$) in $\mathcal{S}$ is a sub-collection $\mathcal{R} \subseteq \mathcal{S}$ of size $r$ satisfying the follow condition: there exists $X \subseteq U$ such that $R \cap R' = X$ for any different sets $R, R' \in \mathcal{R}$; we call $X$ the \textit{core} of the sunflower $\mathcal{R}$.
One of the most important results about sunflowers is the following sunflower lemma.

\begin{lemma}[sunflower lemma] \label{lem-sunflower}
Let $\mathcal{S}$ be a set system in which every set is of size at most $k$.
If $|\mathcal{S}| \geq (r-1)^k \cdot k!$, then there exists an $r$-sunflower in $\mathcal{S}$.
\end{lemma}



\paragraph{Minors, graphs of bounded expansion, and weakly coloring number.}

For a graph $G$, and two vertices $u,v\in V(G)$, $\dist_G(u,v)$ denotes the length of a shortest path from $u$ to $v$ in $G$. For a graph $G$, ${\sf rad}(G)$ denotes the radius of $G$. That is, ${\sf rad}(G)=\min_{u\in V(G)} \max_{x\in V(G)} \dist_G(u,x)$

\begin{definition}[minor]
A graph $H$ is a \textbf{minor} of a graph $G$, denoted by $H \preceq G$, if there is a function $\phi$ from $V(H)$ to connected subgraphs of $G$ such that following conditions hold. 
\begin{itemize}
    \item For any two distinct vertices $u,v\in V(H)$, $V(\phi(u))\cap V(\phi(v))=\emptyset$.  
    \item For any edge $(u,v)\in E(H)$, there exist $x\in V(\phi(u))$ and $y\in V(\phi(v))$ such that $(x,y)\in E(G)$. 
\end{itemize}
The function $\phi$ is called the model of $H$ in $G$. 
\end{definition}

\begin{definition}[depth-$r$ minor]
Let $r$ be a positive integer. 
A graph $H$ is a \textbf{depth-$r$ minor} of $G$, denoted by $H \preceq_r G$, if $H$ is a minor of $G$ with a model $\phi$ such that for all $v\in V(H)$, ${\sf rad}(\phi(v))\leq r$. 
\end{definition}

\begin{definition}[bounded expansion and polynomial expansion]
Let $G$ be a graph and $r \geq 1$ be an integer. We define 
$$\nabla_r(G) =\sup \left\{ \frac{|E(H)|}{|V(H)|}~\colon~ H \preceq_r G\right\}.$$
We say that a graph class $\mathcal{G}$ has \textbf{bounded expansion}, if there is a function $f~\colon~\mathbb{N}\rightarrow \mathbb{N}$ such that for all $r\in \mathbb{N}$ and $G\in \mathcal{G}$, $\nabla_r \leq f(r)$. We use $\nabla_r(\mathcal{G})$  to denote $\sup_{G\in \mathcal{G}} \nabla_r(G)$. 
If $f$ is a polynomial function, then $\mathcal{G}$ is a class of \textbf{polynomial expansion}. 
\end{definition}

\begin{theorem}[\cite{NESETRIL2008777}]
\label{thm:IsoBoundedExp}
  Let $\mathcal{G}$ be a graph class of bounded expansion and let $H$ be a fixed graph. There is a linear time algorithm that given a graph $G\in \mathcal{G}$ and $S\subseteq V(G)$ and decides whether there is an (induced) subgraph $G'$ of $G$ containing at least one vertex from $S$ and isomorphic to $H$. 
\end{theorem}


Let $G$ be a graph and $\sigma$ be an ordering of the vertices of $G$.
For an integer $r \geq 0$ and two vertices $u,v \in V(G)$, we say $u$ is \textit{$r$-weakly reachable} from $v$ \textit{under $\sigma$} if there exists a path $\pi$ in $G$ between $v$ and $u$ of length at most $r$ such that $u \geq_\sigma w$ for any vertex $w$ on $\pi$ (in particular, $u \geq_\sigma v$).
Let $\text{WR}_r(G,\sigma,v)$ denote the set of vertices in $G$ that are $r$-weakly reachable from $v$ under $\sigma$.
The \textit{weak $r$-coloring number} of $G$ \textit{under $\sigma$} is $\mathsf{wcol}_r(G,\sigma) = \max_{v \in V(G)} |\text{WR}_r(G,\sigma,v)|$.
Then the \textit{weak $r$-coloring number} of $G$ is defined as $\mathsf{wcol}_r(G) = \min_{\sigma \in \Pi(G)} \mathsf{wcol}_r(G,\sigma)$ where $\Pi(G)$ is the set of all orderings of the vertices of $G$.

\begin{theorem}[\cite{nevsetvril2008grad,zhu2009colouring}] \label{thm-weakcolor}
A graph class $\mathcal{G}$ is of bounded expansion iff there exists a function $\chi_\mathcal{G}: \mathbb{N} \rightarrow \mathbb{N}$ such that $\mathsf{wcol}_r(G) \leq \chi_\mathcal{G}(r)$ for all $G \in \mathcal{G}$ and $r \in \mathbb{N}$.
\end{theorem}

\begin{theorem}[\cite{reidl2019characterising}] \label{thm-neighborcomp}
Let $\mathcal{G}$ be a graph class of bounded expansion.
Then there exists some constant $c_\mathcal{G}$ such that for any graph $G \in \mathcal{G}$ and any $V \subseteq V(G)$, we have $|\{N_G[\{v\}] \cap V: v \in V(G)\}| \leq c_\mathcal{G} \cdot |V|$.
\end{theorem}

\begin{theorem}[\cite{dvovrak2013constant}] \label{thm-ordering}
Let $\mathcal{G}$ be a graph class of bounded expansion.
Given a graph $G \in \mathcal{G}$ and an integer $r \in [n]$, one can compute in $O(n)$ time an ordering $\sigma$ of $V(G)$ satisfying $\mathsf{wcol}_r(G,\sigma) \leq \mathsf{wcol}_r^{r+1}(G)$ and a number $\alpha$ satisfying $\alpha \leq \mathsf{wcol}_r(G) \leq \alpha^{r+1}$.
\end{theorem}

\paragraph{First-order model checking.}
We deal with first-order logic for graphs, where the variables correspond to vertices of the given graph.
Let $G$ be a graph and $Q \subseteq V(G)$ be a unary relation on $V(G)$, i.e., the set of variables.
The first-order formulas over $(G,Q)$ are built by combining three types of {\em atomic formulas}, defined as follows:
\begin{itemize}
    \item Equality for variables $u,v$, denoted $u=v$: True iff $u$ and $v$ are equal (i.e., the same vertex).
    \item Adjacency for variables $u,v$, denoted $\mathsf{Adj}(u,v)$: True iff $u$ and $v$ are adjacent in $G$.
    \item $Q$-satisfability for a variable $v$, denoted by $Q(v)$: True iff $v$ has the relation $Q$, i.e., $v \in Q$.
\end{itemize}
Then, combination is done based on the following constructions:
\begin{itemize}
    \item Boolean connectives for formulas $\varphi,\psi$: And, denoted $\varphi\wedge\psi$, which is true if and only if both $\varphi$ and $\psi$ are true; Or, denoted $\varphi\vee\psi$, which is true if and only if at least one among $\varphi$ and $\psi$ is true; Not, denoted $\neg\varphi$, which is true if and only if $\varphi$ is false.
    \item Quantification for formula $\varphi$ and variable $v$: Exists, denoted $\exists v\ \varphi(v)$, which is true iff there exists a vertex in $G$ so that, when it assigned to $v$, $\varphi$ is true; For all, denoted $\forall v\ \varphi(v)$, which is true ifff for every vertex in $G$, when it assigned to $v$, $\varphi$ is true.
\end{itemize}

We shall use the following dynamic data structure for first-order model-checking on graphs of bounded expansion.

\begin{theorem}[\cite{dvovrak2013testing}]
\label{thm:fologic}
Let $\mathcal{G}$ be a graph class of bounded expansion.
Given an $n$-vertex graph $G \in \mathcal{G}$, a unary relation $Q \subseteq V(G)$, and a first-order formula $\varphi$ over $(G,Q)$, one can build in $O(n)$ time a data structure on $G$ that supports the following operations (where $f$ is some function):
\begin{itemize}
    \item Testing whether $(G,Q)$ satisfies $\varphi$ in $f(|\varphi|)$ time.
    \item Deleting an element from $Q$ or deleting an edge from $G$ with $f(|\varphi|)$ update time.
    \item Adding back an element to $Q$ or adding back an edge to $G$ with $f(|\varphi|)$ update time.
    Note that here one can only add elements (resp., edges) that were previously deleted from $Q$ (resp., $G$).
\end{itemize}
\end{theorem}

\section{Degree reduction and lossy kernels}

In this section, we prove our degree-reduction theorem, which is restated below.

\mainthm*


\subsection{Reducing to connected forbidden patterns} \label{sec-toconnect}

Our first step is to reduce the \textsc{(Induced) Subgraph Hitting} with general forbidden patterns to the same problem with \textit{connected} forbidden patterns.
For a graph $F$, let $\mathcal{C}_F$ denote the set of connected components of $F$ (which is a set of connected graphs).
For a finite set $\mathcal{F} = \{F_1,\dots,F_m\}$ of graphs, define $\mathsf{Conn}(\mathcal{F}) = \{\{C_1,\dots,C_m\}: C_i \in \mathcal{C}_{F_i}\}$, that is, $\mathsf{Conn}(\mathcal{F})$ consists of all sets $\mathcal{F}'$ of graphs where $\mathcal{F}'$ can be constructed by including exactly one graph in $\mathcal{C}_F$ for each $F \in \mathcal{F}$.
For a graph $G$ and for a finite set of graph $\mathcal{F}$, we use $\mathsf{opt}(G,\mathcal{F})$ and $\mathsf{opt}^*(G,\mathcal{F})$ to denote the cardinalities of the optimum solutions to $(G,\mathcal{F})$ for the problems {\sc Subgraph Hitting} and {\sc Induced Subgraph Hitting}, respectively.
We have the following observation for \textsc{Subgraph Hitting}.

\begin{lemma}
\label{lem:sub1}
For every finite set $\mathcal{F}$ of graphs, there exists a number $\alpha_\mathcal{F}$ such that for any graph $G$, we have $\mathsf{opt}(G,\mathcal{F}') \leq \mathsf{opt}(G,\mathcal{F}) + \alpha_\mathcal{F}$ for some $\mathcal{F'} \in \mathsf{Conn}(\mathcal{F})$. 
Here, $\alpha_\mathcal{F}=|\mathcal{F}|\cdot \ell\cdot c^2$, where $\ell$ is the maximum number of connected components in a graph $F$ in $\mathcal{F}$ and $c$ is the maximum number of vertices in a connected component of a graph in $\mathcal{F}$.
\end{lemma}

\begin{proof}
Let $S$ be an optimum solution for the \textsc{Subgraph Hitting} problem  on $(G,\mathcal{F})$. 
Now we claim that for each $F\in \mathcal{F}$, there is a connected component $D_F$ in $\mathcal{C}_F$ such that the maximum number of vertex disjoint subgraphs of $G-S$ that are isomorphic to $D_F$ is at most $\ell \cdot c$. Towards the proof let us fix a graph $F\in \mathcal{F}$. Let $C_1,\ldots, C_{\ell'}$ be the connected components of $F$. 
For the sake of contradiction, assume that for all  $i\in \{1,2,\ldots,\ell'\}$, the maximum number of vertex disjoint subgraphs of $G-S$ that are isomorphic to $C_i$ is strictly more than $\ell \cdot c$. Now we explain a procedure to obtain a subgraph in $G-S$ that is isomorphic to $F$ which will be a contradiction to the fact that $S$ is a solution for \textsc{Subgraph Hitting} on $(G,\mathcal{F})$.
First, we pick a subgraph $H_1$  of $G-S$ that is isomorphic to $C_1$. Since $|V(H_1)|\leq c$, for all $j\in \{2,\ldots,\ell'\}$, there are strictly more than $(\ell-1)c$ vertex disjoint subgraphs of $G_1=G-(S\cup V(H_1))$ that are isomorphic to $C_j$. Next, we pick a subgraph $H_2$ of $G_1$ that is isomorphic to $C_2$. At the end of step $i$, we have vertex disjoint subgraphs $H_1,H_2,\ldots,H_i$ of $G-S$ such that for all $i'\in \{1,2,\ldots,i\}$, $H_{i'}$ is isomorphic to $C_{i'}$ and  for all $j\in \{i+1,\ldots,\ell'\}$, there are strictly more than $(\ell-i)c$ vertex disjoint subgraphs of $G_i=G-(S\cup \bigcup_{r\in [i]} V(H_r))$ that are isomorphic to $C_j$. In step $i+1$, we pick a subgraph $H_{i+1}$ of $G_i$ that is isomorphic to $C_{j+1}$. Thus, since $\ell'\leq \ell$,  at the end of step $\ell'$ we get vertex disjoint subgraphs $H_1,H_2,\ldots,H_q$ of $G-S$ such that for all $r\in \{1,2,\ldots,\ell'\}$, $H_{r}$ is isomorphic to $C_{r}$. This implies that $F$ is a subgraph of $G-S$ which is a contradiction.   

Thus, we have proved that for each $F\in \mathcal{F}$, there is a connected component $D_F$ in $\mathcal{C}_F$ such that the maximum number of vertex disjoint subgraphs of $G-S$ that are isomorphic to $D_F$ is at most $\ell \cdot c$. 
Then,  for each $F\in \mathcal{F}$, there is a vertex subset $S_F$ of $V(G)\setminus S$, of size at most $\ell c^2$ such that $S_F$ intersects with all the subgraphs of $G-S$ that are isomorphic to $D_F$. This implies that $S^{\star}=S\cup \bigcup_{F\in \mathcal{F}} S_F$ is a solution for \textsc{Subgraph Hitting}  on $(G,\mathcal{F}')$ where $\mathcal{F}' =\{D_F ~:~F\in \mathcal{F}\}$ and $|S^{\star}|\leq |S|+|{\mathcal{F}}|\cdot \ell \cdot c^2 = \mathsf{opt}(G,\mathcal{F}) + \alpha_\mathcal{F}$. Here, $\alpha_\mathcal{F}=|\mathcal{F}|\cdot \ell \cdot c^2$. 
\end{proof}


For \textsc{Induced Subgraph Hitting}, we can have a result similar to the above lemma.
However, it only applies to sparse graphs, or more precisely, graphs of bounded degeneracy.

\begin{lemma}
For every graph $G$ of degeneracy $d$ and every finite set $\mathcal{F}$ of graphs, there exists a number $\alpha_\mathcal{F}^*$ (depending on $\mathcal{F}$ and $d$) such that $\mathsf{opt}^*(G,\mathcal{F}') \leq \mathsf{opt}^*(G,\mathcal{F}) + \alpha_\mathcal{F}^*$ for some $\mathcal{F'} \in \mathsf{Conn}(\mathcal{F})$. Here, $\alpha^*_\mathcal{F}=|\mathcal{F}|\ell^2 c^3 d$, where $\ell$ is the maximum number of connected components in a graph $F$ in $\mathcal{F}$ and $c$ is the maximum number of vertices in a connected component of a graph in~$\mathcal{F}$.
\end{lemma}

\begin{proof}
The proof strategy is similar to the proof of Lemma~\ref{lem:sub1}. But, here we use the fact that the degeneracy of ${G}$ is $d$. 
  Let $S$ be an optimum solution for the \textsc{Induced Subgraph Hitting} problem  on $(G,\mathcal{F})$. 
First, we claim that for each $F\in \mathcal{F}$, there is a connected component $D_F$ in $\mathcal{C}_F$ such that the maximum number of vertex disjoint induced subgraphs of $G-S$ that are isomorphic to $D_F$ is at most $\ell^2 c^2 d$. Towards the proof of the claim, we fix a graph $F\in \mathcal{F}$ and let $C_1,\ldots, C_{\ell'}$ be the connected components of $F$. 
Suppose the claim is not true. Then, for all  $i\in \{1,2,\ldots,\ell'\}$, the maximum number of vertex disjoint induced subgraphs of $G-S$ that are isomorphic to $C_i$ is strictly more than $\ell^2 c^2 d$. Then, we explain a procedure to obtain an induced subgraph isomorphic to $F$ 
in $G-S$, leading to a contradiction to the fact that $S$ is a solution for \textsc{Induced Subgraph Hitting} on $(G,\mathcal{F})$. 

Let $p=\ell \cdot c \cdot d$. 
As a first step we construct sets of induced subgraphs $\mathcal{I}_1,\ldots,.\mathcal{I}_{\ell'}$ of $G-S$ such that the graphs in $\bigcup_{i\in [\ell']} \mathcal{I}_i$ are pairwise vertex disjoint and for each $i\in \{1,\ldots,\ell'\}$, $\vert \mathcal{I}_i \vert=p$ and the graphs in $\mathcal{I}_i$ are isomorphic to $C_i$. Towards that, we pick a set of $p$ vertex disjoint induced subgraphs $\mathcal{I}_1= \{H_{1,1},H_{1,2},\ldots,H_{1,p}\}$ of $G-S$ that are isomorphic to $C_1$.  Since $|V(H_{1,r})|\leq c$ for all $r\in \{1,2,\ldots,p\}$, there are strictly more than $\ell^2c^2d-(p\cdot c)$ vertex disjoint induced subgraphs of $G_1=G-(S\cup V_1)$ that are isomorphic to $C_j$ for all $j\in \{2,\ldots,\ell'\}$, where $V_1=\bigcup_{H\in \mathcal{I}_1} V(H)$. 
Next, we pick a set of $p$ induced subgraphs $\mathcal{I}_2=\{H_{2,1},\ldots,H_{2,p}\}$ of $G_1$ that are isomorphic to $C_2$. 
In general, at the end of step $i$, we have sets $\mathcal{I}_1,\ldots \mathcal{I}_i$ of induced subgraphs of $G-S$ such that the following properties holds. 
\begin{itemize}
\item For each $j\in \{1,2,\ldots,i\}$, $|\mathcal{I}_i|=p$ and each graph in $\mathcal{I}_i$ is isomorphic to $C_i$. 
\item The induced graphs of $G-S$ in $\mathcal{H}_{\leq i}=\bigcup_{j\in [i]} \mathcal{I}_j$ are pairwise vertex disjoint.  
    \item For any $j\in \{i+1,\ldots,\ell\}$
there are strictly more than $\ell^2c^2d - (p\cdot i\cdot c)=d \ell c^2(\ell- i)$ vertex disjoint induced subgraphs of $G_i=G-(S\cup \bigcup_{H\in \mathcal{H}_{\leq i}} V(H))$   that are isomorphic to $C_j$.
\end{itemize}

In step $i+1$, we pick a set of $p$ vertex disjoint induced subgraphs $\mathcal{I}_{i+1}=\{H_{i+1,1},\ldots, H_{i+1,p}\}$ of $G_i$ that are isomorphic to $C_{i+1}$. Thus, since $\ell' \leq \ell$, at the end of  step $\ell'$,   
we obtain a set of vertex disjoint induced subgraphs $\mathcal{H}_{\leq \ell'}=\bigcup_{r\in [\ell']} \mathcal{I}_r$ of $G-S$. Notice that for any $i\in [\ell']$, the graphs in $\mathcal{I}_i$ are isomorphic to $C_i$.  


Next, for each $i\in [\ell']$, we identify a graph $J_i$ from $\mathcal{I}_i$ such that there is no edge in $G$ between any two vertices from distinct $J_i$ and $J_{i'}$. In other words, $J_1\cup J_2\cup \ldots \cup J_{\ell'}$ form an induced subgraph of $G-S$ that is isomorphic to $F$ which will be a contradiction. Let $G^{\star}$ be the subgraph of $G-S$ induced on the union of the vertices of the graphs in  $\mathcal{H}_{\leq \ell'}$. Notice that the number of vertices in each graph in $\mathcal{H}_{\leq \ell'}$ is at most $c$. Since $G-S$ (and hence $G^{\star}$)
has degeneracy at most $d$, there is a graph $J\in \mathcal{H}_{\leq \ell'}$ such that $\delta_{G^{\star}
}(V(J))\leq c \cdot d$. Let $i_1 \in [\ell']$ be the index such that $J\in \mathcal{I}_{i_1}$. Now, we set $J_{i_1}=J$. Let $G^{\star}_1$ be the subgraph of $G^{\star}$ induced on the union of the vertices of the graphs $H$ in $H_{\leq \ell'}\setminus I_{i_1}$
such that no vertex in $H$ is adjacent to a vertex in $J_{i_1}$. Notice that for each $j\in \{1,2,\ldots,\ell'\} \setminus \{i_1\}$, there are at least $p-(c\cdot d)$ graphs in $\{H \in I_j ~\colon~ V(H)\cap N(V(J_{i_1}))=\emptyset\}$. Now, there is an index $i_2\in \{1,2,\ldots,\ell'\}\setminus \{i_1\}$ and a graph $J_{i_2}$ in $\{H \in I_{i_2} ~\colon~ V(H)\cap N(V(J_{i_1}))=\emptyset\}$ such that $\delta_{G^{\star}_1}(V(J_{i_2}))\leq c \cdot d$. We continue this process, and at the end of step $j$, we have $j$ distinct indices $i_1,i_2,\ldots,i_j$ and graphs $J_{i_1}\in \mathcal{I}_{i_1}, J_{i_2}\in \mathcal{I}_{i_2}, \ldots, J_{i_j}\in \mathcal{I}_{i_j}$ with the following properties. 

\begin{itemize}
    \item There do not exist two distinct indices $s$ and $t$ in $\{i_1,i_2,\ldots,i_j\}$ such that there is an edge in $G$ between a vertex in $V(J_{s})$ and a vertex in $V(J_t)$. 
    \item For each $j'\in \{1,2,\ldots,q\}\setminus \{i_1,i_2,\ldots,i_j\}$, 
$\vert \{H \in \mathcal{I}_{j'} ~\colon~ V(H)\cap N(\bigcup_{r\in [j]} V(J_{i_r}))=\emptyset\}\vert$ is at least $p- (j \cdot c \cdot d)$. 

\end{itemize}

Then, we find an index $i_{j+1}$ and a graph $J_{j+1}\in \mathcal{I}_{j+1}$ such that no vertex in $V(J_{i_{j+1}})$ is adjacent to a vertex in $\bigcup_{r\in [j]}V(J_{i_r})$ and 
$\delta_{G^{\star}_{j}
}(V(J_{i_{j+1}}))\leq c \cdot d$. Here, $G^{\star}_{j}$ is the graph induced on the union of the vertices in the graphs $H$ in $\mathcal{H}_{\leq \ell'} \setminus (\bigcup_{r\in [j]} \mathcal{I}_{i_r})$ such that no vertex in $V(H)$ is adjacent to a vertex in $\bigcup_{r\in [j]} V(J_{i_j})$.  
Thus, at the end of step $\ell'$, we obtain  $J_{1}\in \mathcal{I}_{1}, J_{2}\in \mathcal{I}_{2}, \ldots, J_{{\ell'}}\in \mathcal{I}_{{\ell'}}$ with the following property.  
There do not exist two distinct indices $s$ and $t$ in $\{1,2,\ldots,\ell'\}$ such that there is an edge in $G$ between a vertex in $V(J_{s})$ and a vertex in $V(J_t)$. This implies that the subgraph of $G-S$ induced on $\bigcup_{r\in [\ell']} V(J_r)$ is isomorphic to $F$. This is a contradiction to the fact that $S$ is a solution. 

Thus, we have proved that for each $F\in \mathcal{F}$, there is a connected component $D_F$ in $\mathcal{C}_F$ such that the maximum number of vertex disjoint induced subgraphs of $G-S$ that are isomorphic to $D_F$ is at most $\ell^2c^2 d$. Then,  for each $F\in \mathcal{F}$, there is a vertex subset $S_F$ of $V(G)\setminus S$, of size at most $\ell^2 c^3 d$ such that $S_F$ intersects with all the induced subgraphs of $G-S$ that are isomorphic to $D_F$. This implies that $S^{\star}=S\cup \bigcup_{F\in \mathcal{F}} S_F$ is a solution for \textsc{Induced Subgraph Hitting} on $(G,\mathcal{F}')$ where $\mathcal{F}' =\{D_F ~:~F\in \mathcal{F}\}$ and $|S^{\star}|\leq |S|+|\mathcal{F}|\ell^2 c^3 d = \mathsf{opt}^*(G,\mathcal{F}) + \alpha^*_\mathcal{F}$. Here $\alpha^*_\mathcal{F}=|\mathcal{F}|\ell^2 c^3 d$. 
\end{proof}

Using the above two lemmas, we can do our reduction from general forbidden patterns to connected forbidden patterns.

\begin{corollary} \label{cor-toconnect}
Let $\mathcal{G}$ be any graph class of bounded expansion.
If \textsc{(Induced) Subgraph Hitting} on $\mathcal{G}$ with connected forbidden patterns admits an approximation scheme with running time $f_0(\varepsilon,\mathcal{F}) \cdot n^c$ for some function $f_0$, then \textsc{(Induced) Subgraph Hitting} on $\mathcal{G}$ (with arbitrary forbidden patterns) admits an approximation scheme with running time $f(\varepsilon,\mathcal{F}) \cdot n^c$ for some function $f$.
\end{corollary}
\begin{proof}
We prove the lemma for \textsc{Subgraph Hitting} and the case for \textsc{(Induced) Subgraph Hitting} is identical and hence omitted. 
Let $\mathcal{A}$ be an algorithm that given $(G,\mathcal{F})$ and $\varepsilon>0$, outputs a $(1+\varepsilon)$ approximate solution for \textsc{Subgraph Hitting} on  $(G,\mathcal{F})$, where $\mathcal{F}$ is a set of connected forbidden patterns. 

Now, we design an algorithm $\mathcal{A}'$  using $\mathcal{A}$ when the forbidden patterns in $\mathcal{F}$ are arbitrary. Let $(G,\mathcal{F})$ be the input and $\varepsilon>0$ be the given error parameter. 
Recall that $\mathsf{Conn}(\mathcal{F}) = \{\{C_1,\dots,C_m\}: C_i \in \mathcal{C}_{F_i}\}$. Notice that for any $\mathcal{F}_1\in \mathsf{Conn}(\mathcal{F})$ a solution for \textsc{Subgraph Hitting} on $(G,\mathcal{F}_1)$ is also a solution to $(G,\mathcal{F})$. Moreover, by Lemma~\ref{lem:sub1} we know that there exist a constant $\alpha_{\mathcal{F}}$ (which is computable) and  $\mathcal{F}'\in \mathsf{Conn}(\mathcal{F})$ such that ${\sf opt}(G,\mathcal{F}')\leq {\sf opt}(G,\mathcal{F})+\alpha_{\mathcal{F}}$. First, we test whether ${\sf opt}(G,\mathcal{F}) \leq \frac{5}{\varepsilon} \alpha_{\mathcal{F}}$ and if so, we compute an optimum solution. Towards that, we can design a simple branching algorithm running in time $g(\mathcal{F}) \cdot n$ using  
Theorem~\ref{thm:IsoBoundedExp}, for some function $g$. That is, find a subgraph $G'$ of $G$ isomorphic to a graph in $\mathcal{F}$ and for each $v\in V(G')$, recursively find a solution to $(G-v,\mathcal{F})$ (if the there is a solution of size at most $\frac{5}{\varepsilon} \alpha_{\mathcal{F}}-1$) and output the best solution.  Now on, we assume  that ${\sf opt}(G,\mathcal{F})> \frac{5}{\varepsilon} \alpha_{\mathcal{F}}$. 

Now, by Lemma~\ref{lem:sub1}, we know that there exists $\mathcal{F}'\in \mathsf{Conn}(\mathcal{F})$ such that 
${\sf opt}(G,\mathcal{F}')\leq {\sf opt}(G,\mathcal{F})+\alpha_{\mathcal{F}} \leq (1+\frac{\varepsilon}{5}){\sf opt}(G,\mathcal{F})$. Hence a $(1+\frac{\varepsilon}{2})$-approximate solution to $(G,\mathcal{F}')$ is a $(1+\varepsilon)$-approximate solution to $(G,\mathcal{F}')$.  
Thus, for each $\mathcal{F}_1\in \mathsf{Conn}(\mathcal{F})$, we run the algorithm $\mathcal{A}$ (because $\mathcal{F}_1$ is a set of connected graphs) and output the best solution among them. 

Notice that the number of times  we run $\mathcal{A}$ is 
$|\mathsf{Conn}(\mathcal{F})|$ which is bounded by a function of~$\mathcal{F}$.
Thus, the total running time of our algorithm is $f(\varepsilon,\mathcal{F}) \cdot n^c$ for some function $f$. 
\end{proof}

\subsection{Algorithm for connected forbidden patterns}

Based on the discussion in the previous section, it suffices to design an $f(\varepsilon,\mathcal{F}) \cdot n^c$-time approximation scheme for \textsc{(Induced) Subgraph Hitting} on $\mathcal{G}$ with connected forbidden patterns, provided an approximation scheme for the same problem running in $f_0(\varepsilon,\mathcal{F},\Delta) \cdot n^c$ time (for convenience, we call it the \textit{degree-sensitive} approximation scheme).
We denoted by $\textsc{Hitting}'$ the provided degree-sensitive approximation scheme, where $\textsc{Hitting}'(G,\mathcal{F},\varepsilon)$ returns a $(1+\varepsilon)$-approximation solution for the \textsc{(Induced) Subgraph Hitting} instance $(G,\mathcal{F})$.

Let $(G,\mathcal{F})$ be the input where $G \in \mathcal{G}$, and $\varepsilon > 0$ be the approximation ratio.
Suppose all graphs in $\mathcal{F}$ are connected.
Set $n = |V(G)|$ and $\gamma = \max_{F \in \mathcal{F}} |V(F)|$.
Our approximation schemes for \textsc{Subgraph Hitting} and \textsc{Induced Subgraph Hitting} are essentially identical, so we show them together in Algorithm~\ref{alg-connected}.
At the beginning, we define $\mathcal{V} \coloneqq \mathcal{V}_\mathcal{F}(G)$ for \textsc{Subgraph Hitting} and $\mathcal{V} \coloneqq \mathcal{V}_{\mathcal{F}}^\mathsf{in}(G)$ for \textsc{Induced Subgraph Hitting} (line~1).
After this, every step is the same for \textsc{Subgraph Hitting} and \textsc{Induced Subgraph Hitting} (the underlying problem of the two sub-routines $\textsc{Hitting}'$ and $\textsc{Hitting}''$ is always the same as the problem we are solving).

Let $\delta = O_{\varepsilon,\mathcal{F}}(1)$ and $\delta' = O_{\varepsilon,\mathcal{F}}(1)$ be two parameters to be determined later.
Besides defining $\mathcal{V}$, we also define a collection $\varGamma_H$ and a set $R_H$ for every induced subgraph $H$ of $G$, where $\varGamma_H$ consists of all nonempty subsets $X \subseteq V(H)$ that is the core of a $\delta$-sunflower in $\mathcal{V} \Cap V(H)$ and $R_H$ consists of all vertices $v \in V(H)$ such that every $V \in \mathcal{V} \Cap V(H)$ with $v \in V$ contains some set in $\varGamma_{H - \{v\}}$.
Note that in line~1 we only define (conceptually) these sets instead of computing them; in fact, we cannot afford to compute them and also do not need to.
In line~2-5, we construct an induced subgraph $G_1$ of $G$.
Initially, we set $G_1 = G$.
Whenever $R_{G_1} \neq \emptyset$, we arbitrarily pick $v \in R_{G_1}$ and remove it from $G_1$.
We keep doing this until $R_{G_1} = \emptyset$.
Note that $G_1$ is changing during this procedure and hence $R_{G_1}$ is also changing accordingly.
Next, in line~6-7, we in turn construct an induced subgraph $G_2$ of $G_1$.
This construction is simply removing all large-degree vertices from $G_1$, using $\delta'$ as the threshold.
Specifically, we define $G_2 = G_1 - V^*$ where $V^* \subseteq V(G_1)$ consists of all vertices in $G_1$ of degree at least $\delta'$.

\begin{algorithm}[h]
    \caption{\textsc{Hitting}$(G,\mathcal{F},\varepsilon)$}
	\begin{algorithmic}[1]
	    \State $\mathcal{V} \coloneqq \mathcal{V}_{\mathcal{F}}(G)$ for \textsc{Subgraph Hitting} and $\mathcal{V} \coloneqq \mathcal{V}_{\mathcal{F}}^\mathsf{in}(G)$ for \textsc{Induced Subgraph Hitting}
	    
	    \noindent
	    $\varGamma_H \coloneqq \{X \subseteq V(H): X \neq \emptyset \text{ is the core of a } \delta \text{-sunflower in } \mathcal{V} \Cap V(H)\}$ for $H \subseteq_\mathsf{in} G$ 
	    
	    \noindent
	    $R_H \coloneqq \{v \in V(H): \text{every } V \in \mathcal{V} \Cap V(H) \text{ with } v \in V \text{ contains some } X \in \varGamma_{H - \{v\}} \}$ for $H \subseteq_\mathsf{in} G$
            \smallskip
	    \State $G_1 \leftarrow G$
	    \While{$R_{G_1} \neq \emptyset$}
            \State $v \leftarrow$ an arbitrary vertex in $R_{G_1}$
            \State $G_1 \leftarrow G_1-\{v\}$
	    \EndWhile
	    \State $V^* \leftarrow \{v \in V(G_1): \mathsf{deg}_{G_1}(v) \geq \delta'\}$
	    \State $G_2 \leftarrow G_1 - V^*$
	    \smallskip
	    \State $S_2 \leftarrow \textsc{Hitting}'(G_2,\mathcal{F},\varepsilon/4)$ \Comment{a solution of $(G_2,\mathcal{F})$}
	    \State $S_1 \leftarrow S_2 \cup V^*$ \Comment{a solution of $(G_1,\mathcal{F})$}
            \State $\sigma \leftarrow \textsc{Order}(G,\gamma)$ \Comment{applying Theorem~\ref{thm-ordering}}
            \State $S \leftarrow S_1$
            \While{$\exists$ $v \in V(G) \backslash S$ such that $\rho(v,S) \geq \delta - \mathsf{wcol}_\gamma(G,\sigma)$}
                \State $S \leftarrow S \cup \{v\}$
            \EndWhile
	    \State \textbf{return} $S$ \Comment{a solution of $(G,\mathcal{F})$}
	\end{algorithmic}
	\label{alg-connected}
\end{algorithm}

The rest of the algorithm (lines 8-14) is dedicated to computing a solution for the instance $(G,\mathcal{F})$.
This is done backwards: we first compute a solution $S_2 \subseteq V(G_2)$ for $(G_2,\mathcal{F})$, then use it to obtain a solution $S_1 \subseteq V(G_1)$ for $(G_1,\mathcal{F})$ and in turn obtain the desired solution $S \subseteq V(G)$ for $(G,\mathcal{F})$.
Note that the maximum degree of $G_2$ is bounded by $\delta' = O_{\varepsilon,\mathcal{F}}(1)$.
Furthermore, $G_2 \in \mathcal{G}$, as $\mathcal{G}$ is hereditary.
Thus, we can use the degree-sensitive approximation scheme $\textsc{Hitting}'$ to efficiently compute $S_2$ (line~8).
Here we choose the approximation ratio of $S_2$ to be $1+\frac{\varepsilon}{4}$.
Then we simply define $S_1 = S_2 \cap V^*$ (line~9).
Since $G_1 - S_1 = G_2 - S_2$, $S_1$ is a solution for $(G_1,\mathcal{F})$.
To further compute the solution $S$ for $(G,\mathcal{F})$ is more complicated.
We first apply the algorithm of Theorem~\ref{thm-ordering} on $G$ and $r = \gamma$ to obtain an ordering $\sigma$ of the vertices of $G$ (line~10).
For a vertex $v \in V(G)$ and a set $A \subseteq V(G)$, we write $\rho(v,A) = |\{u \in A: v \in \text{WR}_\gamma(G,\sigma,u)\}|$, i.e., the number of vertices $u \in A$ such that $v$ is weakly reachable from $u$ under $\sigma$ within distance $\gamma$.
Then we construct $S$ iteratively as follows.
Initially, set $S = S_1$ (line~11).
Whenever there exists a vertex $v \in V(G) \backslash S$ such that $\rho(v,S) \geq \delta - \mathsf{wcol}_\gamma(G,\sigma)$, we add it to $S$.
We keep doing this until we cannot find such a vertex $v$.
Note that $S$ is changing during this procedure, and hence $\rho(v,S)$ might also change.
When the procedure terminates, we return $S$ (line~14).
Later we will see that $S$ is a solution for $(G,\mathcal{F})$.

In the following sections, we analyze the correctness of Algorithm~\ref{alg-connected} and show how to implement it with the desired running time.

\subsection{Analysis} \label{sec-correct}
In this section, we analyze the correctness of Algorithm~\ref{alg-connected}.
We shall show that if the parameters $\delta$ and $\delta'$ are chosen properly, then the set $S$ returned by Algorithm~\ref{alg-connected} is a $(1+\varepsilon)$-approximation solution for the instance $(G,\mathcal{F})$.
For convenience, for every induced subgraph $H$ of $G$, we write $\mathsf{opt}(H)$ as the minimum size of a hitting set of $\mathcal{V} \Cap V(H)$.
Clearly, $\mathsf{opt}(H)$ is the optimum for the instance $(H,\mathcal{F})$.
Throughout the analysis, we only one basic properties of $\mathcal{V}$ (which hold no matter whether $\mathcal{V} = \mathcal{V}_\mathcal{F}(G)$ or $\mathcal{V} = \mathcal{V}_\mathcal{F}^\mathsf{in}(G)$): $G[V]$ is connected for all $V \in \mathcal{V}$.
As such, we do not need to distinguish the cases $\mathcal{V} = \mathcal{V}_\mathcal{F}(G)$ and $\mathcal{V} = \mathcal{V}_\mathcal{F}^\mathsf{in}(G)$, and can do the analysis for \textsc{Subgraph Hitting} and \textsc{Induced Subgraph Hitting} simultaneously.

\subsubsection{From $S_1$ to $S$}

We first show that if the parameter $\delta$ is chosen properly and $S_1$ is a $(1+\frac{\varepsilon}{2})$-approximation solution of $(G_1,\mathcal{F})$, then $S$ is a $(1+\varepsilon)$-approximation solution of $(G,\mathcal{F})$.
The following lemma gives a bound on the size of $S$.
\begin{lemma} \label{lem-S1'size}
    $|S \backslash S_1| \leq \frac{\mathsf{wcol}_\gamma(G,\sigma)}{\delta - \mathsf{wcol}_\gamma(G,\sigma)} \cdot |S|$ and thus $|S| \leq \left(1+\frac{\mathsf{wcol}_\gamma(G,\sigma)}{\delta - 2\mathsf{wcol}_\gamma(G,\sigma)}\right) \cdot |S_1|$.
\end{lemma}
\begin{proof}
    By the definition of weakly reachability, $|\text{WR}_\gamma(G,\sigma,v)| \leq \mathsf{wcol}_\gamma(G,\sigma)$ for all $v \in V(G)$.
    For every vertex $v \in V(G)$, $\rho(v,S)$ is equal to the number of pairs $(u,v)$ such that $u \in S$ and $v \in \text{WR}_\gamma(G,\sigma,u)$.
    Therefore,
    \begin{equation*}
        \sum_{v \in V(G)} \rho(v,S) = \sum_{u \in S} |\text{WR}_\gamma(G,\sigma,u)| \leq \mathsf{wcol}_\gamma(G,\sigma) \cdot |S|.
    \end{equation*}
    In particular, $\sum_{v \in S \backslash S_1} \rho(v,S) \leq \mathsf{wcol}_\gamma(G,\sigma) \cdot |S|$.
    For every $v \in S \backslash S_1$, $\rho(v,S) \geq \delta - \mathsf{wcol}_\gamma(G,\sigma)$ when we add $v$ to $S$ in line~13, and thus we also have $\rho(v,S) \geq \delta - \mathsf{wcol}_\gamma(G,\sigma)$ after the entire while-loop (line~12-13).
    Therefore, $\sum_{v \in S \backslash S_1} \rho(v,S) \geq (\delta - \mathsf{wcol}_\gamma(G,\sigma)) \cdot |S \backslash S_1|$.
    It follows that 
    \begin{equation*}
        (\delta - \mathsf{wcol}_\gamma(G,\sigma)) \cdot |S \backslash S_1| \leq \sum_{v \in S \backslash S_1} \rho(v,S) \leq \mathsf{wcol}_\gamma(G,\sigma) \cdot |S|,
    \end{equation*}
    and hence $|S \backslash S_1| \leq \frac{\mathsf{wcol}_\gamma(G,\sigma)}{\delta - \mathsf{wcol}_\gamma(G,\sigma)} \cdot |S|$.
    This further implies that $|S_1| \geq \left(1-\frac{\mathsf{wcol}_\gamma(G,\sigma)}{\delta - \mathsf{wcol}_\gamma(G,\sigma)}\right) \cdot |S|$.
    So we finally have $|S| \leq \left(1+\frac{\mathsf{wcol}_\gamma(G,\sigma)}{\delta - 2\mathsf{wcol}_\gamma(G,\sigma)}\right) \cdot |S_1|$.
\end{proof}

Next, we show that $S$ is truly a solution of $(G,\mathcal{F})$.
Equivalently, we need to show that $S$ hits all sets in $\mathcal{V}$.
Suppose the while-loop in line~3-5 has $r$ iterations, and let $v_i$ denote the vertex removed from $G_1$ in the $i$-th iteration for $i \in [r]$.
Then $G_1 = G - \{v_1,\dots,v_r\}$.
Define $H_i = G - \{v_1,\dots,v_{r-i}\}$ for $i \in \{0\} \cup [r]$.
Note that $H_r = G$ and $H_0 = G_1$.
\begin{lemma} \label{lem-isasol}
    Suppose $S_1$ is a solution for $(G_1,\mathcal{F})$.
    Then for every $i \in \{0\} \cup [r]$, $S$ hits all sets in $\mathcal{V} \Cap V(H_i)$.
    In particular, $S$ hits all sets in $\mathcal{V}$.
\end{lemma}
\begin{proof}
We apply induction on $i$.
When $i = 0$, the statement trivially holds, since $H_0 = G_1$ and the sets in $\mathcal{V} \Cap V(G_1)$ are all hit by $S_1$ (and thus by $S$) according to our assumption.
Suppose the statement holds for $i-1$, and we shall show it also holds for $i$.
We first show that $S$ hits all sets in $\varGamma_{H_{i-1}}$.
Consider a set $X \in \varGamma_{H_{i-1}}$, and assume $X \cap S = \emptyset$ in order to deduce a contradiction.
By definition, $X$ is the core of a sunflower $V_1,\dots,V_\delta \in \mathcal{V} \Cap V(H_{i-1})$.
Let $v^* \in X$ be the largest vertex under the ordering $\sigma$, i.e., $v^* \geq_\sigma v$ for all $v \in X$.
Define $I$ as the set of all indices $i \in [\delta]$ such that all vertices in $V_i$ are smaller than or equal to $v^*$ under the ordering $\sigma$.
For each $i \in [\delta]$, let $v_i \in V_i$ denote the largest vertex under $\sigma$.
Note that the induced subgraph $G[V_i]$ is connected, as $V_i \in \mathcal{V}$.
So there exists a path $\pi_i$ in $G[V_i]$ of length at most $|V_i| \leq \gamma$ which connects $v^*$ and $v_i$.
Since $v_i$ is the largest vertex in $V_i$, it is also the largest vertex on $\pi_i$, which implies $v_i \in \text{WR}_\gamma(G,\sigma,v^*)$.
By the definition of $I$, we have $v_i >_\sigma v^*$ for all $i \in [\delta] \backslash I$ and thus $v_i \in V_i \backslash X$, since $v^*$ is the largest vertex in $X$ under the ordering $\sigma$.
Therefore, the vertices $v_i$ for $i \in [\delta] \backslash I$ are distinct, and there are $\delta-|I|$ such vertices.
Using the fact that all these vertices are contained in $\text{WR}_\gamma(G,\sigma,v^*)$, we have
\begin{equation*}
    \delta - |I| \leq |\text{WR}_\gamma(G,\sigma,v^*)| \leq \mathsf{wcol}_\gamma(G,\sigma).
\end{equation*}
Equivalently, $|I| \geq \delta - \mathsf{wcol}_\gamma(G,\sigma)$.
As $X \cap S = \emptyset$, we must have $(V_i \backslash X) \cap S \neq \emptyset$ for all $i \in [\delta]$, because $S$ is a hitting set of $\mathcal{V} \Cap V(H_{i-1})$ and thus hits $V_1,\dots,V_\delta$.
For each $i \in [\delta]$, we pick a vertex $p_i \in (V_i \backslash X) \cap S$.
Note that the vertices $p_1,\dots,p_\delta$ are distinct, since $V_1 \backslash X,\dots,V_\delta \backslash X$ are disjoint.
We claim that $v^* \in \text{WR}_\gamma(G,\sigma,p_i)$ for all $i \in I$.
Indeed, if $i \in I$, then $v^*$ is the largest vertex in $V_i$ under $\sigma$.
Since $G[V_i]$ is connected, there is a path in $G[V_i]$ of length at most $|V_i| \leq \gamma$ which connects $p_i$ and $v^*$.
On this path, $v^*$ is the largest vertex, which implies $v^* \in \text{WR}_\gamma(G,\sigma,p_i)$.
It follows that $\rho(v^*,S) \geq |I|$, as the vertices $p_1,\dots,p_\delta$ are distinct.
Recall that $|I| \geq \delta-\mathsf{wcol}_\gamma(G,\sigma)$.
Therefore, $\rho(v^*,S) \geq \delta - \mathsf{wcol}_\gamma(G,\sigma)$.
This implies $v^* \in S$, for otherwise $v^*$ will be added to $S$ during the while-loop in line~12-13.
However, this contradicts with our assumption $X \cap S = \emptyset$.
So we conclude that $S$ hits all sets in $\varGamma_{H_{i-1}}$.

To complete the induction argument, we need to further show that $S$ hits all sets in $\mathcal{V} \Cap V(H_i)$.
Set $v = v_{r-i+1}$.
Then $H_{i-1} = H_i - \{v\}$.
Consider a set $V \in \mathcal{V} \Cap V(H_i)$.
If $v \notin V$, then $V \in \mathcal{V} \Cap V(H_{i-1})$ and $V$ is hit by $S$ according to our induction hypothesis.
So assume $v \in V$.
Note that $v \in R_{H_i}$, since in the while-loop in line~3-5 we removed $v$ from $H_i$ to obtain $H_{i-1}$.
By the definition of $R_{H_i}$ and the fact $v \in V$, there exists $X \in \varGamma_{H_i - \{v\}} = \varGamma_{H_{i-1}}$ such that $X \subseteq V$.
As shown above, $X \cap S \neq \emptyset$, which implies $V \cap S \neq \emptyset$.
Therefore, $S$ hits all sets in $\mathcal{V} \Cap V(H_i)$.
Finally, apply the statement with $i = r$, we conclude that $S$ hits all sets in $\mathcal{V} \Cap V(H_r) = \mathcal{V}$.
\end{proof}

\begin{corollary} \label{cor-S1toS}
If $\delta \geq \frac{2+3\varepsilon}{\varepsilon} \cdot \mathsf{wcol}_\gamma^{\gamma+1}(G)$ and $S_1$ is a $(1+\frac{\varepsilon}{2})$-approximation solution for the instance $(G_1,\mathcal{F})$, then $S$ is a $(1+\varepsilon)$-approximation solution for the instance $(G,\mathcal{F})$.
\end{corollary}
\begin{proof}
By Lemma~\ref{lem-isasol}, $S$ is a hitting set of $\mathcal{V}$ and thus a solution of $(G,\mathcal{F})$.
By Theorem~\ref{thm-ordering}, we have $\mathsf{wcol}_\gamma(G,\sigma) \leq \mathsf{wcol}_\gamma^{\gamma+1}(G)$.
If $\delta \geq \frac{2+3\varepsilon}{\varepsilon} \cdot \mathsf{wcol}_\gamma^{\gamma+1}(G)$, then Lemma~\ref{lem-S1'size} implies
\begin{equation*}
    |S| \leq \left(1+\frac{\mathsf{wcol}_\gamma(G,\sigma)}{\delta - 2\mathsf{wcol}_\gamma(G,\sigma)}\right) \cdot |S_1| \leq \left(1+\frac{\mathsf{wcol}_\gamma^{\gamma+1}(G)}{\delta - 2\mathsf{wcol}_\gamma^{\gamma+1}(G)}\right) \cdot |S_1| \leq \left(1+\frac{\varepsilon}{2+\varepsilon}\right) \cdot |S_1|.
\end{equation*}
If we further have $|S_1| \leq (1+\frac{\varepsilon}{2}) \cdot \mathsf{opt}(G_1)$, then $|S| \leq (1+\varepsilon) \cdot \mathsf{opt}(G_1) \leq (1+\varepsilon) \cdot \mathsf{opt}(G)$.
\end{proof}

\subsubsection{From $S_2$ to $S_1$}

In this section, we show that if the parameters $\delta$ and $\delta'$ are chosen properly, then $|S_1| \leq (1+\frac{\varepsilon}{2}) \cdot \mathsf{opt}(G_1)$, based on the fact that $|S_2| \leq (1+\frac{\varepsilon}{4}) \cdot \mathsf{opt}(G_2)$.
Combining this with Corollary~\ref{cor-S1toS} completes our analysis.
Let $V^* \subseteq V(G_1)$ be the set defined in line~6 of Algorithm~\ref{alg-connected}, which consists of the vertices in $V(G_1)$ with degree at least $\delta'$.
We have $|S_1| = |S_2| + |V^*|$.
Since $|S_2| \leq (1+\frac{\varepsilon}{4}) \cdot \mathsf{opt}(G_2) \leq (1+\frac{\varepsilon}{4}) \cdot \mathsf{opt}(G_1)$, to have $|S_1| \leq (1+\frac{\varepsilon}{2}) \cdot \mathsf{opt}(G_1)$, it suffices to show $|V^*| \leq \frac{\varepsilon}{4} \cdot \mathsf{opt}(G_1)$, i.e., there are not many high-degree vertices in $G_1$.

Our proof relies on two properties of $G_1$: \textbf{(i)} $R_{G_1} = \emptyset$ and \textbf{(ii)} $G_1$ has bounded weakly coloring numbers.
For a vertex $v \in V(G_1)$ and a set $V \in \mathcal{V} \Cap V(G_1)$, we say $V$ is witnessed by $v$ if $v \in V$ or $v$ is neighboring to some vertex in $V$.
\begin{lemma} \label{lem-seedisjoint}
    For an integer $r > \max\{\gamma,\delta\}$, if $\mathsf{deg}_{G_1}(v) \geq r^\gamma \gamma!$, then there exists $r$ disjoint sets in $\mathcal{V} \Cap V(G_1)$ witnessed by $v$.
\end{lemma}
\begin{proof}
For each $u \in V(G_1)$, since $u \notin R_{G_1}$, there exists a set $V_u \in \mathcal{V} \Cap V(G_1)$ such that $u \in V_u$ and $V_u$ does not contain any set in $\varGamma_{G_1-\{u\}}$.
We have $|V_u| \leq \gamma$ for all $u \in V(G_1)$.
Suppose $\mathsf{deg}_{G_1}(v) \geq r^\gamma \gamma!$ or equivalently $|N_{G_1}(v)| \geq r^\gamma \gamma!$, where $r > \max\{\gamma,\delta\}$.
By the sunflower lemma (Lemma~\ref{lem-sunflower}), $\{V_u: u \in N_{G_1}(v)\}$ contains a sunflower of size $r$, i.e., there exists $U \subseteq N_{G_1}(v)$ with $|U| = r$ such that the sets in $\{V_u: u \in U\}$ form a sunflower.
All sets in this sunflower are witnessed by $v$.
We claim that the core $K$ of the sunflower is empty.
Assume $K \neq \emptyset$.
As $r > \gamma$ and $|K| \leq \gamma$, there exists $u \in U$ such that $u \notin K$.
Therefore, for any $u' \in U \backslash \{u\}$, we have $u \notin V_{u'}$ and thus $V_{u'} \in \mathcal{V} \Cap V(G - \{u\})$.
It follows that $\{V_{u'}: u' \in U \backslash \{u\}\}$ is a $(r-1)$-sunflower in $\mathcal{V} \Cap V(G - \{u\})$ with core $K$.
Since $K \neq \emptyset$ and $r-1 \geq \delta$, we have $K \in \varGamma_{G - \{u\}}$.
However, this contradicts the fact that $V_u$ does not contain any set in $\varGamma_{G_1-\{u\}}$.
So we have $K = \emptyset$ and the sets in $\{V_u: u \in U\}$ are disjoint, which implies that $v$ witnesses $r$ disjoint sets in $\mathcal{V} \Cap V(G_1)$.
\end{proof}

\begin{lemma} \label{lem-seesolution}
    Let $S$ be a solution of the instance $(G_1,\mathcal{F})$.
    If a vertex $v \in V(G_1)$ witnesses $r$ disjoint sets in $\mathcal{V} \Cap V(G_1)$, then $\rho(v,S) \geq r - \mathsf{wcol}_\gamma(G)$.
\end{lemma}
\begin{proof}
Suppose $v$ witnesses $V_1,\dots,V_r \in \mathcal{V} \Cap V(G_1)$, which are disjoint.
For each $i \in [r]$ and each $u \in V_i$, there exists a path $\pi_u$ between $v$ and $u$ in $G[V_i \cup \{v\}]$ of length at most $\gamma$, since $G[V_i]$ is connected and $v$ is either in $V_i$ or neighboring to $V_i$.
Let $v_i \in V_i$ be the largest vertex under $\sigma$, and define $I = \{i \in [r]: v_i >_\sigma v\}$.
As $V_1,\dots,V_r$ are disjoint, $v_1,\dots,v_r$ are distinct.
Furthermore, for $i \in I$, the vertex $v_i$ is the largest vertex on $\pi_{v_i}$, which implies that $v_i \in \text{WR}_\gamma(G,\sigma,v)$.
It follows that $|I| \leq |\text{WR}_\gamma(G,\sigma,v)| \leq \mathsf{wcol}_\gamma(G,\sigma) = \mathsf{wcol}_\gamma(G)$.
Thus, $|[r] \backslash I| \geq r- \mathsf{wcol}_\gamma(G)$.
As $S$ is a solution of $(G_1,\mathcal{F})$, we have $S \cap V_i \neq \emptyset$ for all $i \in [r]$.
Let $p_i \in S \cap V_i$ for $i \in [r]$.
Again, as $V_1,\dots,V_r$ are disjoint, $p_1,\dots,p_r$ are distinct.
For $i \in [r] \backslash I$, we have $v \geq_\sigma v_i$ and thus $v \geq_\sigma u$ for all $u \in V_i$.
Therefore, $v$ is the largest vertex on $\pi_{p_i}$ and $v \in \text{WR}_\gamma(G,\sigma,p_i)$, for all $i \in [r] \backslash I$.
It follows that $\rho(v,S) \geq |[r] \backslash I| \geq r- \mathsf{wcol}_\gamma(G)$.
\end{proof}

\begin{corollary} \label{cor-S2toS1}
For any $r > \max\{\gamma,\delta\}$, there exist at most $\frac{\mathsf{wcol}_\gamma(G)}{r - \mathsf{wcol}_\gamma(G)} \cdot \mathsf{opt}(G_1)$ vertices in $G_1$ of degree at least $r^\gamma \gamma!$.
In particular, if $\delta' \geq (\frac{4+\varepsilon}{\varepsilon} \cdot \mathsf{wcol}_\gamma(G) + \max\{\gamma,\delta\})^\gamma \gamma!$, then $|V^*| \leq \frac{\varepsilon}{4} \cdot \mathsf{opt}(G_1)$ and $S_1$ is a $(1+\frac{\varepsilon}{2})$-approximation solution for the instance $(G_1,\mathcal{F})$.
\end{corollary}
\begin{proof}
Let $S^* \subseteq V(G_1)$ be an optimal solution for the instance $(G_1,\mathcal{F})$.
By Lemma~\ref{lem-seedisjoint}, every vertex in $G_1$ of degree at least $r^\gamma \gamma!$ witnesses $r$ disjoint sets in $\mathcal{V} \Cap V(G_1)$.
Further applying Lemma~\ref{lem-seesolution}, we have $\rho(v,S^*) \geq r - \mathsf{wcol}_\gamma(G)$ for all $v \in V(G_1)$ with $\mathsf{deg}_{G_1}(v) \geq r^\gamma \gamma!$.
Note that
\begin{equation*}
    \sum_{v \in V(G_1)} \rho(v,S^*) = \sum_{u \in S^*} |\text{WR}_\gamma(G,\sigma,u)| \leq \mathsf{wcol}_\gamma(G) \cdot |S^*| = \mathsf{wcol}_\gamma(G) \cdot \mathsf{opt}(G_1).
\end{equation*}
By an averaging argument, the number of vertices $v \in V(G_1)$ with $\rho(v,S^*) \geq r - \mathsf{wcol}_\gamma(G)$ is at most $\frac{\mathsf{wcol}_\gamma(G)}{r - \mathsf{wcol}_\gamma(G)} \cdot \mathsf{opt}(G_1)$, which also bounds the number of vertices of degree at least $r^\gamma \gamma!$.
If we set $r = \frac{4+\varepsilon}{\varepsilon} \cdot \mathsf{wcol}_\gamma(G) + \max\{\gamma,\delta\}$ and $\delta' \geq r^\gamma \gamma!$, then $r > \max\{\gamma,\delta\}$ and thus
\begin{equation*}
    |V^*| \leq \frac{\mathsf{wcol}_\gamma(G)}{r - \mathsf{wcol}_\gamma(G)} \cdot \mathsf{opt}(G_1) \leq \frac{\mathsf{wcol}_\gamma(G)}{\frac{4+\varepsilon}{\varepsilon} \cdot \mathsf{wcol}_\gamma(G) + \max\{\gamma,\delta\} - \mathsf{wcol}_\gamma(G)} 
    \cdot \mathsf{opt}(G_1) \leq \frac{\varepsilon}{4} \cdot \mathsf{opt}(G_1).
\end{equation*}
Finally, $|S_1| = |S_2| + |V^*| \leq (1+\frac{\varepsilon}{4}) \cdot \mathsf{opt}(G_2) + \frac{\varepsilon}{4} \cdot \mathsf{opt}(G_1) \leq (1+\frac{\varepsilon}{2}) \cdot \mathsf{opt}(G_1)$.
\end{proof}

\subsection{Linear-time implementation} \label{sec-implement}
In this section, we show how to implement Algorithm~\ref{alg-connected} (except the call of the sub-routine $\textsc{Hitting}'$ in line~8) in linear time, or more precisely, $O_{\varepsilon,\mathcal{F}}(n)$ time.
Note that if this can be done, then the overall time complexity of Algorithm~\ref{alg-connected} is $f(\varepsilon,\mathcal{F}) \cdot n^c$ time for some function $f$, because line~8 takes $f_0(\varepsilon,\mathcal{F},\Delta) \cdot n^c$ time where $\Delta$ is the maximum degree of $G_2$ which is bounded by $\delta' = O_{\varepsilon,\mathcal{F}}(1)$.
We divide Algorithm~\ref{alg-connected} into two parts: the part for computing $G_2$ (line~1-7) and the part for computing $S$ (line~9-13).
We discuss these two parts in the following two sections.

\subsubsection{Computing $G_2$}
Consider line~1-7 of Algorithm~\ref{alg-connected}.
The sets defined in line~1 are not computed explicitly.
It suffices to the while-loop in line~3-5 in linear time.
The idea is to formulate the computation tasks as first-order model-checking problems, and properly apply the efficient data structure of Theorem~\ref{thm:fologic}.

\begin{lemma}\label{lem:fo1}
Whether a vertex $v \in V(G_1)$ is in $R_{G_1}$ can be expressed as a first-order formula $\varphi(v)$ whose length depends only on $\varepsilon$ and $\mathcal{F}$.
Thus, the while-loop in line~3-5 can be implemented in $f(\varepsilon,\mathcal{F}) \cdot n$ time for some function $f$.
\end{lemma}

\begin{proof}
Towards the presentation of the formula  to check the membership of $v$ in $R_{G_1}$, we first present several simpler formulas that will be used as building blocks.
Recall that for two vertices $u,u'\in V(G_1)$, $\mathsf{Adj}(u,u')$ is the atomic formula that is true if and only if $u$ and $u'$ are adjacent in $G_1$.
We start with some simple formulas.
\begin{itemize}
    \item For every $\ell\in\mathbb{N}$, $\mathsf{Distinct}(v_1,v_2,\ldots,v_\ell)$: True iff $v_1,v_2,\ldots,v_\ell$ are distinct. Defined as follows:
    \[\bigwedge_{i,j\in[\ell],i\neq j}v_i\neq v_j.\]
    \item For every $r,\ell\in\mathbb{N}$, $\mathsf{Subset}(u_1,u_2,\ldots,u_r,v_1,v_2,\ldots,v_\ell)$: True iff $\{u_1,u_2,\ldots,u_r\}\subseteq\{v_1,v_2,\ldots,v_\ell\}$. Defined as follows:
    \[\bigwedge_{i\in[r]}\left(\bigvee_{j\in[\ell]}u_i=v_j\right).\]
\end{itemize}

For every $F\in{\cal F}$, we arbitrary order $V(F)$, denoted by $V(F)=\{v^F_1,v^F_2,\ldots,v^F_{|V(F)|}\}$. We proceed to formulas to check isomorphism.
\begin{itemize}
    \item For every $F\in{\cal F}$, $F$-$\mathsf{Isomorphic}(v_1,v_2,\ldots,v_\ell)$ where $\ell=|V(F)|$: True iff $G_1[\{v_1,v_2,\ldots,v_\ell\}]$ contains a spanning subgraph isomorphic to $F$ where the isomorphism maps $v_i$ to $v^F_i$ for every $i\in [\ell]$. Defined as follows: \[\mathsf{Distinct}(v_1,v_2,\ldots,v_\ell)\wedge \left(\bigwedge_{i,j\in [\ell], \{v^F_i,v^F_j\}\in E(F)}\mathsf{Adj}(v_i,v_j)\right).\]
    \item For every $F\in{\cal F}$, $F$-$\mathsf{Isomorphic}^{\mathsf{in}}(v_1,v_2,\ldots,v_\ell)$ where $\ell=|V(F)|$: True iff  $G_1[\{v_1,v_2,\ldots,v_\ell\}]$ is isomorphic to $F$ where the isomorphism maps $v_i$ to $v^F_i$ for every $i\in [\ell]$. Defined as follows: \[\mathsf{Distinct}(v_1,v_2,\ldots,v_\ell)\wedge \left(\bigwedge_{i,j\in [\ell], \{v^F_i,v^F_j\}\in E(F)}\mathsf{Adj}(v_i,v_j)\right)\wedge \left(\bigwedge_{i,j\in [\ell], \{v^F_i,v^F_j\}\notin E(F)}\neg\mathsf{Adj}(v_i,v_j)\right).\]
    \item For every $\ell\in\mathbb{N}$, $\mathsf{BelongsTo}{\cal V}_{\cal F}(v_1,v_2,\ldots,v_\ell)$: True iff  $G_1[\{v_1,v_2,\ldots,v_\ell\}]$ contains a spanning subgraph isomorphic to some graph $F\in{\cal F}$ where the isomorphism maps $v_i$ to $v^F_i$ for every $i\in [\ell]$. Defined as follows: \[\bigvee_{F\in{\cal F}, |V(F)|=\ell}F\mathrm{-}\mathsf{Isomorphic}(v_1,v_2,\ldots,v_\ell).\]
\end{itemize}

Now, we present formulas to check the membership of a set in $\varGamma_{G_1 - \{v\}}$.
\begin{itemize}
    \item For every $\ell\in\mathbb{N}$, $\mathsf{Core}(v,v_1,v_2,\ldots,v_\ell)$: True iff $\{v_1,v_2,\ldots,v_\ell\} \in \varGamma_{G_1 - \{v\}}$. Defined as follows: 
    \[\begin{array}{l}
    \bigvee_{f:[\delta]\rightarrow[\gamma]}(\exists u^1_1,u^1_2,\ldots,u^1_{f(1)},u^2_1,u^2_2,\ldots,u^2_{f(2)},\ldots,u^\delta_1,u^\delta_2,\ldots,u^\delta_{f(\delta)}\\
    \bigwedge_{i \in [\delta]}((u_1^i \neq v) \wedge ((u_2^i \neq v) \wedge \cdots \wedge (u_{f(i)}^i \neq v)))\\
    \wedge\bigwedge_{i\in[\delta]}(\mathsf{BelongsTo}{\cal V}_{\cal F}(u^i_1,u^i_2,\ldots,u^i_{f(i)})\wedge\mathsf{Subset}(v_1,v_2,\ldots,v_\ell,u^i_1,u^i_2,\ldots,u^i_{f(i)}))\\
    \wedge\bigwedge_{i,j\in[\delta],i\neq j,k\in[f(i)]}(\mathsf{Subset}(u^i_k,v_1,v_2,\ldots,v_\ell)\vee \neg\mathsf{Subset}(u^i_k,u^j_1,u^j_2,\ldots,u^j_{f(j)})).
    \end{array}\]
    We remark that, here, we use $\mathsf{Subset}$ to check the membership of an element (the first argument) in a set.
    \item For every $\ell\in\mathbb{N}$, $\mathsf{ContainsCore}(v_1,v_2,\ldots,v_\ell)$: True iff $\{v_1,v_2,\ldots,v_\ell\}$ contains a set in $\varGamma_{G_1 - \{v\}}$ w.r.t. {\sc Subgraph Hitting}. Defined as follows: \[\bigvee_{r\in[\ell]}\left(\exists  u_1,u_2,\ldots,u_r \mathsf{Distinct}(u_1,\ldots,u_r)\wedge \mathsf{Subset}(u_1,\ldots,u_r,v_1,\ldots,v_\ell)\wedge \mathsf{Heavy}(u_1,\ldots,u_r)\right).\]
\end{itemize}

We are now ready to present the formula to check whether $v \in R_{G_1}$.
We start with {\sc Subgraph Hitting}:
\[\bigwedge_{\ell\in[\gamma]}\left(\forall v_1,v_2,\ldots,v_\ell  \bigvee_{i\in[\ell]}(v=v_i)\wedge\left(\mathsf{BelongsTo}{\cal V}_{\cal F}(v_1,v_2,\ldots,v_\ell)\rightarrow \mathsf{ContainsCore}(v_1,v_2,\ldots,v_\ell)\right)\right).\]
For {\sc Induced Subgraph Hitting}, the required formula is identical except that every use of  $F$-$\mathsf{Isomorphic}(v_1,v_2,\ldots,v_\ell)$ is replaced by the use of $F$-$\mathsf{Isomorphic}^{\mathsf{in}}(v_1,v_2,\ldots,v_\ell)$.

We have seen how to construct a first-order formula $\phi(v)$ such that $v \in R_{G_1}$ iff $\phi(v)$ is true.
Set $\varphi = \exists v (\phi(v))$.
We build in $f(|\varphi|) \cdot n$ time the data structure $\mathcal{D}$ of Theorem~\ref{thm:fologic} on the graph $G$ with the first-order formula $\varphi$.
In each iteration of the while-loop in line~3-5, we test whether $G$ satisfies $\varphi$ or not in $f(|\varphi|)$ time.
If not, then $R_{G_1} = \emptyset$ and we are done.
Otherwise, the data structure can find a vertex $v \in V(G_1)$ that satisfies $\phi(v)$.
Then we need to remove $v$ from $G_1$ and proceed to the next iteration.
Note that the data structure $\mathcal{D}$ only supports edge deletions.
Therefore, instead of truly deleting $v$, we delete all edges incident to $v$ from $G_1$.
In this way, $v$ still exists in $G_1$ but becomes an isolated vertex.
This is in fact equivalent to deleting $v$ from $G_1$, because the forbidden patterns in $\mathcal{F}$ are all connected (and thus isolated vertices do not influence anything).
The time cost for updating $\mathcal{D}$ after each edge-deletion is $f(|\varphi|)$.
As each edge can be deleted at most once, the total update time of $\mathcal{D}$ during the while-loop is $f(|\varphi|) \cdot n$ due to the sparsity of $G_1$.
Also, since the while-loop can have at most $n$ iterations, the total query time of $\mathcal{D}$ is also $f(|\varphi|) \cdot n$.
As a result, the while-loop in line~3-5 can be implemented in $f(|\varphi|) \cdot n$ time.
\end{proof}

\subsubsection{Computing $S$}

Consider line~9-14 in Algorithm~\ref{alg-connected}.
By Theorem~\ref{thm-ordering}, line~9 can be done in $O(n)$ time.
So it suffices to show how to compute $\rho_v$ for all $v \in V(G)$ (line~10) in linear time.
The key here is an efficient algorithm for computing the weakly reachable sets.
\begin{lemma}
    Let $G$ be a graph of $n$ vertices and $m$ edges.
    Given an ordering $\sigma$ of $V(G)$ and $r \in [n]$, one can compute in $O(n+rm \cdot \mathsf{wcol}_r(G,\sigma))$ time the sets $\textnormal{WR}_r(G,\sigma,v)$ for all $v \in V(G)$.
\end{lemma}
\begin{proof}
    Without loss of generality, we can assume $G$ is connected and thus $m = \Omega(n)$.
    We shall compute the sets $\textnormal{WR}_i(G,\sigma,v)$ iteratively for $i = 0,1,\dots,r$.
    We have $\textnormal{WR}_0(G,\sigma,v) = \{v\}$ for all $v \in V(G)$.
    Suppose the sets $\textnormal{WR}_{i-1}(G,\sigma,v)$ are already computed for all $v \in V(G)$.
    To compute $\textnormal{WR}_i(G,\sigma,v)$, the key observation is that $\textnormal{WR}_i(G,\sigma,v)$ consists of exactly the vertices in $\bigcup_{u \in N[v]} \textnormal{WR}_{i-1}(G,\sigma,u)$ which are larger than or equal to $v$ under the ordering $\sigma$; here $N[v]$ is the set of neighbors of $v$ in $G$ (including $v$ itself).
    Indeed, if $x \in \textnormal{WR}_i(G,\sigma,v)$, then there exists a path $(u_0=v,u_1,\dots,u_i=x)$ in $G$ on which $x$ is the largest vertex under the ordering $\sigma$.
    We have $u_1 \in N[v]$ and $x \in \textnormal{WR}_{i-1}(G,\sigma,u_1)$ because of the sub-path $(u_1,\dots,u_i)$.
    Thus, $x \in \bigcup_{u \in N[v]} \textnormal{WR}_{i-1}(G,\sigma,u)$ and $x \geq_\sigma v$.
    On the other hand, if $x \in \textnormal{WR}_{i-1}(G,\sigma,u)$ for some $u \in N[v]$ and $x \geq_\sigma v$, then there exists a path $(u_1 = u,\dots,u_i=x)$ in $G$ on which $x$ is the largest vertex.
    In this case, $(v,u_1,\dots,u_i)$ is a path between $v$ and $x$ on which $x$ is the largest vertex, and thus $x \in \textnormal{WR}_i(G,\sigma,v)$.
    With this observation, to compute $\textnormal{WR}_i(G,\sigma,v)$, we only need to check all vertices in $\bigcup_{u \in N[v]} \textnormal{WR}_{i-1}(G,\sigma,u)$ and take the ones larger than $v$, which can be done in $O(|N[v]| \cdot \mathsf{wcol}_r(G,\sigma))$ time because $|\textnormal{WR}_{i-1}(G,\sigma,u)| \leq \mathsf{wcol}_{i-1}(G,\sigma) \leq \mathsf{wcol}_r(G,\sigma)$.
    Since $\sum_{v \in V(G)} |N[v]| = O(m)$, computing $\textnormal{WR}_i(G,\sigma,v)$ for all $v \in V(G)$ takes $O(m \cdot \mathsf{wcol}_r(G,\sigma))$ time.
    We need $r$ iterations to compute $\textnormal{WR}_r(G,\sigma,v)$, so the overall running time is $O(rm \cdot \mathsf{wcol}_r(G,\sigma))$.
\end{proof}

To efficiently implement the while-loop in line~12-13, we first apply the above lemma to compute $\textnormal{WR}_\gamma(G,\sigma,u)$ for all $u \in V(G)$.
For every $v \in V(G)$, we maintain a number $\rho_v$ which is equal to $\rho(v,S)$ for the current $S$.
To compute the initial values of $\rho_v$, we consider every vertex $u \in S$ and increase $\rho_v$ by 1 for all $v \in \textnormal{WR}_\gamma(G,\sigma,u)$.
By Theorem~\ref{thm-ordering}, $\mathsf{wcol}_\gamma(G,\sigma) \leq \mathsf{wcol}_\gamma^{\gamma+1}(G)$.
Furthermore, by Theorem~\ref{thm-weakcolor}, $\mathsf{wcol}_\gamma(G)$ is bounded by $\chi_\mathcal{G}(\gamma)$ for a function $\chi_\mathcal{G}$, and thus $\mathsf{wcol}_\gamma(G,\sigma)$ is also bounded by a function of $\gamma$.
So computing the initial values of $\rho_v$ takes $f(\gamma) \cdot n$ time.
During the while-loop, we maintain a queue $Q$ that contains the vertices in $V(G) \backslash S$ whose current $rho$ values are at least $\delta - \mathsf{wcol}_\gamma(G,\sigma)$.
In each iteration, we pop a vertex $u$ from $Q$ and add it to $S$.
As $S$ changes, we need to update the values of $\rho_v$, which is done by considering every $v \in \textnormal{WR}_\gamma(G,\sigma,u)$ and increasing $\rho_v$ by 1.
In this procedure, if $\rho_v \geq \delta - \mathsf{wcol}_\gamma(G,\sigma)$ after the update and $v \notin S$, we add it to $Q$.
The while-loop terminates when $Q$ becomes empty.
Clearly, there are at most $n$ iterations and each iteration takes $f(\gamma)$ time.
Therefore, the entire while-loop can be implemented in $f(\gamma) \cdot n$ time.

\subsection{Putting everything together} \label{sec-together}

We pick a sufficiently large $\delta$ that satisfies the condition in Corollary~\ref{cor-S1toS} and a sufficiently large $\delta'$ that satisfies the condition in Corollary~\ref{cor-S2toS1}.
As the weak coloring numbers of $G$ and $G_1$ are all bounded, we can guarantee that $\delta = O_{\varepsilon,\mathcal{F}}(1)$ and $\delta' = O_{\varepsilon,\mathcal{F}}(1)$.
Here a small technical issue is that we cannot compute the weak coloring numbers of $G$ and $G_1$ efficiently.
However, we can use the \textit{admissibility} numbers to approximate the weak coloring numbers.
The $r$-\textit{admissibility} number of $G$, $\mathsf{adm}_r(G)$, satisfies the inequality $\mathsf{adm}_r(G) \leq \mathsf{wcol}_r(G) \leq \mathsf{adm}_r^{r+1}(G)$ \cite{dvovrak2013constant}.
Furthermore, the $r$-admissibility number of bounded-expansion graphs can be computed in $O(n)$ time for any $r \geq 0$ \cite{dvovrak2013constant}.
As such, when picking $\delta$ and $\delta'$, we can simply use the upper bound of the weak coloring numbers provided by admissibility numbers to ensure that $\delta$ and $\delta'$ are sufficiently large while on the other hand still bounded by $O_{\varepsilon,\mathcal{F}}(1)$.
(We omit the definition of admissibility numbers, as it is not used in the paper.)

Next, we consider the overall running time of Algorithm~\ref{alg-connected}.
As shown in Section~\ref{sec-implement}, all steps except line~8 can be done in $f(\varepsilon,\mathcal{F}) \cdot n$ time.
Since $G_2$ is of maximum degree at most $\delta' = O_{\varepsilon,\mathcal{F}}(1)$, line~8 takes $f(\varepsilon,\mathcal{F}) \cdot n^c$ time.
As such, the time complexity of Algorithm~\ref{alg-connected} is $f(\varepsilon,\mathcal{F}) \cdot n^c$.
Combining this with the reduction in Section~\ref{sec-toconnect} (Corollary~\ref{cor-toconnect}) completes the proof of Theorem~\ref{thm-main}.

\paragraph{Variants.}
When the provided degree-sensitive algorithm in Theorem~\ref{thm-main} has a different running time or a different approximation factor, the performance of the algorithm we obtain also changes.
Therefore, the same proof can also result in variants of Theorem~\ref{thm-main}.
For example, when the provided algorithm has an approximation factor $c+\varepsilon$, then we also obtain a $(c+\varepsilon)$-approximation algorithm.
Also, when the provided algorithm has running time $n^{f_0(\varepsilon,\mathcal{F},\Delta)}$, $f_0(\varepsilon,\mathcal{F},\Delta) \cdot n^{g(\varepsilon)}$, and $f_0(\varepsilon,\mathcal{F},\Delta) \cdot n^{g(\mathcal{F})}$, the obtained algorithm has running time $n^{f(\varepsilon,\mathcal{F})}$, $f(\varepsilon,\mathcal{F}) \cdot n^{g(\varepsilon)}$, and $f(\varepsilon,\mathcal{F}) \cdot n^{g(\mathcal{F})}$, respectively.
We believe that these variants can also find their applications in the future.

\subsection{Lossy kernels}

In this section, we apply Theorem~\ref{thm-main} (more specifically, our degree reduction in Algorithm~\ref{alg-connected}) to prove the lossy kernelization results, i.e., Theorem~\ref{thm-kernel}.

\kernel*

Let $\mathcal{G}$ be a hereditary graph class of bounded expansion and $\mathcal{F}$ be a fixed finite set of connected graphs.
Consider an instance $(G,\mathcal{F})$ of \textsc{(Induced) Subgaph Hitting} where $G \in \mathcal{G}$.
Applying Algorithm~\ref{alg-connected} with any $\varepsilon > 0$, we can compute in $O_\varepsilon(n)$ time an instance $(G',\mathcal{F})$ such that $G'$ is an induced subgraph of $G$ of degree $O_\varepsilon(1)$, and given a $c$-approximation solution for $(G',\mathcal{F})$ one can compute in $O_\varepsilon(n)$ time a $(1+\varepsilon) c$-approximation solution for $(G,\mathcal{F})$.

To further obtain a lossy kernel of linear size, we simply remove all \textit{irrelevant} vertices from $G'$.
We say a vertex $v \in V(G')$ is \textit{irrelevant} if it is not contained in any set in $\mathcal{V} \Cap V(G')$.
Here the definition of $\mathcal{V}$ is the same as in Algorithm~\ref{alg-connected}.
Let $I \subseteq V(G')$ be the set of irrelevant vertices in $G'$, and $G'' = G' - I$.
Note that a $c$-approximation solution for the instance $(G'',\mathcal{F})$ is also a $c$-approximation solution for $(G',\mathcal{F})$, simply because $\mathcal{V} \Cap V(G'') = \mathcal{V} \Cap V(G')$.
Therefore, given a $c$-approximation solution for $(G'',\mathcal{F})$, one can compute in $O_\varepsilon(n)$ time a $(1+\varepsilon) c$-approximation solution for $(G,\mathcal{F})$.
To compute $I$ and $G''$ in linear time, we can use the data structure of Theorem~\ref{thm:fologic}, similarly to that in Section~\ref{sec-implement}.
Recall that in the proof of Lemma~\ref{lem:fo1}, we defined the first-order formula $\mathsf{BelongsTo}{\cal V}_{\cal F}(v_1,v_2,\ldots,v_\ell)$ that can check whether $G'[\{v_1,\dots,v_\ell\}]$ contains a spanning subgraph isomorphic to some $F \in \mathcal{F}$.
Now we define
\begin{equation*}
    \phi(v) = \bigvee_{\ell \in [\gamma]} \left( \exists v_1,\dots,v_\ell \left(\bigvee_{i \in [\ell]} v = v_i\right) \wedge \mathsf{BelongsTo}{\cal V}_{\cal F}(v_1,v_2,\ldots,v_\ell) \right),
\end{equation*}
which is true iff $v \notin I$.
Let $Q = V(G')$ be a unary relation on $V(G')$, and set $\varphi = \exists v (\neg \phi(v) \wedge Q(v))$.
We build the data structure $\mathcal{D}$ of Theorem~\ref{thm:fologic} on the graph $G'$ with the unary relation $Q$ and the formula $\varphi$.
We then delete all elements from $Q$ and keep updating $\mathcal{D}$, which takes $f(|\varphi|) \cdot n$ time.
To compute $I$, we iteratively consider the vertex of $G'$.
For a vertex $v \in V(G')$, to check whether $v \in I$, we add back $v$ to $Q$ and check whether $(G',Q)$ satisfies $\varphi$ (and delete $v$ from $Q$ afterwards).
In this way, we can test whether a vertex is in $I$ in $f(|\varphi|)$ time, and hence compute $I$ in $f(|\varphi|) \cdot n$ time.
Finally, we observe the following property of $G''$.
\begin{lemma} \label{lem-largeopt}
    $|V(G'')| \leq f(\varepsilon) \cdot \mathsf{opt}(G'',\mathcal{F})$ for some function $f$.
\end{lemma}
\begin{proof}
It suffices to prove that there are $\Omega_\varepsilon(n)$ vertex disjoint sets of $\mathcal{V} \Cap V(G'')$.
To this end, we show that there is a subset $Z \subseteq V(G'')$ of size $\Omega_\varepsilon(n)$ such that for any pair $(u,v)$ of distinct vertices in $Z$, the distance between them in $G''$ is at least $2\gamma$.
We construct $Z$ as follows.
Initially, we set $Z=\emptyset$.
Then we pick a vertex $v_1\in V(G'')$, add to $Z$, and mark all the vertices which are at a distance at most $2\gamma$ from $v_1$.
There are only $O_{\varepsilon}(1)$ such vertices since the degree of $G''$ is $O_{\varepsilon}(1)$.
In the next round, we further pick an unmarked vertex $v_2$, add to $Z$ and again mark all the vertices which are at a distance at most $2\gamma$ from $v_2$.
We continue this process until all the vertices are marked.
Since in each step we mark at most $O_{\varepsilon}(1)$ vertices, the number of vertices in $Z$ will be $\Omega_\varepsilon(n)$.
Note that every vertex of $G''$ is contained in some set in $\mathcal{V} \Cap V(G'')$, by our construction.
For each $z \in Z$, pick a set $V_z \in \mathcal{V} \Cap V(G'')$ such that $z \in V_z$.
We have $G[V_z]$ is connected, since the graphs in $\mathcal{F}$ are all connected.
It follows that the sets $V_z$ are pairwise disjoint, since the pairwise distances among the vertices in $Z$ are at least $2 \gamma$.
\end{proof}

Note that $\mathsf{opt}(G'',\mathcal{F}) \leq \mathsf{opt}(G,\mathcal{F})$, because $G''$ is an induced subgraph of $G$.
Given a parameterized instance $(G,k)$ for \textsc{(Induced) Subgaph Hitting} with forbidden list $\mathcal{F}$, our kernelization algorithm simply constructs the graph $G''$ as above.
If $|V(G'')| \leq f(\varepsilon) \cdot k$ for the function $f$ in Lemma~\ref{lem-largeopt}, then we output the parameterized instance $(G'',k)$, whose size is linear in $k$.
Otherwise, we know that $k < \mathsf{opt}(G'',\mathcal{F}) \leq \mathsf{opt}(G,\mathcal{F})$ and thus $(G,k)$ is a NO instance.
The running time of the kernelization algorithm as well as the time cost for turning a solution of $(G'',k)$ to a solution of $(G,k)$ is $g(\varepsilon) \cdot n$ for some function $g$.
This proves Theorem~\ref{thm-kernel}.

\section{Local search}\label{sec-local}

In this section, we prove and further discuss our results on the local search algorithms in classes with polynomial expansion.
Let us start with some preparatory work.

\subsection{Tools}

A collection $\mathcal{C}$ of subsets of vertices of a graph $G$ is \emph{$(\omega, t)$-shallow} if every vertex of $G$ appears in
at most $\omega$ elements of $\mathcal{C}$ and $G[C]$ has radius at most $t$ for every $C\in \mathcal{C}$.  The \emph{packing graph} $G[\mathcal{C}]$
is defined as the graph with vertex set $\mathcal{C}$ where $C_1$ and $C_2$ are adjacent if there exist $v_1\in C_1$ and $v_2\in C_2$
such that $v_1=v_2$ or $v_1v_2\in E(G)$.  In particular, an $r$-shallow minor is the packing graph of a $(1,r)$-shallow collection.
\begin{theorem}[Har-Peled and Quanrud~\cite{har2017approximation}]
Let $G$ be a graph with expansion bounded by a function $f$.  For every $(\omega, t)$-shallow collection $\mathcal{C}$ of subsets of vertices of $G$,
the graph $G[C]$ has expansion bounded by the function
$$f'(r)=5\omega^2(2t+1)^2(2r + 1)^2f((2t+1)r+t).$$
\end{theorem}

\begin{corollary}\label{cor-preserve}
For every class $\mathcal{G}$ with polynomial expansion and all non-negative integers $\omega$ and $t$, there exists
a class $\mathcal{G}'$ with polynomial expansion such that the packing graph of any $(\omega, t)$-shallow collection of subsets of vertices
of a graph $G\in \mathcal{G}$ belongs to $\mathcal{G}'$.
\end{corollary}

Next, we need to introduce an important notion called \textit{strongly sublinear separator}.

\begin{definition}
Let $G$ be an $n$-vertex graph. A vertex subset $X\subseteq V(G)$ is called a {\em balanced separator} if each component of $G-X$ has at most $2n/3$ vertices. Let $s(G)$ denote the minimum size of a balanced separator in $G$.  For a graph class $\mathcal{G}$, let $s_{\mathcal{G}}~\colon~{\mathbb N}_{0} \rightarrow {\mathbb N}_0$ be the function defined by  $s_{\mathcal{G}}(n)=\max \{s(H)~\colon~ H \subseteq G \in \mathcal{G}, |V(H)|\leq n\}$.
The graph class $\mathcal{G}$ has {\em strongly sublinear separators} if $s_{\mathcal{G}}(n)=O(n^\eta)$ for some $\eta<1$.  
\end{definition}

The following theorem gives a nice characterization of polynomial-expansion graphs.

\begin{theorem}[\cite{DorakN2016}] \label{thm-polyexp=sep}
A graph class is of polynomial expansion iff it has strongly sublinear separators. 
\end{theorem}

A well-known fact dating back to the work of Lipton and Tarjan~\cite{LiptonTargenSep} is that in any graph $G$ from a class $\mathcal{G}$ with strongly sublinear separators (satisfying $s_{\mathcal{G}}(n)=O(n^\eta)$ for some $\eta<1$),
we can for every $\varepsilon>0$ delete at most $\varepsilon|V(G)|$ vertices to break up $G$ to components of size $O\bigl(((1/\varepsilon)^\frac{1}{1-\eta}\bigr)$.
We need a more technical form of this statement, proved in~\cite{HarPeledQ16}.
Let $\mathcal{K}$ be a system of subsets of vertices of a graph $G$.  We say that $\mathcal{K}$ is a \emph{cover} of $G$ if $\bigcup_{K\in\mathcal{K}} G[K]=G$.
For $K\in \mathcal{K}$, let $\partial K$ be the set of vertices of $K$ that belong to more than one element of $\mathcal{K}$.  The \emph{excess}
of a cover $\mathcal{K}$ is $\sum_{K\in \mathcal{K}} |\partial K|$.

\begin{lemma}\label{lemma-break}
Let $\mathcal{G}$ be a class of graphs with strongly sublinear separators.  There exists a polynomial $p$ such that for
every $\varepsilon>0$, every graph $G\in \mathcal{G}$ has a cover with excess at most $\varepsilon|V(G)|$ consisting of sets of size
at most $p(1/\varepsilon)$.
\end{lemma}

\subsection{Proof of Theorem~\ref{thm-lsearch-hitting} }

The proof of Theorem~\ref{thm-lsearch-hitting} follows the idea of Har-Peled and Quanrud \cite{har2017approximation}, with the following crucial difference: Instead of
comparing the local search solution with the optimal one, we compare it with a suitably chosen enlargement of the optimal solution.
Recall that $\text{WR}_r(G,\sigma,v)$ denote the set of vertices in $G$ that are $r$-weakly reachable from a vertex $v$ under an ordering $\sigma$ of vertices of $G$.
Let $\text{WR}^{-1}_r(G,\sigma,u)$ denote the set of vertices $v$ from which $u$ is $r$-weakly reachable, i.e., the set $\{v\in V(G):u\in \text{WR}_r(G,\sigma,v)\}$.
For a set $O$ of vertices a non-negative integer $m$, we say a vertex $u$ is \emph{$(r,\sigma,m,O)$-rich} if $|\text{WR}^{-1}_r(G,\sigma,u)\cap O|\ge m$,
i.e., if $u$ is $r$-weakly reachable from many vertices in $O$.

\begin{observation}\label{obs-fewrich}
Let $G$ be a graph, let $O$ be a set of vertices of $G$, let $r$ and $m$ be non-negative integers, and let $\sigma$
be a linear ordering of vertices of $G$.  Let $O'$ be the set of all $(r,\sigma,m,O)$-rich vertices.
If the weak $r$-coloring number of $G$ under $\sigma$ is at most $b$, then
$$|O'|\le \frac{b}{m}|O|.$$
\end{observation}
\begin{proof}
By the definition of richness, there are at least $m|O'|$ pairs of vertices $(u,v)$ such that $u\in O'$
and $v\in \text{WR}^{-1}_r(G,\sigma,u)\cap O$.  However, then $u\in \text{WR}_r(G,\sigma,v)$,
and $|\text{WR}_r(G,\sigma,v)|\le b$ for each $v\in V(G)$ by the assumptions.  Hence, the number of such pairs is also at most $b|O|$.
\end{proof}

In the proof, $O$ will be an optimal solution and we are going to enlarge it slightly by adding the vertices of $O'$.
Let us consider the setting of Observation~\ref{obs-fewrich}, and let $A$ be another set of vertices of $G$ (in the proof, this will be the
solution returned by the local search).  The key idea of Har-Peled and Quanrud \cite{har2017approximation} was to consider
an auxiliary packing graph of suitably chosen neighborhoods of vertices of $O\cup A$.  In our setting, we are going to consider
the packing graph of the following set system.
Let $\mathcal{C}_{r,\sigma,m,O,A}=\{C_u:u\in A\cup O\cup O'\}$ be a system of subsets of vertices of $G$ defined
as follows:
\begin{itemize}
\item For $u\in O'\setminus O$, let $C_u=\text{WR}^{-1}_r(G,\sigma,u)$.
\item For $u\in O$, let
$$C_u=\text{WR}^{-1}_r(G,\sigma,u)\cup \bigcup_{x\in \text{WR}_r(G,\sigma,u)\setminus O'} \text{WR}^{-1}_r(G,\sigma,x).$$
\item For $u\in A\setminus (O\cup O')$, let $C_u=\{u\}$.
\end{itemize}

The system is chosen so that we get the following lemma, stating that when $O$ and $A$ both hit a bounded diameter
subgraph, this is reflected by an adjacency in the packing graph of $\mathcal{C}_{r,\sigma,m,O,A}$.

\begin{lemma}\label{lemma-main-hitting}
Let $G$ be a graph, let $A$ and $O$ be sets of vertices of $G$, let $r$ and $m$ be non-negative integers, and let $\sigma$
be a linear ordering of vertices of $G$.  Let $O'$ be the set of all $(r,\sigma,m,O)$-rich vertices.
Let $H$ be the packing graph of $\mathcal{C}=\mathcal{C}_{r,\sigma,m,O,A}$.
If $F$ is a subgraph of $G$ of diameter at most $r$ and $V(F)\cap O\neq\emptyset$, then
there exists $u\in V(F)\cap(O\cup O')$ such that $C_uC_v\in E(H)$ for every $v\in V(F)\cap A\setminus (O\cup O')$.
Moreover, if the weak $r$-coloring number of $G$ under $\sigma$ is at most $b$, then
$\mathcal{C}$ is $(bm+1,2r)$-shallow.
\end{lemma}
\begin{proof}
Let $x$ be the largest vertex of $V(F)$ in the ordering $\sigma$.  
Since $F$ has diameter at most $r$, for every $y\in V(F)$, there exists a path from $y$ to $x$ of length at most
$r$ in $F$, and by the choice of $x$, the vertices of this path appear before $x$ in the ordering $\sigma$.
Consequently $x\in \text{WR}_r(G,\sigma,y)$.  Since this holds for every $y\in V(F)$, we have $V(F)\subseteq \text{WR}^{-1}_r(G,\sigma,x)$.

If $x\in O'$, then let $u=x$, otherwise choose $u$ as an arbitrary vertex of $V(F)\cap O$.
In the former case, we have $V(F)\subseteq \text{WR}^{-1}_r(G,\sigma,u)\subseteq C_u$.
In the latter case, we have $x\in \text{WR}_r(G,\sigma,u)\setminus O'$, and
thus $V(F)\subseteq \text{WR}^{-1}_r(G,\sigma,u)\subseteq C_u$.  Hence, $C_uC_v\in E(H)$ for every $v\in V(F)\cap A\setminus (O\cup O')$.

Next, let us bound the diameter of the sets in $\mathcal{C}$.
For any $z\in V(G)$ and for every vertex $w$ of $\text{WR}^{-1}_r(G,\sigma,z)$, there is a path of length at most $r$ from $w$ to $z$
with all vertices appearing before $z$ in $\sigma$; observe that the vertices of this path also belong to $\text{WR}^{-1}_r(G,\sigma,z)$.
Therefore, every vertex of $\text{WR}^{-1}_r(G,\sigma,z)$ is at distance at most $r$ from $z$ in $G[\text{WR}^{-1}_r(G,\sigma,z)]$
and thus $G[\text{WR}^{-1}_r(G,\sigma,z)]$ has diameter at most $2r$.  For $u\in O\cup O'$, we conclude that $C_u$
is the union of vertex sets of graphs of diameter at most $2r$, each of them containing $u$, and thus $G[C_u]$ has radius at most $2r$.
For $u\in A\setminus (O\cup O')$, $G[C_u]$ has radius $0$.

Finally, let us argue that no vertex belongs to too many of the sets in $\mathcal{C}$. If a vertex $z$ belongs to $C_u$ for $u\in O'\setminus O$,
then $z\in \text{WR}^{-1}_r(G,\sigma,u)$, and thus $u$ is one of at most $b$ vertices in $\text{WR}_r(G,\sigma,z)$.
If $z$ belongs to $C_u$ for $u\in O$, then either $z\in \text{WR}^{-1}_r(G,\sigma,u)$ and $u$ is again one vertices in $\text{WR}_r(G,\sigma,z)$,
or $z\in \text{WR}^{-1}_r(G,\sigma,x)$ for some $x\in \text{WR}_r(G,\sigma,u)\setminus O'$. There are at most
$b$ choices for $x$ in $\text{WR}_r(G,\sigma,z)$, and for fixed $x\not\in O'$, there are less than $m$ vertices $u\in O$
such that $x\in \text{WR}_r(G,\sigma,u)$ (or equivalently, $u\in \text{WR}^{-1}_r(G,\sigma,x)\cap O$).  Finally, vertex $z$ belongs to $C_u$ for
at most one $u\in A\setminus (O\cup O')$.  Therefore, $z$ belongs to $C_v$ for at most $b+b(m-1)+1=bm+1$ vertices $v\in A\cup O\cup O'$.
Therefore, $\mathcal{C}$ is $(bm+1,2r)$-shallow.
\end{proof}

We are now ready to prove the main result, Theorem~\ref{thm-lsearch-hitting}, which we restate for convenience.

\lsearchhitting*
\begin{proof}
Without loss of generality, assume that $\varepsilon\le 1$.  Let $r$ be the diameter of $\pi$.
By Theorem~\ref{thm-weakcolor}, there exists $b$ such that the weak $r$-coloring number of every graph $G\in \mathcal{G}$ is at most $b$.
Let $m=\lceil 3b/\varepsilon\rceil$.  Let $\mathcal{G}'$ be the class of graphs with polynomial expansion
given by Corollary~\ref{cor-preserve} with $\omega=bm+1$ and $t=2r$.  Let $p$ be the polynomial
from Lemma~\ref{lemma-break} for this class $\mathcal{G}'$.
Let us define $c=\lceil p(18/\varepsilon)\rceil$.

Consider any graph $G\in\mathcal{G}$, and let $\sigma$ be an ordering of vertices of $G$ such that $G$ has weak $r$-coloring number at most $b$ under $\sigma$.
Let $O$ be a $\pi$-hitting set in $G$ of size $\gamma_\pi(G)$, and let $A$ be the $\pi$-hitting
set returned by the $c$-local search algorithm.  Let $O'$ be the set of all $(r,\sigma,m,O)$-rich vertices in $G$;
by Observation~\ref{obs-fewrich}, we have $|O'|\le \tfrac{b}{m}|O|\le \tfrac{\varepsilon}{3}|O|$.
Let $\mathcal{C}=\mathcal{C}_{r,\sigma,m,O,A}$; by Lemma~\ref{lemma-main-hitting}, $\mathcal{C}$ is $(bm+1,2r)$-shallow, and
by Corollary~\ref{cor-preserve}, the packing graph $H=G[\mathcal{C}]$ belongs to $\mathcal{G}'$.  For any $Y\subseteq V(H)$,
let $\overline{Y}=\{v\in V(G):C_v\in Y\}$.  By Lemma~\ref{lemma-break},
there exists a cover $\mathcal{K}$ of $H$ with excess at most $\tfrac{\varepsilon}{18}|V(H)|\le \tfrac{\varepsilon}{18}(|A|+|O|+|O'|)\le \tfrac{\varepsilon}{6}|A|$
such that each element of $\mathcal{K}$ has size at most $p(18/\varepsilon)\le c$.

Consider any $K\in \mathcal{K}$, and let $A'=(A\setminus (\overline{K\setminus\partial K}))\cup ((O\cup O')\cap \overline{K})$.
We claim that $A'$ is a $\pi$-hitting set.  Indeed, suppose for a contradiction that there exists $Z\subseteq V(G)$ such
that $(G,Z)$ satisfies the property $\pi$, but $A'\cap Z=\emptyset$.  Since $A$ is a $\pi$-hitting set, there exists $v\in A\cap Z$,
and since $v\not\in A'$, we have $v\in (\overline{K\setminus\partial K})\setminus (O\cup O')$.
Since $O$ is a $\pi$-hitting set, we have $Z\cap O\neq\emptyset$,
and by Lemma~\ref{lemma-main-hitting} applied to $F=G[Z]$, there exists $u\in Z\cap (O\cup O')$ such that $C_uC_v\in E(H)$.
Since $\mathcal{K}$ is a cover of $H$, there exists $K'\in \mathcal{K}$ such that $C_uC_v\in E(H[K'])$.
Since $C_v\in K\setminus\partial K$, it follows that $K'=K$, and thus this implies $C_u\in K$.
Therefore, we have $u\in (O\cup O')\cap \overline{K}$, and $u\in A'\cap Z$, which is a contradiction.

Note that $A\triangle A'\subseteq \overline{K}$, and thus $|A\triangle A'|\le c$.  Hence, the $c$-local search algorithm
considered the solution $A'$ and did not improve $A$ to $A'$, which implies that $|A'|\ge |A|$.
We conclude that
$$|(O\cup O')\cap \overline{K}|\ge |A\cap (\overline{K\setminus\partial K})|$$
for every $K\in \mathcal{K}$.  Therefore, denoting by $s$ the excess of $\mathcal{K}$, we have
\begin{align*}
|A|&\le \sum_{K\in\mathcal{K}} |A\cap \overline{K}|\le s+\sum_{K\in\mathcal{K}} |A\cap (\overline{K\setminus\partial K})|\\
&\le s+\sum_{K\in\mathcal{K}} |(O\cup O')\cap \overline{K}|\le 2s+\sum_{K\in\mathcal{K}} |(O\cup O')\cap (\overline{K\setminus\partial K})|\\
&\le 2s+|O|+|O'|\le \tfrac{\varepsilon}{3}|A|+(1+\varepsilon/3)|O|.
\end{align*}
Therefore,
$$|A|\le \frac{1+\varepsilon/3}{1-\varepsilon/3}|O|\le (1+\varepsilon)|O|,$$
as required. We will discuss the running time of the algorithm in Section~\ref{sec:efficient-clocalsearch}. 
\end{proof}

The proof for the packing version is much simpler.
\lsearchpacking*
\begin{proof}
Without loss of generality, assume that $\varepsilon\le 1$.  Let $r$ be the diameter of $\pi$.
Let $\mathcal{G}'$ be the class of graphs with polynomial expansion given by Corollary~\ref{cor-preserve}
with $\omega=2$ and $t=r$.  Let $p$ be the polynomial from Lemma~\ref{lemma-break} for this class $\mathcal{G}'$.
Let us define $c=\lceil p(4/\varepsilon)\rceil$.

Consider any graph $G\in\mathcal{G}$.  Let $O$ be an (induced) $\pi$-packing in $G$ of size $\alpha_\pi(G)$ (or $\alpha'_\pi(G)$),
and let $A$ be the (induced) $\pi$-packing returned by the $c$-local search algorithm. 
Let $\mathcal{C}=A\cup O$, and note that $\mathcal{C}$ is $(2,r)$-shallow.  By Corollary~\ref{cor-preserve}, the packing graph $H=G[\mathcal{C}]$
belongs to $\mathcal{G}'$.  By Lemma~\ref{lemma-break}, there exists a cover $\mathcal{K}$ of $H$ with excess at most
$\tfrac{\varepsilon}{4}|V(H)|\le \tfrac{\varepsilon}{4}(|A|+|O|)\le \tfrac{\varepsilon}{2}|O|$
such that each element of $\mathcal{K}$ has size at most $p(4/\varepsilon)\le c$.

Consider any $K\in \mathcal{K}$, and let $A'=(A\setminus K)\cup (O\cap (K\setminus \partial K))$.
We claim that $A'$ is an (induced) $\pi$-packing.  Indeed, since both $A$ and $O$ have this property, the only
way how this could be false is if there existed $P\in A\setminus K$ and $P'\in O\cap (K\setminus \partial K)$
that intersect (or contain adjacent vertices in the induced case).  However, then $PP'\in E(H)$, which
is a contradiction since $P'\in K\setminus \partial K$ and $P\not\in K$.
Observe also that $|A'|\le |A|$, as otherwise the $c$-local search algorithm
would extend $A\setminus K$ to a solution larger than $A$ (if $|A'|>|A|$, then $K\not\subseteq A$,
and thus $A\setminus K$ is obtained from $A$ by deleting less than $c$ elements).

We conclude that
$$|A\cap K|\ge |O\cap (K\setminus\partial K)|$$
for every $K\in \mathcal{K}$.  Therefore, denoting by $s$ the excess of $\mathcal{K}$, we have
\begin{align*}
|A|&\ge \sum_{K\in\mathcal{K}} |A\cap (K\setminus \partial K)|\ge -s+\sum_{K\in\mathcal{K}} |A\cap K|\\
&\ge -s+\sum_{K\in\mathcal{K}} |O\cap (K\setminus \partial K)|\ge -2s+\sum_{K\in\mathcal{K}} |O\cap K|\\
&\ge -2s+|O|\ge (1-\varepsilon)|O|,
\end{align*}
as required.  We discuss the running time of the algorithm in Section~\ref{sec:efficient-clocalsearch}. 
\end{proof}

\subsection{Efficient implementation of the local search algorithms}
\label{sec:efficient-clocalsearch}

The time complexity of the $c$-local-search algorithm (for a straightforward implementation) is $O(n^{c+1})$ times the complexity of
testing whether a set is $\pi$-hitting or finding an extension of a $\pi$-packing.  Consequently, the exponent of the time
complexity of the algorithms arising from Theorem~\ref{thm-lsearch-hitting} and \ref{thm-lsearch-packing} depends on the desired precision;
i.e., we obtain a PTAS, not an EPTAS.  However, in all the problems that we considered ({\sc Induced Subgraph Hitting}, {\sc Distance $r$-Domination} or {\sc Independent Set}),
the property $\pi$ is \emph{FO-definable} in the following sense: There exists a formula $\varphi(x_1,\ldots,x_m,z)$
with free variables $x_1$, \ldots, $x_m$, and $z$ in the first-order logic (allowing quantification over vertices, but not over sets of vertices or over edges)
using the adjacency predicate, such that $(G,Z)$ with $Z$ non-empty has the property
$\pi$ if and only if for some $v_1,\ldots,v_m\in G$,
$$Z=\{v\in V(G):G\models \varphi(v_1,\ldots,v_m,v)\}.$$
E.g., for the property $\pi$ used to define $r$-dominating and $2r$-independent set problems, we can set
$\varphi(x_1,z)\equiv\text{dist}_r(x_1,z)$,
where
$$\text{dist}_r(x,y)\equiv (\exists p_0,\ldots,p_r)\;x=p_0\land y=p_r\land \bigwedge_{i=1}^r (p_{i-1}=p_i\lor p_{i-1}p_i\in E).$$
In this case, we can express the testing and extension in first-order logic.  Let us demonstrate this on the most difficult
example, the $c$-local search for the {\sc $\pi$-Packing} problem.

Let $\mathcal{P}$ be the current $\pi$-packing in a graph $G$ that we are trying to improve.
To represent $\mathcal{P}$ in a way accessible in a first-order formula, for each element $P$ of $\mathcal{P}$,
color the edges of a BFS spanning tree the subgraph of $G$ induced by $P$ red; we will access this information about edge colors using a new binary predicate $R$.
Let $\text{dist}^R_r(x,y)$ be the predicate defined analogously to $\text{dist}_r(x,y)$, but with $E$ replaced by $R$.
Let $T$ be the unary predicate interpreted as the set $\bigcup \mathcal{P}$.  Note that in this representation, the sets in $\mathcal{P}$
are exactly the components of the red-edge graph contained in the set $T$.
We can delete at most $c-1$ elements from $\mathcal{P}$ and add $c$ other elements to obtain a $\pi$-packing exactly if the following
first-order formula is satisfied:
\begin{align*}
(\exists& r_1,\ldots,r_{c-1},x_1^1,\ldots,x_m^1,\ldots,x_1^c,\ldots,x_m^c)\\
&\bigwedge_{j=1}^c (\exists z)\;\varphi(x_1^j,\ldots,x_m^j,z)\\
&\land \bigwedge_{j=1}^c (\forall z)\;\Bigl(\varphi(x_1^j,\ldots,x_m^j,z)\land z\in T\Rightarrow \bigvee_{i=1}^{c-1} r_i\in T\land \text{dist}_{2r}^R(z,r_i)\Bigr)\\
&\land \bigwedge_{1\le j_1 < j_2\le c} (\forall z_1,z_2)\;\bigl(\varphi(x_1^{j_1},\ldots,x_m^{j_1},z_1)\land \varphi(x_1^{j_2},\ldots,x_m^{j_2},z_2)\Rightarrow z_1\neq z_2\bigr).
\end{align*}
The elements to be removed from $\mathcal{P}$ are those intersected by $r_1$, \ldots, $r_{c-1}$, while those added are given by the formula $\varphi$
with parameters $x_1^j$, \ldots, $x_m^j$ for $j\in\{1,\ldots,c\}$.
\begin{itemize}
\item The second line expresses that for each $j$, the formula $\varphi$ parameterized by the $m$-tuple $x_1^j$, \ldots, $x_m^j$ defines a non-empty set $Z_j$.
\item The third line states that $Z_j$ intersects $T$ only in the elements containing $r_1$, \ldots, $r_{c-1}$
that are to be removed from $\mathcal{P}$.
\item The final line describes that the sets $Z_1$, \ldots, $Z_m$ are disjoint.
\end{itemize}

Using the data structure from Theorem~\ref{thm:fologic} (in a variant supporting colors of edges),
we can maintain the predicates $T$ and $R$ representing the current state of $\mathcal{P}$ in constant time
per recoloring, and find a value for the variables $r_1$, \ldots, $r_{c-1}$, $x_1^1$, \ldots, $x_m^1$, \ldots, $x_1^c$, \ldots, $x_m^c$
satisfying the formula above (if they exist) in constant time per query.

We conclude that for an FO-definable property of bounded diameter, an $n$-vertex graph $G$ from any fixed class with bounded expansion, and any fixed positive integer $c$, we can run
\begin{itemize}
\item $c$-local search for minimum $\pi$-hitting in time $O(n)$, and
\item (induced) $c$-local search for maximum $\pi$-packing in time $O(n^2)$.
\end{itemize}
The reason for the quadratic time in the second case is due to the need to recolor the red edges after each iteration,
which may result in $O(n)$ time in case the sets of the (induced) $\pi$-packing have unbounded size.
In case they have bounded size, e.g., for the maximum (induced) $F$-matching problem, the time complexity becomes $O(n)$.
Moreover, even if the elements of the packing have unbounded size, it may be possible to use a different way how to
represent the packing in the first-order logic (e.g., in the case of maximum $t$-independent set, we can just mark the vertices
forming the set), again resulting in total time $O(n)$.

\section{Applications to geometric intersection graphs}

\subsection{Fat-object graphs and pseudo-disk graphs}

A graph class $\mathcal{G}$ has \textit{bounded clique-size} if there exists an integer $k \in \mathbb{N}$ such that $\omega(G) \leq k$ for all $G \in \mathcal{G}$.
We say a graph class $\mathcal{G}$ is \textit{of clique-dependent polynomial expansion} if any sub-class $\mathcal{G}' \subseteq \mathcal{G}$ that has bounded clique-size is of polynomial expansion.
Similarly, we say a graph class $\mathcal{G}$ is \textit{of clique-dependent bounded degeneracy} if any sub-class $\mathcal{G}' \subseteq \mathcal{G}$ that has bounded clique-size is of bounded degeneracy.
By definition, it is clear that if a graph class is of clique-dependent polynomial expansion, then it is also of clique-dependent bounded degeneracy.
In what follows, we shall discuss two examples of geometric intersection graphs that are of clique-dependent polynomial expansion.

The first example is intersection graphs of \textit{fat objects}.
A geometric object $X$ in $\mathbb{R}^d$ is \textit{$\alpha$-fat} for a number $\alpha \geq 1$ if $X$ is convex and there are two balls $B_\mathsf{in},B_\mathsf{out}$ in $\mathbb{R}^d$ such that $B_\mathsf{in} \subseteq X \subseteq B_\mathsf{out}$ and $\mathsf{rad}(B_\mathsf{out}) \leq \alpha \cdot \mathsf{rad}(B_\mathsf{in})$.
Clearly, balls are exactly the $1$-fat geometric objects.
We say $\mathcal{G}$ is a class of \textit{fat-object graphs} if there exist $\alpha \geq 1$ and $d \in \mathbb{N}$ such that every graph in $\mathcal{G}$ is the intersection graph of a set of $\alpha$-fat geometric objects in $\mathbb{R}^d$.
It was known that fat-object graphs with bounded clique size form a graph class with strongly sublinear separators \cite{de2020framework,dvovrak2021approximation,ErlebachJS05}.
Therefore, a class of fat-object graphs is of clique-dependent polynomial expansion.

The second example is intersection graphs of \textit{pseudo-disks}.
A set of geometric objects in the plane are called \textit{pseudo-disks} if each of them is homeomorphic to a disk and the boundaries of any two objects are either disjoint or intersect exactly twice.
A graph is a \textit{pseudo-disk graph} if it is the intersection graph of a set of pseudo-disks in the plane.
We show that pseudo-disk graphs with bounded clique size have strongly sublinear separators, and thus form a graph class of clique-dependent polynomial expansion.

\begin{lemma} \label{lem-pdiskcdpoly}
    The class of pseudo-disk graphs is of clique-dependent polynomial expansion.
\end{lemma}
\begin{proof}
Consider a set $\mathcal{S}$ of pseudo-disks in the plane and let $G$ be the intersection graph of $\mathcal{S}$.
Suppose $\omega(G) \leq p$.
For each $S \in \mathcal{S}$, we denote by $v_S \in V(G)$ its corresponding vertex.
The boundaries of the pseudo-disks in $\mathcal{S}$ subdivide the plane into connected regions called \textit{faces}.
The \textit{depth} of a face is the number of pseudo-disks in $\mathcal{S}$ containing it.
Two faces are \textit{adjacent} if their boundaries share a common arc.
The \textit{arrangement graph} $A(\mathcal{S})$ of $S$ is a planar graph whose vertices are the faces contained in at least one pseudo-disk in $\mathcal{S}$ and two vertices are connected by an edge if they are adjacent.
Observe that each $S \in \mathcal{S}$, the subgraph of $A(\mathcal{S})$ induced by the faces contained in $S$ is connected, because $S$ is connected.
It is well-known \cite{kedem1986union} that the number of faces of depth at most $d$ is bounded by $O(dn)$, where $n = |\mathcal{S}|$.
Since $\omega(G) \leq p$, all faces are of depth at most $p$ and thus the total number of faces is $O(pn)$, which implies that $A(\mathcal{S})$ has $O(pn)$ vertices.
By recursively applying the separator theorem for planar graphs, we can find a set $X$ of $f(p) \cdot \sqrt{n}$ vertices of $A(\mathcal{S})$ such that each connected component of $A(\mathcal{S})-X$ has at most $n/(2p)$ vertices.
Let $\mathcal{S}' \subseteq S$ be the subset consisting of pseudo-disks which contain a face corresponding to some $x \in X$.
Since the depth of each face is at most $p$, we have $|\mathcal{S}'| \leq p \cdot |X| = f'(p) \cdot \sqrt{n}$ where $f'(p) = pf(p)$.
We claim that $\mathcal{S}'$ is a balanced separator of $G$.
Let $C_1,\dots,C_r$ be the connected components of $A(\mathcal{S})-X$.
For each pseudo-disk $S \in \mathcal{S} \backslash \mathcal{S}'$, let $Y_S$ be the set of vertices of $A(\mathcal{S})$ whose corresponding faces are contained in $S$.
Observe that the vertices in $Y_S$ belong to the same connected component of $A(\mathcal{S})-X$, because $Y_S \cap X = \emptyset$ and $Y_S$ induces a connected subgraph of $A(\mathcal{S})$; we say $S$ \textit{belongs to} $C_i$ if the vertices in $Y_S$ belong to $C_i$.
Also notice that if $S,S' \in \mathcal{S} \backslash \mathcal{S}'$ belong to different connected components of $A(\mathcal{S})-X$, then there is no edge between $S$ and $S'$ in $G$, since $S$ and $S'$ do not intersect.
Therefore, the vertices in one connected component of $G - \mathcal{S}'$ must belong to the same $C_i$.
Since each $C_i$ has at most $n/(2p)$ vertices and the face corresponding to each vertex can be contained in at most $p$ pseudo-disks in $\mathcal{S}$, the number of vertices of $G - \mathcal{S}'$ belonging to $C_i$ is at most $n/2$.
As a result, each connected component of $G - \mathcal{S}'$ contains at most $n/2$ vertices and thus $\mathcal{S}'$ is a balanced separator of $G$.
Now we see every $n$-vertex pseudo-disk graph with maximum clique size at most $p$ has a balanced separator of size $f(p) \cdot \sqrt{n}$ for some function $f$.
Thus, the class of pseudo-disk graphs is of clique-dependent polynomial expansion.
\end{proof}

Now we show how to extend our results in the previous sections to graph classes of clique-dependent polynomial expansion.
The key is a clique-decomposition scheme.
For a graph $G$ and an integer $k \geq 1$, a \textit{$k$-clique decomposition} of $G$ is a partition of $V(G)$ into $V_0,V_1,\dots,V_r$ such that $\omega(G[V_0])<k$ and $G[V_i]$ is a $k$-clique for all $i \in [r]$.
Note that a $k$-clique decomposition always exists for any graph and any $k$.
Indeed, one can keep removing $k$-cliques from the graph until the maximum clique is of size smaller than $k$.

\begin{lemma} \label{lem-decompdeg}
Given a graph $G$ and an integer $k \geq 1$, one can compute a $k$-clique decomposition of $G$ in $f(d) \cdot n$ time for some function $f$, where $n = |V(G)|$ and $d$ is the degeneracy of $G$.
\end{lemma}
\begin{proof}
Using the classical algorithm by Matula and Beck \cite{matula1983smallest}, one can compute a degeneracy ordering of $G$ in $O(n+|E(G)|)$ time, i.e., $O(dn)$ time.
Let $v_1,\dots,v_n$ be the degeneracy ordering computed.
Define $E_i = \{(v_j,v_i) \in E(G): j<i\}$ for $i \in [n]$.
By the property of a degeneracy ordering, $|E_i| \leq d$ for all $i \in [n]$.
Furthermore, one can easily compute $E_1,\dots,E_n$ in $O(dn)$ time.
Next, we consider each vertex $v_i$ for $i$ from $1$ to $n$.
During this procedure, we maintain a $k$-clique decomposition of the graph $G[\{v_1,\dots,v_i\}]$.
Suppose $V_0,V_1,\dots,V_r$ is a $k$-clique decomposition of the graph $G[\{v_1,\dots,v_{i-1}\}]$.
In order to obtain a $k$-clique decomposition of $G[\{v_1,\dots,v_i\}]$, we consider the graph $G[V_0 \cup \{v_i\}]$.
If $\omega(G[V_0 \cup \{v_i\}])<k$, then $V_0 \cup \{v_i\},V_1,\dots,V_r$ is a $k$-clique decomposition of $G[\{v_1,\dots,v_i\}]$.
Otherwise, there exists a $k$-clique $K$ in $G[V_0 \cup \{v_i\}]$, which must contain $v_i$ since $\omega(G[V_0])<k$.
In this case, $V_0 \backslash V(K),V_1,\dots,V_r,K$ is a $k$-clique decomposition of $G[\{v_1,\dots,v_i\}]$, because $\omega(G[V_0 \backslash V(K)]) \leq \omega(G[V_0])<k$.
So it suffices to find the clique $K$ (or verify its non-existence).
Since we know $v_i \in V(K)$, the other vertices of $K$ are all in $N_G(v_i) \cap V_0$.
The problem now becomes finding a $(k-1)$-clique in $G[N_G(v_i) \cap V_0]$.
Note that the vertices in $N_G(v_i) \cap V_0$ are exactly those incident to the edges in $E_i$ and contained in $V_0$.
As $|E_i| \leq d$, by checking all endpoints of the edges in $E_i$, we can compute $N_G(v_i) \cap V_0$ in $O(d)$ time.
Indeed, we can determine whether a vertex is in $V_0$ or not in $O(1)$ time by maintaining a mark on each vertex in $V_0$.
After obtaining $N_G(v_i) \cap V_0$, we then retrieve the edges of $G[N_G(v_i) \cap V_0]$.
This can be done in $O(d^2)$ time since every edge of $G[N_G(v_i) \cap V_0]$ is contained in $E_j$ for some $v_j \in N_G(v_i) \cap V_0$.
Once we have the graph $G[N_G(v_i) \cap V_0]$ in hand, we can find a $(k-1)$-clique in $G[N_G(v_i) \cap V_0]$ or determine no such a clique in $2^{O(d)}$ time by brute-force, because $G[N_G(v_i) \cap V_0]$ has at most $d$ vertices.
In this way, we obtain a $k$-clique decomposition of $G[\{v_1,\dots,v_i\}]$ from the $k$-clique decomposition of $G[\{v_1,\dots,v_{i-1}\}]$ in $2^{O(d)}$ time.
After considering all vertices $v_1,\dots,v_n$, we finally obtain a $k$-clique decomposition of $G$, and the total time cost is $f(d) \cdot n$ for some function $f$.
\end{proof}

\begin{lemma} \label{lem-cliquedecomp}
Let $\mathcal{G}$ be a hereditary graph class of clique-dependent bounded degeneracy.
Given a graph $G \in \mathcal{G}$ and an integer $k \geq 1$, one can compute a $k$-clique decomposition of $G$ in $f(k) \cdot n \log n + O(m)$ time for some function $f$, where $n = |V(G)|$ and $m = |E(G)|$.
\end{lemma}
\begin{proof}
For $p \in \mathbb{N}$, we write $\mathcal{G}_p = \{G \in \mathcal{G}: \omega(G) \leq p\}$.
Since $\mathcal{G}$ is of clique-dependent bounded degeneracy, there exists a function $d: \mathbb{N} \rightarrow \mathbb{N}$ such that all graphs in $\mathcal{G}_p$ has degeneracy at most $d(p)$ for every $p \in \mathbb{N}$.
We compute a $k$-clique decomposition of a given graph $G \in \mathcal{G}$ using a simple divide-and-conquer algorithm.
First, we evenly partition the vertices of $G$ into two sets $V'$ and $V''$, each of which contains $n/2$ vertices.
Then we recursively compute a $k$-clique decomposition $V_0',V_1',\dots,V_{r'}'$ of $G[V']$ and a $k$-clique decomposition $V_0'',V_1'',\dots,V_{r''}''$ of $G[V'']$.
Finally, we compute a $k$-clique decomposition $V_0,V_1,\dots,V_r$ of $G[V_0' \cup V_0'']$ using the algorithm in Lemma~\ref{lem-decompdeg}.
Observe that $V_0,V_1,\dots,V_r,V_0',V_1',\dots,V_{r'}',V_1'',\dots,V_{r''}''$ is a $k$-clique decomposition of $G$.
Indeed, this is clearly a partition of $V(G)$ in which $\omega(G[V_0]) < k$ and all the other parts induce $k$-cliques.
Therefore, we can return it as our output.
To analyze the running time of our algorithm, note that $\omega(G[V_0' \cup V_0'']) < 2k$, because $\omega(G[V_0'])<k$ and $\omega(G[V_0''])<k$.
Also, $G[V_0' \cup V_0''] \in \mathcal{G}$ since $\mathcal{G}$ is hereditary.
It follows that $G[V_0' \cup V_0''] \in \mathcal{G}_{2k}$ and the degeneracy of $G[V_0' \cup V_0'']$ is at most $d(2k)$.
Therefore, applying the algorithm in Lemma~\ref{lem-decompdeg} takes $f_0(d(2k)) \cdot n$ time for some function $f_0$.
Besides, before the recursive calls of the algorithm on $G[V']$ and $G[V'']$, we need to construct these two graphs.
Also, we need to construct $G[V_0' \cup V_0'']$ before working on it.
If we do these graph constructions explicitly, it takes $O(m)$ time.
In this way, the running time of our algorithm satisfies the recurrence
\begin{equation*}
    T(n,m) = T(n/2,m') + T(n/2,m'') + f_0(d(2k)) \cdot n + O(m),
\end{equation*}
where $m'$ and $m''$ are the numbers of edges in $G[V']$ and $G[V'']$, respectively.
As $m'+m'' \leq m$, This recurrence solves to $T(n,m) = f(k) \cdot n \log n + O(m \log n)$ for some function $f$.

In fact, we can further improve the running time to $f(k) \cdot n \log n + O(m)$ as follows.
The recursion tree of the above algorithm is a binary tree $T$ of depth $O(\log n)$ whose leaves one-to-one correspond to the $n$ vertices in $G$.
Each node of $T$ corresponds to a recursive call of the algorithm.
For each node $t \in T$, let $G_t$ denote the subgraph of $G$ induced by the vertices corresponding to the leaves of the subtree of $T$ rooted at $t$.
The call at the node $t$ returns a $k$-clique decomposition of $G_t$.
In order to avoid the $O(m \log n)$ term in the above algorithm, we do a preprocessing as follows.
We first build the recursion tree $T$, which can be done in $O(n)$ time.
Next, we store every edge $e = (u,v) \in E(G)$ at the lowest common ancestor of the two leaves of $T$ corresponding to $u$ and $v$.
Using the classical algorithm by Harel and Tarjan \cite{harel1984fast}, the lowest common ancestor can be found in $O(1)$ time.
Thus, this step can be done in $O(m)$ time in total.
Let $E_t$ denote the set of edges stored at each node $t \in T$.
After this preprocessing, we apply the above recursive algorithm with some small modifications.
First, in the call at a node $t \in T$, we not only return a $k$-clique decomposition $V_0,V_1,\dots,V_k$ of $G_t$, but also return the induced subgraph $G[V_0]$.
Achieving this is not difficult, as the part $V_0$ is finally obtained by the algorithm of Lemma~\ref{lem-decompdeg} and that algorithm can actually also return the induced subgraph $G[V_0]$ in the same running time.
Second, when we recursively call the algorithm at the children of a node, we do not construct the corresponding graphs of the children beforehand.
This modification brings us an issue: when solving the problem at a node $t$, we do not have the graph $G_t$ in hand.
In this situation, we can still do the recursive calls at the children $t'$ and $t''$ of $t$ (as we do not need to construct the corresponding graphs).
However, after we obtain the $k$-clique decomposition $V_0',V_1',\dots,V_{r'}'$ of $G_{t'}$ and the $k$-clique decomposition $V_0'',V_1'',\dots,V_{r''}''$ of $G_{t''}$ returned by the recursive calls, we have to efficiently construct the graph $G[V_0' \cup V_0'']$ without knowing $G_t$ and then apply the algorithm of Lemma~\ref{lem-decompdeg} on $G[V_0' \cup V_0'']$.
By assumption, the recursive calls already return us the graphs $G[V_0']$ and $G[V_0'']$.
Observe that all edges of $G[V_0' \cup V_0'']$ except those in $G[V_0']$ and $G[V_0'']$ are in the set $E_t$.
Therefore, to construct $G[V_0' \cup V_0'']$ provided $G[V_0']$ and $G[V_0'']$, we only need to check every edge in $E_t$ and find those with one endpoint in $V_0'$ and the other endpoint in $V_0'$, which can be done in $O(|E_t|)$ time.
Now if we use $T(t)$ to denote the running time of the call at the node $t$, it satisfies the recurrence
\begin{equation*}
    T(t) = T(t') + T(t'') + f_0(d(2k)) \cdot |V(G_t)| + O(|E_t|).
\end{equation*}
Since $\sum_{t \in T} |E_t| = m$, the recurrence solves to $T(\mathsf{rt}) = f(k) \cdot n \log n + O(m)$ for some function $f$, where $\mathsf{rt} \in T$ is the root of $T$.
\end{proof}

Based on the clique-decomposition algorithm, we can extend our lossy kernels and approximation schemes to any classes of geometric intersection graphs that are of clique-dependent polynomial expansion, which include in particular fat-object graphs and pseudo-disk graphs.

\begin{restatable}{theorem}{kernelgeo} \label{thm-kernelgeo}
    Let $\mathcal{G}$ be a class of fat-object graphs or the class of pseudo-disk graphs.
    For any fixed (finite) $\mathcal{F}$ of connected graphs, \textsc{Subgraph Hitting} on $\mathcal{G}$ with forbidden list $\mathcal{F}$ admits a $(1+\varepsilon)$-approximation lossy kernel of size $f(\varepsilon) \cdot k$ for some function $f$.
    The kernelization algorithm runs in $g(\varepsilon) \cdot n \log n + O(m)$ time for some function $g$, where $n$ (resp., $m$) is the number of vertices (resp., edges) of the input graph $G \in \mathcal{G}$.
\end{restatable}
\begin{proof}
In this proof, all instances of \textsc{Subgraph Hitting} are with forbidden list $\mathcal{F}$.
For $p \in \mathbb{N}$, we write $\mathcal{G}_p = \{G \in \mathcal{G}: \omega(G) \leq p\}$.
Since $\mathcal{G}$ is of clique-dependent polynomial expansion, $\mathcal{G}_p$ is of polynomial expansion for every $p \in \mathbb{N}$.
Let $G \in \mathcal{G}$ be a \textsc{Subgraph Hitting} instance and $\varepsilon > 0$ be the approximation factor.
As before, set $\gamma = \max_{F \in \mathcal{F}} |V(F)|$.
Define $t = \frac{1+\varepsilon}{\varepsilon} \cdot (\gamma-1)$.
We first apply the algorithm of Lemma~\ref{lem-cliquedecomp} to compute a $t$-clique decomposition $V_0,V_1,\dots,V_r$ of $G$.
We have $\omega(G[V_0]) < t$.
It follows that $G[V_0] \in \mathcal{G}_t$, since $\mathcal{G}$ is hereditary.

We observe that if $S_0 \subseteq G[V_0]$ is a $c$-approximation solution for the instance $G[V_0]$, then $S = S_0 \cup (\bigcup_{i=1}^r V_i)$ is a $(1+\varepsilon)c$-approximation solution for the instance $G$.
Clearly, $S$ is a feasible solution for $G$.
Indeed, $G-S = G[V_0]-S_0$ and the latter contains no forbidden pattern in $\mathcal{F}$ as a subgraph since $S_0$ is a solution of $G[V_0]$.
It suffices to show that $|S| \leq (1+\varepsilon)c \cdot |S'|$ for any solution $S' \subseteq V(G)$.
First observe that $S' \cap V_0$ is a feasible solution of $G[V_0]$, for otherwise the graph $G[V_0] - S' \cap V_0$ contains a forbidden pattern in $\mathcal{F}$ as a subgraph and so does $G - S'$, which contradicts that $S'$ is a solution for $G$.
Therefore, $|S_0| \leq c \cdot |S' \cap V_0|$.
Now we consider the sets $V_1,\dots,V_r$, which induce $t$-cliques in $G$.
For each $i \in [r]$, we must have $|V_i \backslash S'| \leq \gamma-1$, for otherwise $G-S'$ contains a $\gamma$-clique and thus contains any forbidden pattern in $\mathcal{F}$.
It follows that $|S' \cap V_i| \geq t-(\gamma-1) = \frac{t}{1+\varepsilon}$ for each $i \in [r]$.
So we have
\begin{equation*}
    |S| = |S_0| + rt \leq c \cdot |S' \cap V_0| + (1+\varepsilon) \cdot \sum_{i=1}^r |S' \cap V_i| = (1+\varepsilon)c \cdot |S'|.
\end{equation*}

Now we see that given a $c$-approximation solution for $G[V_0]$, one can compute in linear time a $(1+\varepsilon)c$-approximation solution for $G$.
So we can reduce the instance $G$ to the instance $G[V_0]$.
As $G[V_0] \in \mathcal{G}_t$ and $\mathcal{G}_t$ is of polynomial expansion, we can further apply Theorem~\ref{thm-kernel} to reduce $G[V_0]$ to a linear-size lossy kernel.
The kernel size is $f(\varepsilon,t) \cdot k$ and the running time is $g(\varepsilon,t) \cdot n \log n + O(m)$.
Note that $t$ depends only on $\varepsilon$.
Thus, we have the bounds in the theorem.
\end{proof}

\begin{restatable}{theorem}{schemegeo} \label{thm-schemegeo}
    Let $\mathcal{G}$ be a class of fat-object graphs or the class of pseudo-disk graphs.
    \textsc{Subgraph Hitting} on $\mathcal{G}$ admits an approximation scheme with running time $f(\varepsilon,\mathcal{F}) \cdot n \log n + O(m)$ for some function $f$, where $n$ (resp., $m$) is the number of vertices (resp., edges) of the input graph $G \in \mathcal{G}$.
\end{restatable}
\begin{proof}
The proof is similar to that of Theorem~\ref{thm-kernelgeo}.
For $p \in \mathbb{N}$, we write $\mathcal{G}_p = \{G \in \mathcal{G}: \omega(G) \leq p\}$.
Since $\mathcal{G}$ is of clique-dependent polynomial expansion, $\mathcal{G}_p$ is of polynomial expansion for every $p \in \mathbb{N}$.
Let $(G,\mathcal{F})$ be a \textsc{Subgraph Hitting} instance where $G  \in \mathcal{G}$, and $\varepsilon > 0$ be the approximation factor.
As before, set $\gamma = \max_{F \in \mathcal{F}} |V(F)|$.
Define $t = \frac{1+\varepsilon}{\varepsilon} \cdot (\gamma-1)$.
We apply Lemma~\ref{lem-cliquedecomp} to compute a $t$-clique decomposition $V_0,V_1,\dots,V_r$ of $G$.
We have $\omega(G[V_0]) < t$ and $G[V_0] \in \mathcal{G}_t$.
Now apply Corollary~\ref{cor-lsearchsub} to compute $(1+\varepsilon)$-approximation $S_0 \subseteq V_0$ for $(G[V_0],\mathcal{F})$ in $f(\varepsilon,t) \cdot n$ time.
The same analysis as in the proof of Theorem~\ref{thm-schemegeo} shows that $S = S_0 \cup (\bigcup_{i=1}^r V_i)$ is a $(1+\varepsilon)$-approximation solution for $(G,\mathcal{F})$.
As $t$ depends only on $\varepsilon$ and $\mathcal{F}$, the total running time is $f(\varepsilon,\mathcal{F}) \cdot n$.
\end{proof}

Lemma~\ref{lem-cliquedecomp} also gives us a byproduct for \textsc{$k$-Subgraph Isomorphism}.
\begin{theorem} \label{thm-kiso}
Let $\mathcal{G}$ be a class of fat-object graphs or the class of pseudo-disk graphs.
\textsc{$k$-Subgraph Isomorphism} on $\mathcal{G}$ can be solved in $f(k) \cdot n \log n + O(m)$ time for some function $f$, where $n$ and $m$ are the numbers of vertices and edges in the input graph $G \in \mathcal{G}$, respectively.
\end{theorem}
\begin{proof}
For $p \in \mathbb{N}$, we write $\mathcal{G}_p = \{G \in \mathcal{G}: \omega(G) \leq p\}$.
Since $\mathcal{G}$ is of clique-dependent polynomial expansion, $\mathcal{G}_p$ is of polynomial expansion for every $p \in \mathbb{N}$.
Thus, by Theorem~\ref{thm:IsoBoundedExp}, for each $p \in \mathbb{N}$, there exists a function $f_p$ such that \textsc{$k$-Subgraph Isomorphism} on $\mathcal{G}_p$ can be solved in $f_p(k) \cdot n$ time.
Now consider a graph $G \in \mathcal{G}$, and let $H$ be the $k$-vertex graph we are looking for.
We first apply Lemma~\ref{lem-cliquedecomp} to compute a $k$-clique decomposition $V_0,V_1,\dots,V_r$ of $G$ in $f(k) \cdot n \log n + O(m)$ time.
If $r \geq 1$, then we find a $k$-clique that contains $H$ as a subgraph.
Otherwise, $r = 0$ and $G = G[V_0]$, which implies that $\omega(G) < k$, i.e., $G \in \mathcal{G}_k$.
In this case, we can solve \textsc{$k$-Subgraph Isomorphism} on $G$ in $f_k(k) \cdot n$ time.
Finally, the total running time is bounded by $f(k) \cdot n \log n + O(m)$ time for some function $f$.
\end{proof}

\subsection{String graphs}

String graphs are intersection graphs of arbitrary connected geometric objects in the plane.
Unfortunately, string graphs are not of clique-dependent polynomial expansion.
However, these graphs are ``biclique-dependent'' polynomial expansion: any subclass of string graphs that has bounded biclique-size is of polynomial expansion.

For a graph $G$ and an integer $k \geq 1$, a \textit{$k$-biclique decomposition} of $G$ is a partition of $V(G)$ into $V_0,V_1,\dots,V_r$ such that $\omega'(G[V_0])<k$ and $G[V_i]$ is a $(k,k)$-biclique for all $i \in [r]$.
We show that $k$-biclique decompositions can be computed efficiently on graph classes of biclique-dependent polynomial expansion.
\begin{lemma} \label{lem-bicliquedecomp}
Let $\mathcal{G}$ be a hereditary graph class of biclique-dependent polynomial expansion.
Given a graph $G \in \mathcal{G}$ and an integer $k$, one can compute a $k$-biclique decomposition of $G$ in $f(k) \cdot n^2$ time for some function $f$, where $n = |V(G)|$.
\end{lemma}
\begin{proof}
For $p \in \mathbb{N}$, we write $\mathcal{G}_p = \{G \in \mathcal{G}: \omega'(G) \leq p\}$.
Since $\mathcal{G}$ is of biclique-dependent polynomial expansion, each $\mathcal{G}_p$ is a graph class of polynomial expansion.
So \textsc{Subgraph Isomorphism} (for a fixed subgraph) can be solved in linear time on $\mathcal{G}_p$~\cite{dvovrak2013testing}.
In other words, there exists a function $f_p$ such that one can find a subgraph of $q$ vertices in a given $G \in \mathcal{G}_p$ in $f_p(q) \cdot n$ time.
For convenience, we denote this \textsc{Subgraph Isomorphism} algorithm by $\mathbf{A}_p$.

We compute a $k$-biclique decomposition of a given graph $G \in \mathcal{G}$ as follows.
Suppose $V(G) = \{v_1,\dots,v_n\}$.
We consider each $v_i$ iteratively for $i=1,\dots,n$ and maintain a $k$-biclique decomposition of the induced subgraph $G[\{v_1,\dots,v_i\}]$.
Let $V_0,V_1,\dots,V_r$ be a $k$-biclique decomposition of $G[\{v_1,\dots,v_{i-1}\}]$.
To compute a $k$-biclique decomposition of $G[\{v_1,\dots,v_i\}]$, we simply call the algorithm $\mathbf{A}_k$ to find a $(k,k)$-biclique in $G[V_0 \cup \{v_i\}]$.
Note that $\mathbf{A}_k$ can be applied to $G[V_0 \cup \{v_i\}]$, and it takes $f_k(2k) \cdot n$ time.
Indeed, we have $\omega'(G[V_0 \cup \{v_i\}]) \leq \omega'(G[V_0])+1 \leq k$, since $\omega'(G[V_0])<k$ by the definition of a $k$-biclique decomposition.
If $G[V_0 \cup \{v_i\}]$ does not contain a $(k,k)$-biclique, then $V_0 \cup \{v_i\},V_1,\dots,V_r$ is a $k$-biclique decomposition of $G[\{v_1,\dots,v_i\}]$.
Otherwise, let $V_{r+1}$ be (the vertex set of) a $(k,k)$-biclique in $G[V_0 \cup \{v_i\}]$.
Observe that $v_i \in V_{r+1}$, for otherwise $G[V_0]$ contains a $(k,k)$-biclique, which contradicts the fact $\omega'(G[V_0])<k$.
Therefore, $\omega'(G[(V_0 \cup \{v_i\}) \backslash V_{r+1}])<k$, as $(V_0 \cup \{v_i\}) \backslash V_{r+1} \subseteq V_0$.
In this case, $(V_0 \cup \{v_i\}) \backslash V_{r+1},V_1,\dots,V_r,V_{r+1}$ is a $k$-biclique decomposition of $G[\{v_1,\dots,v_i\}]$.
After the last vertex $v_n$ is considered, we obtain a $k$-biclique decomposition of $G$.
The overall time complexity of the algorithm is $f(k) \cdot n^2$ for some function $k$, since each iteration takes linear time.
\end{proof}

Based on the biclique-decomposition algorithm, we can extend our lossy kernels and approximation schemes to string graphs, when the forbidden  graphs contain a bipartite graph.

\begin{restatable}{theorem}{kernelstr} \label{thm-kernelstr}
    Let $\mathcal{G}$ be the class of string graphs.
    For any fixed (finite) set $\mathcal{F}$ of connected graphs among which at least one is bipartite, \textsc{Subgraph Hitting} on $\mathcal{G}$ with forbidden list $\mathcal{F}$ admits a $(2+\varepsilon)$-approximation lossy kernel of size $f(\varepsilon) \cdot k$ for some function $f$.
    The kernelization algorithm runs in $g(\varepsilon) \cdot n^2$ time for some function $g$, , where $n$ is the number of vertices of the input graph $G \in \mathcal{G}$.
\end{restatable}
\begin{proof}
The proof is almost the same as that of Theorem~\ref{thm-kernelgeo}.
For $p \in \mathbb{N}$, we write $\mathcal{G}_p = \{G \in \mathcal{G}: \omega'(G) \leq p\}$.
Since $\mathcal{G}$ is of biclique-dependent polynomial expansion, $\mathcal{G}_p$ is of polynomial expansion for every $p \in \mathbb{N}$.
Let $G \in \mathcal{G}$ be a \textsc{Subgraph Hitting} instance (with forbidden list $\mathcal{F}$) and $\varepsilon > 0$ be the approximation factor.
As before, set $\gamma = \max_{F \in \mathcal{F}} |V(F)|$.
By assumption, there exists a bipartite graph $F \in \mathcal{F}$.
Define $t = \frac{2+\varepsilon}{\varepsilon} \cdot (\gamma-1)$.
We first apply the algorithm of Lemma~\ref{lem-bicliquedecomp} to compute a $t$-biclique decomposition $V_0,V_1,\dots,V_r$ of $G$.
We have $\omega'(G[V_0]) < t$.
It follows that $G[V_0] \in \mathcal{G}_t$, since $\mathcal{G}$ is hereditary.

We observe that if $S_0 \subseteq G[V_0]$ is a $c$-approximation solution for the instance $G[V_0]$, then $S = S_0 \cup (\bigcup_{i=1}^r V_i)$ is a $(2+\varepsilon)c$-approximation solution for the instance $G$.
Clearly, $S$ is a feasible solution for $G$.
Indeed, $G-S = G[V_0]-S_0$ and the latter contains no forbidden pattern in $\mathcal{F}$ as a subgraph since $S_0$ is a solution of $G[V_0]$.
It suffices to show that $|S| \leq (2+\varepsilon)c \cdot |S'|$ for any solution $S' \subseteq V(G)$.
First observe that $S' \cap V_0$ is a feasible solution of $G[V_0]$, for otherwise the graph $G[V_0] - S' \cap V_0$ contains a forbidden pattern in $\mathcal{F}$ as a subgraph and so does $G - S'$, which contradicts that $S'$ is a solution for $G$.
Therefore, $|S_0| \leq c \cdot |S' \cap V_0|$.
Now we consider the sets $V_1,\dots,V_r$, each of which forms a $(t,t)$-biclique in $G$.
For each $i \in [r]$, $S'$ must intersect one side of $G[V_i]$ with at least $t-(\gamma-1)$ vertices, for otherwise $G-S'$ contains a $(\gamma,\gamma)$-biclique and thus contains the bipartite forbidden pattern $F$.
It follows that $|S' \cap V_i| \geq t-(\gamma-1) = \frac{2t}{2+\varepsilon}$ for each $i \in [r]$.
So we have
\begin{equation*}
    |S| = |S_0| + 2rt \leq c \cdot |S' \cap V_0| + (2+\varepsilon) \cdot \sum_{i=1}^r |S' \cap V_i| = (2+\varepsilon)c \cdot |S'|.
\end{equation*}

Now we see that given a $c$-approximation solution for $G[V_0]$, one can compute in linear time a $(2+\varepsilon)c$-approximation solution for $G$.
So we can reduce the instance $G$ to the instance $G[V_0]$.
As $G[V_0] \in \mathcal{G}_t$ and $\mathcal{G}_t$ is of polynomial expansion, we can further apply Theorem~\ref{thm-kernel} to reduce $G[V_0]$ to a linear-size lossy kernel.
The kernel size is $f(\varepsilon,t) \cdot k$ and the running time is $g(\varepsilon,t) \cdot n^2$.
Note that $t$ depends only on $\varepsilon$.
Thus, we have the bounds in the theorem.
\end{proof}

\begin{restatable}{theorem}{schemestr} \label{thm-schemestr}
    Let $\mathcal{G}$ be the class of string graphs.
    \textsc{Bipartite Subgraph Hitting} on $\mathcal{G}$ admits a $(2+\varepsilon)$-approximation algorithm with running time $f(\varepsilon,\mathcal{F}) \cdot n^2$ for some function $f$, where $n$ is the number of vertices of the input graph $G \in \mathcal{G}$.
\end{restatable}
\begin{proof}
The proof is similar to that of Theorem~\ref{thm-kernelstr}.
For $p \in \mathbb{N}$, we write $\mathcal{G}_p = \{G \in \mathcal{G}: \omega(G) \leq p\}$.
Since $\mathcal{G}$ is of biclique-dependent polynomial expansion, $\mathcal{G}_p$ is of polynomial expansion for every $p \in \mathbb{N}$.
Let $(G,\mathcal{F})$ be a \textsc{Bipartite Subgraph Hitting} instance where $G  \in \mathcal{G}$, and $\varepsilon > 0$ be the approximation factor.
As before, set $\gamma = \max_{F \in \mathcal{F}} |V(F)|$.
Define $t = \frac{2+\varepsilon}{\varepsilon} \cdot (\gamma-1)$.
We apply Lemma~\ref{lem-cliquedecomp} to compute a $t$-biclique decomposition $V_0,V_1,\dots,V_r$ of $G$.
We have $\omega(G[V_0]) < t$ and $G[V_0] \in \mathcal{G}_t$.
Now apply Corollary~\ref{cor-lsearchsub} to compute $(1+\varepsilon)$-approximation $S_0 \subseteq V_0$ for $(G[V_0],\mathcal{F})$ in $f(\varepsilon,t) \cdot n$ time.
The same analysis as in the proof of Theorem~\ref{thm-schemestr} shows that $S = S_0 \cup (\bigcup_{i=1}^r V_i)$ is a $(2+\varepsilon)$-approximation solution for $(G,\mathcal{F})$.
As $t$ depends only on $\varepsilon$ and $\mathcal{F}$, the total running time is $f(\varepsilon,\mathcal{F}) \cdot n^2$.
\end{proof}

\section{Hitting shallow minors}
In this section, we show that the \textsc{(Induced) Shallow-Minor Hitting} problem can be reduced to \textsc{(Induced) Subgraph Hitting}, and thus Theorem~\ref{thm-main} as well as all its applications can be extended to \textsc{(Induced) Shallow-Minor Hitting}.
\begin{lemma}
    Given a finite set $\mathcal{F}$ of graphs and an integer $d \geq 0$, one can compute another finite set $\mathcal{F}'$ of graphs such that a graph contains a (induced) $d$-shallow minor in $\mathcal{F}$ iff it contains a (induced) subgraph in $\mathcal{F}'$.
\end{lemma}
\begin{proof}
    Let $\delta = \max_{F \in \mathcal{F}} |V(F)|$.
    We define $\mathcal{F}'$ as the set of all graphs $F'$ such that $|V(F')| \leq (d+1) \delta^2$ and $F'$ contains $F$ as an induced $d$-shallow minor for some $F \in \mathcal{F}$.
    Clearly, if a graph $G$ contains some $F' \in \mathcal{F}'$ as a (induced) subgraph, then it contains some $F \in \mathcal{F}$ as a (induced) $d$-shallow minor.
    For the other direction, assume $G$ contains some $F \in \mathcal{F}$ as a (induced) $d$-shallow minor.
    For a vertex $v \in V(F)$, denote by $G_v$ the connected subgraph of $G$ that is contracted to $v$.
    In each $G_v$, we pick a reference vertex $\xi_v \in V(G_v)$.
    The diameter of $G_v$ is at most $d$.
    Thus, for each edge $(u,v) \in E(F)$, there exists a path $\pi_{u,v} = (v_0,v_1,\dots,v_r)$ in $G$ between $\xi_u$ and $\xi_v$ such that \textbf{(i)} $\pi_{u,v}$ contains at most $2(d+1)$ vertices and \textbf{(ii)} there exists $i \in [r]$ such that $v_0,\dots,v_{i-1} \in V(G_u)$ and $v_i,\dots,v_r \in V(G_v)$.
    Let $F_0'$ be the subgraph of $G$ consisting of the paths $\pi_{u,v}$ for all $(u,v) \in E(F)$.
    We now add to $F_0'$ some extra edges.
    Specifically, for each $v \in V(F)$, we add to $F_0'$ all edges of $G_v$ with both endpoints in $V(F_0')$.
    Also, for each $(u,v) \in E(F)$, we add to $F_0'$ all edges of $G$ with one endpoint in $F_0' \cap G_u$ and the other endpoint in $F_0' \cap G_v$.
    Let $F'$ denote the resulting new graph, which is still a subgraph of $G$.
    As $|E(F)| \leq \delta^2/2$, we have $|V(F')| = |V(F_0')| \leq (d+1) \delta^2$.
    Furthermore, $F'$ contains $F$ as an induced $d$-shallow minor.
    To see this, notice that $F' \cap G_v$ is connected for all $v \in V(F)$, simply because $F_0' \cap G_v$ is connected.
    Furthermore, there is an edge in $F'$ between $F' \cap G_u$ and $F' \cap G_v$ iff $(u,v) \in E(F)$.
    Thus, $F$ can be obtained from $F'$ by contracting each $F' \cap G_v$ to a single vertex.
    It follows that $F' \in \mathcal{F}'$.
    Clearly, $F'$ is a subgraph of $G$.
    If $F$ is an \textit{induced} minor of $G$, we claim that $F'$ is an induced subgraph of $G$.
    Consider two vertices $x,y \in V(F')$ such that $(x,y) \in E(G)$.
    If $x,y \in G_v$ for some $v \in V(F)$, then $(x,y) \in E(F')$ by our construction.
    If $x \in G_u$ and $y \in G_v$ where $u \neq v$, then $(u,v) \in E(F)$ because $F$ is an induced minor of $G$.
    Again by our construction, $(x,y) \in E(F')$.
\end{proof}

Recall that in \textsc{(Induced) Shallow-Minor Hitting}, we are given a graph $G$, a forbidden list $\mathcal{F}$, an integer $d \geq 0$, and the goal is to compute a minimum set $S \subseteq V(G)$ such that $G-S$ does not contain any graph in $\mathcal{F}$ as a $d$-shallow minor.
Consider an instance $(G,\mathcal{F},d)$ of \textsc{(Induced) Shallow-Minor Hitting}.
We can use the above lemma to compute a set $\mathcal{F}'$ of graphs, which only depends on $\mathcal{F}$ and $d$.
Now a subset $S \subseteq V(G)$ is a solution of the \textsc{(Induced) Shallow-Minor Hitting} instance $(G,\mathcal{F},d)$ iff it is a solution of the \textsc{(Induced) Subgraph Hitting} instance $(G,\mathcal{F}')$.
On the other hand, there is a trivial reduction from \textsc{(Induced) Subgraph Hitting} to \textsc{(Induced) Shallow-Minor Hitting}, because (induced) subgraphs are exactly (induced) 0-shallow minors.
As such, Theorem~\ref{thm-main} can be extended to \textsc{(Induced) Shallow-Minor Hitting}.
\begin{theorem}
Let $\mathcal{G}$ be any hereditary graph class of bounded expansion.
If \textsc{(Induced) Shallow-Minor Hitting} on $\mathcal{G}$ admits an approximation scheme with running time $f_0(\varepsilon,\mathcal{F},d,\Delta) \cdot n^c$ for some function $f_0$, then the same problem also admits an approximation scheme with running time $f(\varepsilon,\mathcal{F},d) \cdot n^c$ for some function $f$, where $n$ is the number of vertices in the input graph $G \in \mathcal{G}$ and $\Delta$ is the maximum degree of $G$.
\end{theorem}

\noindent
Also, all of our algorithmic results derived from Theorem~\ref{thm-main} can be extended to \textsc{(Induced) Shallow-Minor Hitting}, with the same running time.

\section*{Acknowledgement}
The authors would like to thank anonymous reviewers for their detailed comments on an earlier version of this paper.
Specifically, the comments of one reviewer helped us substantially improve the running time of our approximation schemes.

\bibliographystyle{plainurl}
\bibliography{my_bib.bib}

\begin{thebibliography}{100}

\bibitem{aprile2021tight}
Manuel Aprile, Matthew Drescher, Samuel Fiorini, and Tony Huynh.
\newblock A tight approximation algorithm for the cluster vertex deletion
  problem.
\newblock In {\em Integer Programming and Combinatorial Optimization: 22nd
  International Conference, IPCO 2021, Atlanta, GA, USA, May 19--21, 2021,
  Proceedings 22}, pages 340--353. Springer, 2021.

\bibitem{austrin2011inapproximability}
Per Austrin, Subhash Khot, and Muli Safra.
\newblock Inapproximability of vertex cover and independent set in bounded
  degree graphs.
\newblock {\em Theory of Computing}, 7(1):27--43, 2011.

\bibitem{DBLP:journals/corr/BaiTS16}
Zongwen Bai, Jianhua Tu, and Yongtang Shi.
\newblock An improved algorithm for the vertex cover p\({}_{\mbox{3}}\) problem
  on graphs of bounded treewidth.
\newblock {\em CoRR}, abs/1603.09448, 2016.
\newblock URL: \url{http://arxiv.org/abs/1603.09448}.

\bibitem{Baker94}
Brenda~S. Baker.
\newblock Approximation algorithms for np-complete problems on planar graphs.
\newblock {\em J. {ACM}}, 41(1):153--180, 1994.
\newblock URL: \url{http://doi.acm.org/10.1145/174644.174650}, \href
  {https://doi.org/10.1145/174644.174650} {\path{doi:10.1145/174644.174650}}.

\bibitem{balasubramanian1998improved}
R~Balasubramanian, Michael~R Fellows, and Venkatesh Raman.
\newblock An improved fixed-parameter algorithm for vertex cover.
\newblock {\em Information Processing Letters}, 65(3):163--168, 1998.

\bibitem{bar1985local}
Reuven Bar-Yehuda and Shimon Even.
\newblock A local-ratio theorem for approximating the weighted vertex cover
  problem.
\newblock {\em Annals of Discrete Mathematics}, 25(27-46):50, 1985.

\bibitem{bar2011minimum}
Reuven Bar-Yehuda, Danny Hermelin, and Dror Rawitz.
\newblock Minimum vertex cover in rectangle graphs.
\newblock {\em Computational Geometry}, 44(6-7):356--364, 2011.

\bibitem{betzler2012bounded}
Nadja Betzler, Robert Bredereck, Rolf Niedermeier, and Johannes Uhlmann.
\newblock On bounded-degree vertex deletion parameterized by treewidth.
\newblock {\em Discrete Applied Mathematics}, 160(1-2):53--60, 2012.

\bibitem{DBLP:journals/arscom/BoliacCL04}
Rodica Boliac, Kathie Cameron, and Vadim~V. Lozin.
\newblock On computing the dissociation number and the induced matching number
  of bipartite graphs.
\newblock {\em Ars Comb.}, 72, 2004.

\bibitem{boral2016fast}
Anudhyan Boral, Marek Cygan, Tomasz Kociumaka, and Marcin Pilipczuk.
\newblock A fast branching algorithm for cluster vertex deletion.
\newblock {\em Theory of Computing Systems}, 58:357--376, 2016.

\bibitem{bougeret2019much}
Marin Bougeret and Ignasi Sau.
\newblock How much does a treedepth modulator help to obtain polynomial kernels
  beyond sparse graphs?
\newblock {\em Algorithmica}, 81(10):4043--4068, 2019.

\bibitem{DBLP:journals/dam/BrauseK17}
Christoph Brause and Rastislav Krivos{-}Bellus.
\newblock On a relation between k-path partition and k-path vertex cover.
\newblock {\em Discret. Appl. Math.}, 223:28--38, 2017.

\bibitem{brevsar2013vertex}
Bo{\v{s}}tjan Bre{\v{s}}ar, Marko Jakovac, J{\'a}n Katreni{\v{c}}, Gabriel
  Semani{\v{s}}in, and Andrej Taranenko.
\newblock On the vertex k-path cover.
\newblock {\em Discrete Applied Mathematics}, 161(13-14):1943--1949, 2013.

\bibitem{DBLP:journals/amc/BresarKKS19}
Bostjan Bresar, Tim Kos, Rastislav Krivos{-}Bellus, and Gabriel Semanisin.
\newblock Hitting subgraphs in \emph{P}\({}_{\mbox{4}}\)-tidy graphs.
\newblock {\em Appl. Math. Comput.}, 352:211--219, 2019.

\bibitem{brustle2021approximation}
Noah Br{\"u}stle, Tal Elbaz, Hamed Hatami, Onur Kocer, and Bingchan Ma.
\newblock Approximation algorithms for hitting subgraphs.
\newblock In {\em Combinatorial Algorithms: 32nd International Workshop, IWOCA
  2021, Ottawa, ON, Canada, July 5--7, 2021, Proceedings 32}, pages 370--384.
  Springer, 2021.

\bibitem{buss1993nondeterminism}
Jonathan~F Buss and Judy Goldsmith.
\newblock Nondeterminism within p\^{}.
\newblock {\em SIAM Journal on Computing}, 22(3):560--572, 1993.

\bibitem{cai2015incompressibility}
Leizhen Cai and Yufei Cai.
\newblock Incompressibility of h-free edge modification problems.
\newblock {\em Algorithmica}, 71(3):731--757, 2015.

\bibitem{DBLP:journals/mp/CameronH06}
Kathie Cameron and Pavol Hell.
\newblock Independent packings in structured graphs.
\newblock {\em Math. Program.}, 105(2-3):201--213, 2006.

\bibitem{DBLP:journals/corr/abs-2107-12245}
Radovan Cerven{\'{y}}, Pratibha Choudhary, and Ondrej Such{\'{y}}.
\newblock On kernels for d-path vertex cover.
\newblock {\em CoRR}, abs/2107.12245, 2021.

\bibitem{DBLP:conf/mfcs/CervenyS19}
Radovan Cerven{\'{y}} and Ondrej Such{\'{y}}.
\newblock Faster {FPT} algorithm for 5-path vertex cover.
\newblock In {\em 44th International Symposium on Mathematical Foundations of
  Computer Science, {MFCS} 2019, August 26-30, 2019, Aachen, Germany}, pages
  32:1--32:13, 2019.

\bibitem{DBLP:journals/corr/abs-2111-05896}
Radovan Cerven{\'{y}} and Ondrej Such{\'{y}}.
\newblock Generating faster algorithms for d-path vertex cover.
\newblock {\em CoRR}, abs/2111.05896, 2021.
\newblock URL: \url{https://arxiv.org/abs/2111.05896}.

\bibitem{chan2023finding}
Timothy~M. Chan.
\newblock Finding triangles and other small subgraphs in geometric intersection
  graphs.
\newblock In {\em 34rd Annual ACM-SIAM Symposium on Discrete Algorithms}. SIAM,
  2023.

\bibitem{chandran2004refined}
L~Sunil Chandran and Fabrizio Grandoni.
\newblock Refined memorisation for vertex cover.
\newblock In {\em Parameterized and Exact Computation: First International
  Workshop, IWPEC 2004, Bergen, Norway, September 14-17, 2004. Proceedings 1},
  pages 61--70. Springer, 2004.

\bibitem{chang20141}
Maw-Shang Chang, Li-Hsuan Chen, Ling-Ju Hung, Yi-Zhi Liu, Peter Rossmanith, and
  Somnath Sikdar.
\newblock An $o^*(1.4658^n)$-time exact algorithm for the maximum
  bounded-degree-1 set problem.
\newblock In {\em Proceedings of the 31st workshop on combinatorial mathematics
  and computation theory}, pages 9--18, 2014.

\bibitem{2018moderately}
Maw-Shang Chang, Li-Hsuan Chen, Ling-Ju Hung, Yi-Zhi Liu, Peter Rossmanith, and
  Somnath Sikdar.
\newblock Moderately exponential time algorithms for the maximum
  bounded-degree-1 set problem.
\newblock {\em Discrete Applied Mathematics}, 251:114--125, 2018.

\bibitem{DBLP:conf/stacs/ChenFKX05}
Jianer Chen, Henning Fernau, Iyad~A. Kanj, and Ge~Xia.
\newblock Parametric duality and kernelization: Lower bounds and upper bounds
  on kernel size.
\newblock In Volker Diekert and Bruno Durand, editors, {\em {STACS} 2005, 22nd
  Annual Symposium on Theoretical Aspects of Computer Science, Stuttgart,
  Germany, February 24-26, 2005, Proceedings}, volume 3404 of {\em Lecture
  Notes in Computer Science}, pages 269--280. Springer, 2005.

\bibitem{chen2001vertex}
Jianer Chen, Iyad~A Kanj, and Weijia Jia.
\newblock Vertex cover: further observations and further improvements.
\newblock {\em Journal of Algorithms}, 41(2):280--301, 2001.

\bibitem{DBLP:journals/tcs/ChenKX10}
Jianer Chen, Iyad~A. Kanj, and Ge~Xia.
\newblock Improved upper bounds for vertex cover.
\newblock {\em Theor. Comput. Sci.}, 411(40-42):3736--3756, 2010.

\bibitem{DBLP:journals/dam/ChlebikC08}
Miroslav Chleb{\'{\i}}k and Janka Chleb{\'{\i}}kov{\'{a}}.
\newblock Crown reductions for the minimum weighted vertex cover problem.
\newblock {\em Discret. Appl. Math.}, 156(3):292--312, 2008.

\bibitem{CyganFKLMPPS15}
Marek Cygan, Fedor~V. Fomin, Lukasz Kowalik, Daniel Lokshtanov, D{\'{a}}niel
  Marx, Marcin Pilipczuk, Michal Pilipczuk, and Saket Saurabh.
\newblock {\em Parameterized Algorithms}.
\newblock Springer, 2015.
\newblock \href {https://doi.org/10.1007/978-3-319-21275-3}
  {\path{doi:10.1007/978-3-319-21275-3}}.

\bibitem{cygan2017hitting}
Marek Cygan, D{\'a}niel Marx, Marcin Pilipczuk, and Micha{\l} Pilipczuk.
\newblock Hitting forbidden subgraphs in graphs of bounded treewidth.
\newblock {\em Information and Computation}, 256:62--82, 2017.

\bibitem{cygan2017polynomial}
Marek Cygan, Marcin Pilipczuk, Micha{\l} Pilipczuk, Erik~Jan Van~Leeuwen, and
  Marcin Wrochna.
\newblock Polynomial kernelization for removing induced claws and diamonds.
\newblock {\em Theory of Computing Systems}, 60:615--636, 2017.

\bibitem{dawar2006approximation}
Anuj Dawar, Martin Grohe, Stephan Kreutzer, and Nicole Schweikardt.
\newblock Approximation schemes for first-order definable optimisation
  problems.
\newblock In {\em 21st Annual IEEE Symposium on Logic in Computer Science
  (LICS'06)}, pages 411--420. IEEE, 2006.

\bibitem{de2020framework}
Mark De~Berg, Hans~L Bodlaender, S{\'a}ndor Kisfaludi-Bak, D{\'a}niel Marx, and
  Tom~C Van Der~Zanden.
\newblock A framework for exponential-time-hypothesis--tight algorithms and
  lower bounds in geometric intersection graphs.
\newblock {\em SIAM Journal on Computing}, 49(6):1291--1331, 2020.

\bibitem{dell2014satisfiability}
Holger Dell and Dieter Van~Melkebeek.
\newblock Satisfiability allows no nontrivial sparsification unless the
  polynomial-time hierarchy collapses.
\newblock {\em Journal of the ACM (JACM)}, 61(4):1--27, 2014.

\bibitem{demaine2005subexponential}
Erik~D. Demaine, Fedor~V. Fomin, Mohammad~Taghi Hajiaghayi, and Dimitrios~M.
  Thilikos.
\newblock Subexponential parameterized algorithms on bounded-genus graphs and
  \emph{H}-minor-free graphs.
\newblock {\em J. {ACM}}, 52(6):866--893, 2005.
\newblock \href {https://doi.org/10.1145/1101821.1101823}
  {\path{doi:10.1145/1101821.1101823}}.

\bibitem{DBLP:conf/esa/DemaineGKLLSVP19}
Erik~D. Demaine, Timothy~D. Goodrich, Kyle Kloster, Brian Lavallee, Quanquan~C.
  Liu, Blair~D. Sullivan, Ali Vakilian, and Andrew van~der Poel.
\newblock Structural rounding: Approximation algorithms for graphs near an
  algorithmically tractable class.
\newblock In {\em 27th Annual European Symposium on Algorithms, {ESA} 2019,
  September 9-11, 2019, Munich/Garching, Germany}, pages 37:1--37:15, 2019.

\bibitem{DBLP:conf/isaac/DessmarkJL93}
Anders Dessmark, Klaus Jansen, and Andrzej Lingas.
\newblock The maximum k-dependent and f-dependent set problem.
\newblock In {\em Algorithms and Computation, 4th International Symposium,
  {ISAAC} '93, Hong Kong, December 15-17, 1993, Proceedings}, pages 88--98,
  1993.

\bibitem{devi2015computational}
N~Safina Devi, Aniket~C Mane, and Sounaka Mishra.
\newblock Computational complexity of minimum p4 vertex cover problem for
  regular and k1, 4-free graphs.
\newblock {\em Discrete Applied Mathematics}, 184:114--121, 2015.

\bibitem{dinur2005hardness}
Irit Dinur and Samuel Safra.
\newblock On the hardness of approximating minimum vertex cover.
\newblock {\em Annals of mathematics}, pages 439--485, 2005.

\bibitem{drange2016computational}
P{\aa}l~Gr{\o}n{\aa}s Drange, Markus Dregi, and Pim van’t Hof.
\newblock On the computational complexity of vertex integrity and component
  order connectivity.
\newblock {\em Algorithmica}, 76(4):1181--1202, 2016.

\bibitem{dvovrak2013constant}
Zden{\v{e}}k Dvo{\v{r}}{\'a}k.
\newblock Constant-factor approximation of the domination number in sparse
  graphs.
\newblock {\em European Journal of Combinatorics}, 34(5):833--840, 2013.

\bibitem{dvovrak2020baker}
Zden{\v{e}}k Dvo{\v{r}}{\'a}k.
\newblock Baker game and polynomial-time approximation schemes.
\newblock In {\em Proceedings of the Fourteenth Annual ACM-SIAM Symposium on
  Discrete Algorithms}, pages 2227--2240. SIAM, 2020.

\bibitem{dvovrak2021approximation}
Zden{\v{e}}k Dvo{\v{r}}{\'a}k.
\newblock Approximation metatheorems for classes with bounded expansion.
\newblock {\em arXiv preprint arXiv:2103.08698}, 2021.

\bibitem{dvovrak2013testing}
Zden{\v{e}}k Dvo{\v{r}}{\'a}k, Daniel Kr{\'a}l, and Robin Thomas.
\newblock Testing first-order properties for subclasses of sparse graphs.
\newblock {\em Journal of the ACM (JACM)}, 60(5):1--24, 2013.

\bibitem{dvovrak2021approximation2}
Zden{\v{e}}k Dvo{\v{r}}{\'a}k and Abhiruk Lahiri.
\newblock Approximation schemes for bounded distance problems on fractionally
  treewidth-fragile graphs.
\newblock {\em arXiv preprint arXiv:2105.01780}, 2021.

\bibitem{DorakN2016}
Zden\v{e}k Dvo\v{r}\'{a}k and Sergey Norin.
\newblock Strongly sublinear separators and polynomial expansion.
\newblock {\em SIAM Journal on Discrete Mathematics}, 30(2):1095--1101, 2016.
\newblock \href {http://arxiv.org/abs/https://doi.org/10.1137/15M1017569}
  {\path{arXiv:https://doi.org/10.1137/15M1017569}}, \href
  {https://doi.org/10.1137/15M1017569} {\path{doi:10.1137/15M1017569}}.

\bibitem{DBLP:conf/iwpec/EibenL020}
Eduard Eiben, William Lochet, and Saket Saurabh.
\newblock A polynomial kernel for paw-free editing.
\newblock In {\em 15th International Symposium on Parameterized and Exact
  Computation, {IPEC} 2020, December 14-18, 2020, Hong Kong, China (Virtual
  Conference)}, pages 10:1--10:15, 2020.

\bibitem{eiben2022lossy}
Eduard Eiben, Diptapriyo Majumdar, and MS~Ramanujan.
\newblock On the lossy kernelization for connected treedepth deletion set.
\newblock In {\em International Workshop on Graph-Theoretic Concepts in
  Computer Science}, pages 201--214. Springer, 2022.

\bibitem{enright2018deleting}
Jessica Enright and Kitty Meeks.
\newblock Deleting edges to restrict the size of an epidemic: a new application
  for treewidth.
\newblock {\em Algorithmica}, 80(6):1857--1889, 2018.

\bibitem{ErlebachJS05}
Thomas Erlebach, Klaus Jansen, and Eike Seidel.
\newblock Polynomial-time approximation schemes for geometric intersection
  graphs.
\newblock {\em {SIAM} J. Comput.}, 34(6):1302--1323, 2005.
\newblock \href {https://doi.org/10.1137/S0097539702402676}
  {\path{doi:10.1137/S0097539702402676}}.

\bibitem{DBLP:journals/jcss/FellowsGMN11}
Michael~R. Fellows, Jiong Guo, Hannes Moser, and Rolf Niedermeier.
\newblock A generalization of nemhauser and trotter's local optimization
  theorem.
\newblock {\em J. Comput. Syst. Sci.}, 77(6):1141--1158, 2011.

\bibitem{fiorini2016improved}
Samuel Fiorini, Gwena{\"e}l Joret, and Oliver Schaudt.
\newblock Improved approximation algorithms for hitting 3-vertex paths.
\newblock In {\em Integer Programming and Combinatorial Optimization: 18th
  International Conference, IPCO 2016, Li{\`e}ge, Belgium, June 1-3, 2016,
  Proceedings 18}, pages 238--249. Springer, 2016.

\bibitem{fiorini2020improved}
Samuel Fiorini, Gwena{\"e}l Joret, and Oliver Schaudt.
\newblock Improved approximation algorithms for hitting 3-vertex paths.
\newblock {\em Mathematical Programming}, 182(1-2):355--367, 2020.

\bibitem{fiorini2018approximability}
Samuel Fiorini, R~Krithika, NS~Narayanaswamy, and Venkatesh Raman.
\newblock Approximability of clique transversal in perfect graphs.
\newblock {\em Algorithmica}, 80(8):2221--2239, 2018.

\bibitem{DBLP:journals/talg/FominLLSTZ19}
Fedor~V. Fomin, Tien{-}Nam Le, Daniel Lokshtanov, Saket Saurabh, St{\'{e}}phan
  Thomass{\'{e}}, and Meirav Zehavi.
\newblock Subquadratic kernels for implicit 3-hitting set and 3-set packing
  problems.
\newblock {\em {ACM} Trans. Algorithms}, 15(1):13:1--13:44, 2019.

\bibitem{DBLP:conf/icalp/FominLP0Z19}
Fedor~V. Fomin, Daniel Lokshtanov, Fahad Panolan, Saket Saurabh, and Meirav
  Zehavi.
\newblock Decomposition of map graphs with applications.
\newblock In {\em 46th International Colloquium on Automata, Languages, and
  Programming, {ICALP} 2019, July 9-12, 2019, Patras, Greece}, pages
  60:1--60:15, 2019.

\bibitem{FominLPSZ19}
Fedor~V. Fomin, Daniel Lokshtanov, Fahad Panolan, Saket Saurabh, and Meirav
  Zehavi.
\newblock Finding, hitting and packing cycles in subexponential time on unit
  disk graphs.
\newblock {\em Discret. Comput. Geom.}, 62(4):879--911, 2019.
\newblock \href {https://doi.org/10.1007/s00454-018-00054-x}
  {\path{doi:10.1007/s00454-018-00054-x}}.

\bibitem{DBLP:conf/soda/FominLS12}
Fedor~V. Fomin, Daniel Lokshtanov, and Saket Saurabh.
\newblock Bidimensionality and geometric graphs.
\newblock In {\em Proceedings of the Twenty-Third Annual {ACM-SIAM} Symposium
  on Discrete Algorithms, {SODA} 2012, Kyoto, Japan, January 17-19, 2012},
  pages 1563--1575, 2012.

\bibitem{FominLS18}
Fedor~V. Fomin, Daniel Lokshtanov, and Saket Saurabh.
\newblock Excluded grid minors and efficient polynomial-time approximation
  schemes.
\newblock {\em J. {ACM}}, 65(2):10:1--10:44, 2018.
\newblock \href {https://doi.org/10.1145/3154833} {\path{doi:10.1145/3154833}}.

\bibitem{fomin2019kernelization}
Fedor~V Fomin, Daniel Lokshtanov, Saket Saurabh, and Meirav Zehavi.
\newblock {\em Kernelization: theory of parameterized preprocessing}.
\newblock Cambridge University Press, 2019.

\bibitem{fox2010separator}
Jacob Fox and J{\'a}nos Pach.
\newblock A separator theorem for string graphs and its applications.
\newblock {\em Combinatorics, Probability and Computing}, 19(3):371--390, 2010.

\bibitem{fox2011computing}
Jacob Fox and J{\'a}nos Pach.
\newblock Computing the independence number of intersection graphs.
\newblock In {\em Proceedings of the twenty-second annual ACM-SIAM Symposium on
  Discrete algorithms}, pages 1161--1165. SIAM, 2011.

\bibitem{gaikwad2022further}
Ajinkya Gaikwad and Soumen Maity.
\newblock Further parameterized algorithms for the f-free edge deletion
  problem.
\newblock {\em Theoretical Computer Science}, 933:125--137, 2022.

\bibitem{gajarsky2017kernelization}
Jakub Gajarsk{\`y}, Petr Hlin{\v{e}}n{\`y}, Jan Obdr{\v{z}}{\'a}lek, Sebastian
  Ordyniak, Felix Reidl, Peter Rossmanith, Fernando~S{\'a}nchez Villaamil, and
  Somnath Sikdar.
\newblock Kernelization using structural parameters on sparse graph classes.
\newblock {\em Journal of Computer and System Sciences}, 84:219--242, 2017.

\bibitem{galby2023polynomial}
Esther Galby, Andrea Munaro, and Shizhou Yang.
\newblock Polynomial-time approximation schemes for independent packing
  problems on fractionally tree-independence-number-fragile graphs.
\newblock {\em To appear in the 39th International Symposium on Computational
  Geometry (SoCG)}, 2023.

\bibitem{ganian2021structural}
Robert Ganian, Fabian Klute, and Sebastian Ordyniak.
\newblock On structural parameterizations of the bounded-degree vertex deletion
  problem.
\newblock {\em Algorithmica}, 83(1):297--336, 2021.

\bibitem{goldmann2021parameterized}
Lito~Julius Goldmann.
\newblock {\em Parameterized Complexity of Modifying Graphs to be
  Biclique-free}.
\newblock PhD thesis, Institute of Software, 2021.

\bibitem{grohe2017deciding}
Martin Grohe, Stephan Kreutzer, and Sebastian Siebertz.
\newblock Deciding first-order properties of nowhere dense graphs.
\newblock {\em Journal of the ACM (JACM)}, 64(3):1--32, 2017.

\bibitem{groshaus2009cycle}
Marina Groshaus, Pavol Hell, Sulamita Klein, Loana~Tito Nogueira, and F{\'a}bio
  Protti.
\newblock Cycle transversals in bounded degree graphs.
\newblock {\em Electronic Notes in Discrete Mathematics}, 35:189--195, 2009.

\bibitem{gross2013survey}
Daniel Gross, Monika Heinig, Lakshmi Iswara, L~William Kazmierczak, Kristi
  Luttrell, John~T Saccoman, and Charles Suffel.
\newblock A survey of component order connectivity models of graph theoretic
  networks.
\newblock {\em WSEAS Transactions on Mathematics}, 12:895--910, 2013.

\bibitem{gupta2019losing}
Anupam Gupta, Euiwoong Lee, Jason Li, Pasin Manurangsi, and Micha{\l}
  W{\l}odarczyk.
\newblock Losing treewidth by separating subsets.
\newblock In {\em Proceedings of the Thirtieth Annual ACM-SIAM Symposium on
  Discrete Algorithms}, pages 1731--1749. SIAM, 2019.

\bibitem{halperin2002improved}
Eran Halperin.
\newblock Improved approximation algorithms for the vertex cover problem in
  graphs and hypergraphs.
\newblock {\em SIAM Journal on Computing}, 31(5):1608--1623, 2002.

\bibitem{HarPeledQ16}
Sariel Har{-}Peled and Kent Quanrud.
\newblock Notes on approximation algorithms for polynomial-expansion and
  low-density graphs.
\newblock {\em CoRR}, abs/1603.03098, 2016.
\newblock URL: \url{http://arxiv.org/abs/1603.03098}, \href
  {http://arxiv.org/abs/1603.03098} {\path{arXiv:1603.03098}}.

\bibitem{har2017approximation}
Sariel Har-Peled and Kent Quanrud.
\newblock Approximation algorithms for polynomial-expansion and low-density
  graphs.
\newblock {\em SIAM Journal on Computing}, 46(6):1712--1744, 2017.

\bibitem{harel1984fast}
Dov Harel and Robert~Endre Tarjan.
\newblock Fast algorithms for finding nearest common ancestors.
\newblock {\em siam Journal on Computing}, 13(2):338--355, 1984.

\bibitem{hsieh2021d}
Sun-Yuan Hsieh, Van~Bang Le, and Sheng-Lung Peng.
\newblock On the d-claw vertex deletion problem.
\newblock In {\em Computing and Combinatorics: 27th International Conference,
  COCOON 2021, Tainan, Taiwan, October 24--26, 2021, Proceedings 27}, pages
  591--603. Springer, 2021.

\bibitem{huffner2010fixed}
Falk H{\"u}ffner, Christian Komusiewicz, Hannes Moser, and Rolf Niedermeier.
\newblock Fixed-parameter algorithms for cluster vertex deletion.
\newblock {\em Theory of Computing Systems}, 47(1):196--217, 2010.

\bibitem{karakostas2009better}
George Karakostas.
\newblock A better approximation ratio for the vertex cover problem.
\newblock {\em ACM Transactions on Algorithms (TALG)}, 5(4):1--8, 2009.

\bibitem{kardovs2011computing}
Franti{\v{s}}ek Kardo{\v{s}}, J{\'a}n Katreni{\v{c}}, and Ingo Schiermeyer.
\newblock On computing the minimum 3-path vertex cover and dissociation number
  of graphs.
\newblock {\em Theoretical Computer Science}, 412(50):7009--7017, 2011.

\bibitem{karpinski1996approximating}
Marek Karpinski and Alexander Zelikovsky.
\newblock {\em Approximating dense cases of covering problems}.
\newblock Citeseer, 1996.

\bibitem{kedem1986union}
Klara Kedem, Ron Livne, J{\'a}nos Pach, and Micha Sharir.
\newblock On the union of jordan regions and collision-free translational
  motion amidst polygonal obstacles.
\newblock {\em Discrete \& Computational Geometry}, 1(1):59--71, 1986.

\bibitem{khot2008vertex}
Subhash Khot and Oded Regev.
\newblock Vertex cover might be hard to approximate to within 2- $\varepsilon$.
\newblock {\em Journal of Computer and System Sciences}, 74(3):335--349, 2008.

\bibitem{kumar20172lk}
Mithilesh Kumar and Daniel Lokshtanov.
\newblock A 2lk kernel for l-component order connectivity.
\newblock In {\em 11th International Symposium on Parameterized and Exact
  Computation (IPEC 2016)}. Schloss Dagstuhl-Leibniz-Zentrum fuer Informatik,
  2017.

\bibitem{kumar2014approximation}
Mrinal Kumar, Sounaka Mishra, N~Safina Devi, and Saket Saurabh.
\newblock Approximation algorithms for node deletion problems on bipartite
  graphs with finite forbidden subgraph characterization.
\newblock {\em Theoretical Computer Science}, 526:90--96, 2014.

\bibitem{le2022complexity}
Hoang-Oanh Le and Van~Bang Le.
\newblock Complexity of the cluster vertex deletion problem on h-free graphs.
\newblock In {\em 47th International Symposium on Mathematical Foundations of
  Computer Science (MFCS 2022)}. Schloss Dagstuhl-Leibniz-Zentrum f{\"u}r
  Informatik, 2022.

\bibitem{le2022greedy}
Hung Le and Cuong Than.
\newblock Greedy spanners in euclidean spaces admit sublinear separators.
\newblock In {\em Proceedings of the 2022 Annual ACM-SIAM Symposium on Discrete
  Algorithms (SODA)}, pages 3287--3310. SIAM, 2022.

\bibitem{lee2017partitioning}
Euiwoong Lee.
\newblock Partitioning a graph into small pieces with applications to path
  transversal.
\newblock In {\em Proceedings of the Twenty-Eighth Annual ACM-SIAM Symposium on
  Discrete Algorithms}, pages 1546--1558. SIAM, 2017.

\bibitem{li2022improved}
Wenjun Li, Huan Peng, and Yongjie Yang.
\newblock Improved kernel and algorithm for claw and diamond free edge deletion
  based on refined observations.
\newblock {\em Theoretical Computer Science}, 906:83--93, 2022.

\bibitem{LiptonTargenSep}
Richard~J. Lipton and Robert~Endre Tarjan.
\newblock Applications of a planar separator theorem.
\newblock {\em SIAM Journal on Computing}, 9(3):615--627, 1980.
\newblock \href {http://arxiv.org/abs/https://doi.org/10.1137/0209046}
  {\path{arXiv:https://doi.org/10.1137/0209046}}, \href
  {https://doi.org/10.1137/0209046} {\path{doi:10.1137/0209046}}.

\bibitem{DBLP:conf/soda/LokshtanovMRSZ21}
Daniel Lokshtanov, Pranabendu Misra, M.~S. Ramanujan, Saket Saurabh, and Meirav
  Zehavi.
\newblock Fpt-approximation for {FPT} problems.
\newblock In {\em Proceedings of the 2021 {ACM-SIAM} Symposium on Discrete
  Algorithms, {SODA} 2021, Virtual Conference, January 10 - 13, 2021}, pages
  199--218, 2021.

\bibitem{DBLP:conf/stoc/LokshtanovPRS17}
Daniel Lokshtanov, Fahad Panolan, M.~S. Ramanujan, and Saket Saurabh.
\newblock Lossy kernelization.
\newblock In Hamed Hatami, Pierre McKenzie, and Valerie King, editors, {\em
  Proceedings of the 49th Annual {ACM} {SIGACT} Symposium on Theory of
  Computing, {STOC} 2017, Montreal, QC, Canada, June 19-23, 2017}, pages
  224--237. {ACM}, 2017.
\newblock \href {https://doi.org/10.1145/3055399.3055456}
  {\path{doi:10.1145/3055399.3055456}}.

\bibitem{lokshtanov2022subexponential}
Daniel Lokshtanov, Fahad Panolan, Saket Saurabh, Jie Xue, and Meirav Zehavi.
\newblock Subexponential parameterized algorithms on disk graphs (extended
  abstract).
\newblock In {\em Proceedings of the 2022 Annual ACM-SIAM Symposium on Discrete
  Algorithms (SODA)}, pages 2005--2031. SIAM, 2022.

\bibitem{lokshtanov2023framework}
Daniel Lokshtanov, Fahad Panolan, Saket Saurabh, Jie Xue, and Meirav Zehavi.
\newblock A framework for approximation schemes on disk graphs.
\newblock In {\em 34rd Annual ACM-SIAM Symposium on Discrete Algorithms}. SIAM,
  2023.

\bibitem{DBLP:journals/ipl/LozinR03}
Vadim~V. Lozin and Dieter Rautenbach.
\newblock Some results on graphs without long induced paths.
\newblock {\em Inf. Process. Lett.}, 88(4):167--171, 2003.

\bibitem{lund1993approximation}
Carsten Lund and Mihalis Yannakakis.
\newblock The approximation of maximum subgraph problems.
\newblock In {\em International Colloquium on Automata, Languages, and
  Programming}, pages 40--51. Springer, 1993.

\bibitem{marx2022incompressibility}
D{\'a}niel Marx and RB~Sandeep.
\newblock Incompressibility of h-free edge modification problems: Towards a
  dichotomy.
\newblock {\em Journal of Computer and System Sciences}, 125:25--58, 2022.

\bibitem{matula1983smallest}
David~W Matula and Leland~L Beck.
\newblock Smallest-last ordering and clustering and graph coloring algorithms.
\newblock {\em Journal of the ACM (JACM)}, 30(3):417--427, 1983.

\bibitem{mezei2023ptas}
Bal{\'a}zs~F Mezei, Marcin Wrochna, and Stanislav {\v{Z}}ivn{\`y}.
\newblock Ptas for sparse general-valued csps.
\newblock {\em ACM Transactions on Algorithms}, 19(2):1--31, 2023.

\bibitem{miller1997separators}
Gary~L Miller, Shang-Hua Teng, William Thurston, and Stephen~A Vavasis.
\newblock Separators for sphere-packings and nearest neighbor graphs.
\newblock {\em Journal of the ACM (JACM)}, 44(1):1--29, 1997.

\bibitem{monien1985ramsey}
Burkhard Monien and Ewald Speckenmeyer.
\newblock Ramsey numbers and an approximation algorithm for the vertex cover
  problem.
\newblock {\em Acta Informatica}, 22(1):115--123, 1985.

\bibitem{DBLP:journals/jco/MoserNS12}
Hannes Moser, Rolf Niedermeier, and Manuel Sorge.
\newblock Exact combinatorial algorithms and experiments for finding maximum
  k-plexes.
\newblock {\em J. Comb. Optim.}, 24(3):347--373, 2012.

\bibitem{nevsetvril2008grad}
Jaroslav Ne{\v{s}}et{\v{r}}il and Patrice~Ossona De~Mendez.
\newblock Grad and classes with bounded expansion i. decompositions.
\newblock {\em European Journal of Combinatorics}, 29(3):760--776, 2008.

\bibitem{nevsetvril2008grad2}
Jaroslav Ne{\v{s}}et{\v{r}}il and Patrice~Ossona De~Mendez.
\newblock Grad and classes with bounded expansion ii. algorithmic aspects.
\newblock {\em European Journal of Combinatorics}, 29(3):777--791, 2008.

\bibitem{nevsetvril2011nowhere}
Jaroslav Ne{\v{s}}et{\v{r}}il and Patrice~Ossona De~Mendez.
\newblock On nowhere dense graphs.
\newblock {\em European Journal of Combinatorics}, 32(4):600--617, 2011.

\bibitem{nevsetvril2012sparsity}
Jaroslav Ne{\v{s}}et{\v{r}}il and Patrice~Ossona De~Mendez.
\newblock {\em Sparsity: graphs, structures, and algorithms}, volume~28.
\newblock Springer Science \& Business Media, 2012.

\bibitem{nevsetvril2012characterisations}
Jaroslav Ne{\v{s}}et{\v{r}}il, Patrice~Ossona de~Mendez, and David~R Wood.
\newblock Characterisations and examples of graph classes with bounded
  expansion.
\newblock {\em European Journal of Combinatorics}, 33(3):350--373, 2012.

\bibitem{NESETRIL2008777}
Jaroslav Nešetřil and Patrice {Ossona de Mendez}.
\newblock Grad and classes with bounded expansion ii. algorithmic aspects.
\newblock {\em European Journal of Combinatorics}, 29(3):777--791, 2008.
\newblock URL:
  \url{https://www.sciencedirect.com/science/article/pii/S0195669807000571},
  \href {https://doi.org/https://doi.org/10.1016/j.ejc.2006.07.014}
  {\path{doi:https://doi.org/10.1016/j.ejc.2006.07.014}}.

\bibitem{niedermeier1999upper}
Rolf Niedermeier and Peter Rossmanith.
\newblock Upper bounds for vertex cover further improved.
\newblock In {\em STACS 99: 16th Annual Symposium on Theoretical Aspects of
  Computer Science Trier, Germany, March 4--6, 1999 Proceedings 16}, pages
  561--570. Springer, 1999.

\bibitem{DBLP:journals/dam/NishimuraRT05}
Naomi Nishimura, Prabhakar Ragde, and Dimitrios~M. Thilikos.
\newblock Fast fixed-parameter tractable algorithms for nontrivial
  generalizations of vertex cover.
\newblock {\em Discret. Appl. Math.}, 152(1-3):229--245, 2005.

\bibitem{DBLP:journals/dam/OrlovichDFGW11}
Yury~L. Orlovich, Alexandre Dolgui, Gerd Finke, Valery~S. Gordon, and Frank
  Werner.
\newblock The complexity of dissociation set problems in graphs.
\newblock {\em Discret. Appl. Math.}, 159(13):1352--1366, 2011.

\bibitem{pilipczuk2011problems}
Micha{\l} Pilipczuk.
\newblock Problems parameterized by treewidth tractable in single exponential
  time: A logical approach.
\newblock In {\em International Symposium on Mathematical Foundations of
  Computer Science}, pages 520--531. Springer, 2011.

\bibitem{reidl2019characterising}
Felix Reidl, Fernando~S{\'a}nchez Villaamil, and Konstantinos Stavropoulos.
\newblock Characterising bounded expansion by neighbourhood complexity.
\newblock {\em European Journal of Combinatorics}, 75:152--168, 2019.

\bibitem{romero2021treewidth}
Miguel Romero, Marcin Wrochna, and Stanislav {\v{Z}}ivn{\`y}.
\newblock Treewidth-pliability and ptas for max-csps.
\newblock In {\em Proceedings of the 2021 ACM-SIAM Symposium on Discrete
  Algorithms (SODA)}, pages 473--483. SIAM, 2021.

\bibitem{sau2021hitting}
Ignasi Sau and U{\'e}verton dos Santos~Souza.
\newblock Hitting forbidden induced subgraphs on bounded treewidth graphs.
\newblock {\em Information and Computation}, 281:104812, 2021.

\bibitem{DBLP:journals/jcss/ShachnaiZ17}
Hadas Shachnai and Meirav Zehavi.
\newblock A multivariate framework for weighted {FPT} algorithms.
\newblock {\em J. Comput. Syst. Sci.}, 89:157--189, 2017.

\bibitem{SOLEIMANFALLAH2011892}
Arezou Soleimanfallah and Anders Yeo.
\newblock A kernel of order $2k-c$ for vertex cover.
\newblock {\em Discrete Mathematics}, 311(10):892--895, 2011.

\bibitem{stege1999improved}
Ulrike Stege and Michael~Ralph Fellows.
\newblock An improved fixed parameter tractable algorithm for vertex cover.
\newblock {\em Technical report/Departement Informatik, ETH Z{\"u}rich}, 318,
  1999.

\bibitem{ebelpreprint}
Jirı~Sgall Tomas~Ebenlendr, Petr~Kolman.
\newblock An approximation algorithm for bounded degree deletion.
\newblock In {\em Preprint}, 2009.

\bibitem{DBLP:journals/corr/abs-1906-10523}
Dekel Tsur.
\newblock l-path vertex cover is easier than l-hitting set for small l.
\newblock {\em CoRR}, abs/1906.10523, 2019.
\newblock URL: \url{http://arxiv.org/abs/1906.10523}.

\bibitem{tsur2021faster}
Dekel Tsur.
\newblock Faster parameterized algorithm for cluster vertex deletion.
\newblock {\em Theory of Computing Systems}, 65(2):323--343, 2021.

\bibitem{DBLP:journals/dam/Tsur21}
Dekel Tsur.
\newblock An o{\({_\ast}\)}(2.619k) algorithm for 4-path vertex cover.
\newblock {\em Discret. Appl. Math.}, 291:1--14, 2021.

\bibitem{tu2022survey}
Jianhua Tu.
\newblock A survey on the k-path vertex cover problem.
\newblock {\em Axioms}, 11(5):191, 2022.

\bibitem{DBLP:journals/dam/TuJ16}
Jianhua Tu and Zemin Jin.
\newblock An {FPT} algorithm for the vertex cover p\({}_{\mbox{4}}\) problem.
\newblock {\em Discret. Appl. Math.}, 200:186--190, 2016.

\bibitem{DBLP:journals/jco/TuWYC17}
Jianhua Tu, Lidong Wu, Jing Yuan, and Lei Cui.
\newblock On the vertex cover p\({}_{\mbox{3}}\) problem parameterized by
  treewidth.
\newblock {\em J. Comb. Optim.}, 34(2):414--425, 2017.

\bibitem{tu2011factor}
Jianhua Tu and Wenli Zhou.
\newblock A factor 2 approximation algorithm for the vertex cover p3 problem.
\newblock {\em Information Processing Letters}, 111(14):683--686, 2011.

\bibitem{Leeuwen06}
Erik~Jan van Leeuwen.
\newblock Better approximation schemes for disk graphs.
\newblock In Lars Arge and Rusins Freivalds, editors, {\em Algorithm Theory -
  {SWAT} 2006, 10th ScandinavianWorkshop on Algorithm Theory, Riga, Latvia,
  July 6-8, 2006, Proceedings}, volume 4059 of {\em Lecture Notes in Computer
  Science}, pages 316--327. Springer, 2006.
\newblock \href {https://doi.org/10.1007/11785293\_30}
  {\path{doi:10.1007/11785293\_30}}.

\bibitem{williamson2011design}
David~P Williamson and David~B Shmoys.
\newblock {\em The design of approximation algorithms}.
\newblock Cambridge university press, 2011.

\bibitem{DBLP:journals/jcss/Xiao17a}
Mingyu Xiao.
\newblock Linear kernels for separating a graph into components of bounded
  size.
\newblock {\em J. Comput. Syst. Sci.}, 88:260--270, 2017.

\bibitem{DBLP:journals/jcss/Xiao17}
Mingyu Xiao.
\newblock On a generalization of nemhauser and trotter's local optimization
  theorem.
\newblock {\em J. Comput. Syst. Sci.}, 84:97--106, 2017.

\bibitem{DBLP:conf/tamc/XiaoK17}
Mingyu Xiao and Shaowei Kou.
\newblock Kernelization and parameterized algorithms for 3-path vertex cover.
\newblock In {\em Theory and Applications of Models of Computation - 14th
  Annual Conference, {TAMC} 2017, Bern, Switzerland, April 20-22, 2017,
  Proceedings}, pages 654--668, 2017.

\bibitem{xiao2017exact}
Mingyu Xiao and Hiroshi Nagamochi.
\newblock Exact algorithms for maximum independent set.
\newblock {\em Information and Computation}, 255:126--146, 2017.

\bibitem{you2017approximate}
Jie You, Jianxin Wang, and Yixin Cao.
\newblock Approximate association via dissociation.
\newblock {\em Discrete Applied Mathematics}, 219:202--209, 2017.

\bibitem{zhang2020improved}
An~Zhang, Yong Chen, Zhi-Zhong Chen, and Guohui Lin.
\newblock Improved approximation algorithms for path vertex covers in regular
  graphs.
\newblock {\em Algorithmica}, 82(10):3041--3064, 2020.

\bibitem{zhang2017ptas}
Zhao Zhang, Xiaoting Li, Yishuo Shi, Hongmei Nie, and Yuqing Zhu.
\newblock Ptas for minimum k-path vertex cover in ball graph.
\newblock {\em Information Processing Letters}, 119:9--13, 2017.

\bibitem{DBLP:journals/ipl/ZhangLSNZ17}
Zhao Zhang, Xiaoting Li, Yishuo Shi, Hongmei Nie, and Yuqing Zhu.
\newblock {PTAS} for minimum k-path vertex cover in ball graph.
\newblock {\em Inf. Process. Lett.}, 119:9--13, 2017.

\bibitem{zhang2014minimum}
Zhao Zhang, Weili Wu, Lidan Fan, and Ding-Zhu Du.
\newblock Minimum vertex cover in ball graphs through local search.
\newblock {\em Journal of Global Optimization}, 59(2):663--671, 2014.

\bibitem{zhu2009colouring}
Xuding Zhu.
\newblock Colouring graphs with bounded generalized colouring number.
\newblock {\em Discrete Mathematics}, 309(18):5562--5568, 2009.

\end{thebibliography}

\appendix

\section{Other Related Works}\label{sec-related}

We first remark that constant-approximation algorithms are known for various problems on bounded-expansion graphs, including various versions of \textsc{Independent Set} and \textsc{Dominating Set}~\cite{dvovrak2013constant} as well as all monotone maximization problems expressible in first-order logic~\cite{dvovrak2021approximation}. Next, we consider the generic \textsc{(Induced) Subgraph Hitting} problem and some (though not all) of its special cases that were studied independently in the literature. A thorough survey is beyond the scope of our work. Here, our goal is only to mention some of the known results. Further, we do not survey results concerning \textsc{Subgraph Isomorphism} since the literature about it is as rich as the literature on {\sc Subgraph Hitting}. Still, since it is highly relevant to our work, we note that on graphs of bounded expansion, testing first-order properties can be done in linear time~\cite{dvovrak2013testing}, which implies that \textsc{Subgraph Isomorphism} can be solved in linear time for any fixed $k$ (the size of the pattern).

The generic \textsc{(Induced) Subgraph Hitting} problem is a special case of (or, more precisely, trivially reducible to) {\sc $d$-Hitting Set}, and hence all results for {\sc $d$-Hitting Set} (e.g., approximation algorithm, FPT algorithms and kernels~\cite{williamson2011design,CyganFKLMPPS15}) directly transfer. For further discussion on the possibility of better approximation algorithms for \textsc{(Induced) Subgraph Hitting}, we refer to \cite{brustle2021approximation}.
On bounded-treewidth graphs, the problem was studied in \cite{cygan2017hitting,sau2021hitting} from the perspective of parameterized complexity. We remind that, on planar graphs, the classical work of Baker~\cite{Baker94} yields a PTAS, which can be extended to minor-free graphs~\cite{dawar2006approximation,dvovrak2020baker} and unit-disk graphs~\cite{FominLS18}. 
We remark that the generic edge-deletion version of the problem was also considered in the literature (e.g., see \cite{enright2018deleting,gaikwad2022further}). Further, there is a substantial body of works on special cases of the edge-deletion version dedicated to hitting specific graphs on only few vertices such as claws, paws and diamonds, particularly from the perspective of kernelization (see, e.g., \cite{li2022improved,cygan2017polynomial,DBLP:conf/iwpec/EibenL020,cai2015incompressibility,marx2022incompressibility} for a few examples).

\bigskip
\noindent\textsc{Vertex Cover}: The {\sc Vertex Cover} problem is widely considered to be the most extensively studied problem in parameterized complexity, and is often considered a testbed for the introduction of new notions and techniques~\cite{CyganFKLMPPS15}. It is long known that this problem can be approximated within ratio $2$~\cite{williamson2011design}, while ratio $(2-\epsilon)$ is impossible under the Unique Games Conjecture (UGC)~\cite{khot2008vertex}. More involved techniques yield slightly better ratios~\cite{bar1985local,monien1985ramsey,halperin2002improved,karakostas2009better}; for example, consider the ratio of $2-\Theta(\frac{1}{\sqrt{\log n}})$ by~\cite{karakostas2009better}. Under the assumption P$\neq$NP, the known lower bound is weaker~\cite{dinur2005hardness}. Constants better than $2$ are known for various graph classes, such as bounded degree graphs, where we know of matching upper~\cite{halperin2002improved} and lower~\cite{austrin2011inapproximability} bounds (under the UGC) for the approximation ratio, dense graphs~\cite{karpinski1996approximating}, and rectangle graphs~\cite{bar2011minimum}. Further, {\sc Vertex Cover} was already known to admit a PTAS on bounded expansion graphs~\cite{har2017approximation} and an EPTAS on disk graphs~\cite{ErlebachJS05,Leeuwen06,lokshtanov2023framework}.

From the perspective of parameterized complexity, there has been long races to achieve the best possible running times on general graphs and on bounded degree graphs, as well as when considering exact exponential-time algorithms (see~\cite{balasubramanian1998improved,
buss1993nondeterminism,chen2001vertex,
chandran2004refined,chen2001vertex,niedermeier1999upper,
DBLP:journals/tcs/ChenKX10,stege1999improved,xiao2017exact} for a few of the papers); we only briefly mention that the best known parameterized algorithm for general graphs achieves the bound $1.28^k\cdot n^{O(1)}$~\cite{DBLP:journals/tcs/ChenKX10}.  The problem is long known to admit a kernel of size $O(k^2)$~\cite{buss1993nondeterminism}, which is essentially optimal unless the polynomial hierarchy collapses~\cite{dell2014satisfiability}. Nevertheless, the number of vertices in the kernel can be reduced to $2k$~\cite{chen2001vertex} with slight improvements given in~\cite{SOLEIMANFALLAH2011892,DBLP:journals/dam/ChlebikC08}, but it is unlikely that it can be reduced to $(2-\epsilon)k$~\cite{DBLP:conf/stacs/ChenFKX05}. Further, Bidimesnionality theory for minor-free and unit-disk graphs~\cite{demaine2005subexponential,DBLP:conf/soda/FominLS12,FominLPSZ19} as well as map graphs~\cite{DBLP:conf/icalp/FominLP0Z19} yields subexponential FPT algorithms.

\bigskip
\noindent\textsc{(Induced) $P_k$-Hitting:} Clearly, {\sc Vertex Cover} is the special case of {\sc $P_k$-Hitting} (also known as {\sc $k$-Path Vertex Cover}) when $k=2$. Still, even for small values of $k$, it is harder. For example, for $k=3$, it is APX-complete on bipartite graphs~\cite{kumar2014approximation} (whereas {\sc Vertex Cover} is solvable in polynomial time), and for $k=4$, it is APX-complete on cubic bipartite graphs~\cite{devi2015computational}. Still, for $k=3$, it is solvable in polynomial-time on various graph classes~\cite{DBLP:journals/arscom/BoliacCL04,DBLP:journals/dam/BrauseK17,DBLP:journals/amc/BresarKKS19,DBLP:journals/mp/CameronH06,DBLP:journals/ipl/LozinR03,DBLP:journals/dam/OrlovichDFGW11} such as $P_5$-free graphs~\cite{DBLP:journals/dam/BrauseK17}. 
Lee~\cite{lee2017partitioning} gave an FPT $O(\log k)$-approximation algorithm where the parameter is $k$. His result was generalized and improved upon (in terms of running time) by Gupta et al.~\cite{gupta2019losing}. For $k=3$, a $2$-approximation algorithm was given in \cite{tu2011factor} (improving upon \cite{kardovs2011computing}).
Recently, for $k\geq 3$, Br{\"u}stle~\cite{brustle2021approximation} presented a ($k-\frac{1}{2}$)-approximation algorithm (even for a more general case of hitting trees). Other approximation algorithms for this problem can be found in \cite{2018moderately,zhang2020improved}. We note that the problem admits an EPTAS on disk graphs~\cite{lokshtanov2023framework} for $3\leq k\leq 5$, and a PTAS on ball graphs for any fixed $k$~\cite{DBLP:journals/ipl/ZhangLSNZ17}. 
Exact exponential-time algorithms for $k=3$ were developed in \cite{kardovs2011computing,2018moderately,chang20141,xiao2017exact}. Also, for various fixed values of $k\geq3$, the problem was intensively studied in parameterized complexity: parameterized algorithms were developed in \cite{DBLP:conf/mfcs/CervenyS19,DBLP:journals/corr/abs-2111-05896,DBLP:journals/corr/abs-1906-10523,DBLP:journals/dam/Tsur21,DBLP:journals/dam/TuJ16,DBLP:journals/corr/BaiTS16,DBLP:journals/jcss/ShachnaiZ17,DBLP:journals/jco/TuWYC17} and kernels were developed in \cite{DBLP:conf/tamc/XiaoK17,DBLP:journals/corr/abs-2107-12245,DBLP:journals/jcss/FellowsGMN11,DBLP:journals/jcss/Xiao17,DBLP:conf/tamc/XiaoK17}.
Further, the  Bidimesnionality theory for minor-free and unit-disk graphs~\cite{demaine2005subexponential,DBLP:conf/soda/FominLS12,FominLPSZ19} as well as map graphs~\cite{DBLP:conf/icalp/FominLP0Z19} yields subexponential FPT algorithms. 
For more information on {\sc $P_k$-Hitting}, we refer to the survey of~\cite{tu2022survey}.

It is known that {\sc Induced $P_3$-Hitting} is equivalent to the {\sc Cluster Vertex Deletion} problem (where the objective is to delete a minimum number of vertices to attain a {\em cluster graph}, which is a collection of cliques), which has been extensively studied in the literature. As such, it trivially has a $3$-approximation (polynomial-time) algorithm. The first non-trivial approximation algorithm  was a $\frac{5}{3}$-
approximation due to You et al.~\cite{you2017approximate}. Shortly afterward, Fiorini et al.~gave a $\frac{7}{3}$-approximation~\cite{fiorini2016improved}, and subsequently a $\frac{9}{4}$-approximation~\cite{fiorini2020improved}. Afterwards, a $2$-approximation was found~\cite{aprile2021tight}, which is easily seen to be optimal under the UGC (by a simple reduction from {\sc Vertex Cover}). The problem also received attention from the perspectives of FPT algorithms~\cite{boral2016fast,huffner2010fixed,tsur2021faster} and kernelization~\cite{DBLP:journals/talg/FominLLSTZ19}.  Further, it was studied on special graph classes such as $H$-free graphs~\cite{le2022complexity}.

\bigskip
\noindent\textsc{Degree Modulator:}  Notice that {\sc Degree Modulator}, where the goal is to delete a minimum number of vertices to reduce the maximum degree of the graph to be $k$, is equivalent to hitting stars of size $k+1$. When $k=0$, we derive {\sc Vertex Cover}, and when $k=1$, we derive {\sc $P_2$-Hitting}, both discussed earlier. So, here, we only consider the case where $k\geq 2$. First, we remark that the FPT $O(\log k)$-approximation algorithm by Gupta et al.~\cite{gupta2019losing} is also noted to apply to this problem. However, previously, a randomized (polynomial-time) $O(\log k)$-approximation algorithm was already known~\cite{ebelpreprint}. A matching inapproximability result was given in~\cite{DBLP:conf/esa/DemaineGKLLSVP19}. The problem is also known to admit an EPTAS on disk graphs~\cite{lokshtanov2023framework}.
Moreover, FPT algorithms, kernels and W[1]-hardness results for this problem (sometimes referred to by a different name) parameterized by $k$ plus (or only by) the solution size or a structural measure such as treewidth of the graph were given in~\cite{betzler2012bounded,ganian2021structural,DBLP:journals/dam/NishimuraRT05,DBLP:journals/jcss/FellowsGMN11,DBLP:journals/jco/MoserNS12,DBLP:conf/isaac/DessmarkJL93,DBLP:conf/soda/LokshtanovMRSZ21}.  We note that the complexity of the problem of hitting induced stars on $k+1$ vertices, also known as induced $k$-claws, was also previously studied (see, e.g.,~\cite{hsieh2021d}).

\bigskip
\noindent\textsc{(Induced) $C_k$-Hitting}: 
First, note that the \textsc{(Induced) $C_k$-Hitting} problem only makes sense for $k\geq 3$. The restriction of the case where $k=3$ (i.e., when we aim to hit triangles) to disk graphs admits a subexponential FPT algorithm (when parameterized by solution size)~\cite{lokshtanov2022subexponential} as well as an EPTAS~\cite{lokshtanov2023framework}. Further, the parameterized complexity of {\sc (Induced) $C_k$-Hitting} with respect to treewidth was studied in \cite{pilipczuk2011problems}, and a dichotomy concerning its membership in P was given in 
\cite{groshaus2009cycle}.

\bigskip
\noindent\textsc{$K_k$-Hitting} and \textsc{(Induced) Biclique Hitting:} Approximation and parameterized algorithms for \textsc{$K_k$-Hitting} on perfect graphs were given in~\cite{fiorini2018approximability}. Additionally, the parameterized complexity of \textsc{(Induced) Biclique Hitting} was studied in~\cite{goldmann2021parameterized}.
 
\bigskip
\noindent\textsc{Component Order Connectivity:} The objective of this problem is delete a minimum number of vertices in order to hit all connected graphs of size at most $k$. Clearly, {\sc Vertex Cover} is its special case when $k=2$. For \textsc{Component Order Connectivity}, Bidimesnionality theory for minor-free and unit-disk graphs~\cite{demaine2005subexponential,DBLP:conf/soda/FominLS12,FominLPSZ19} as well as map graphs~\cite{DBLP:conf/icalp/FominLP0Z19} yields subexponential FPT algorithms. Drange et al.~\cite{drange2016computational} provided an FPT and a kernel for~\textsc{Component Order Connectivity} parameterized by $k$ plus the solution size, and, in  addition, studied the complexity of the problem on special graph classes. Iproved $O(ks)$-vertex kernels were given by 
\cite{kumar20172lk,DBLP:journals/jcss/Xiao17a}, where $s$ is the solution size. Moreover, an FPT approximation scheme with respect to treewidth was given  in \cite{DBLP:conf/soda/LokshtanovMRSZ21}.
Additional results can be found in the survey~\cite{gross2013survey}.

\bigskip
\noindent\textsc{Treedepth Modulator:} The objective of this problem is to delete a minimum-sized $t$-treedepth modulator in the given graph, that is, a minimum number of vertices whose deletion from the given graph yields a graph of treedepth at most $t$. Here,  $t\in\mathbb{N}$ is a given fixed constant. The problem can be phrased as a subgraph hitting problem where all subgraphs are of size bounded by a function of $t$. The computation and utility of treedepth modulators were studied from the perspective of parameterized complexity (see, e.g., \cite{bougeret2019much,eiben2022lossy,gajarsky2017kernelization}).

\end{document}